\documentclass[11pt]{article}

    \usepackage{multicol}
    \usepackage[english]{babel}
    \usepackage[utf8]{inputenc}    

    \usepackage{amsmath, amssymb,amsfonts,xspace,graphicx,relsize,mathtools,xcolor}
    \usepackage{mathrsfs}
    \usepackage[pagebackref,final]{hyperref}
        \hypersetup{
    	colorlinks,
    	linkcolor={blue!100!black},
    	citecolor={blue!100!black},
    }
    \usepackage{enumerate}
    \usepackage{bm}
    \usepackage{appendix}
    \usepackage{amsthm}
    \usepackage{fullpage}
    \usepackage{ifdraft}
    \usepackage{verbatim}
    \usepackage{multirow}

\usepackage{mleftright}\mleftright 
\usepackage{bm}
\usepackage{aliascnt} 
\usepackage{bookmark} 
\usepackage{caption,float,subcaption}
\usepackage{tikz,pgfplots}\pgfplotsset{compat=1.16}
\usepackage{algorithm,algcompatible}

\ifdraft{
 \newcommand{\authnote}[3]{{\color{#3} {\bf  #1:} #2}}
 }{
 \newcommand{\authnote}[3]{}
 }

\newcommand{\pnote}[1]{\authnote{ András}{#1}{cyan}}

\newcommand{\ci}{\mathbf{i}}
\newcommand{\vecone}{\vec{1}}
\renewcommand{\pmb}[1]{{#1}}

\addto\extrasenglish{%
  \def\equationautorefname~#1\null{Equation~(#1)\null}
}
\newtheorem{theorem}{Theorem}

\newcommand{\addalias}[3][theorem]{
  \newaliascnt{#2}{#1}
  \newtheorem{#2}[#2]{#3}
  \aliascntresetthe{#2}
  \expandafter\def\csname #2autorefname\endcsname{#3}
}

    \addalias{definition}{Definition}
    \addalias{lemma}{Lemma}
    \addalias{proposition}{Proposition}
    \addalias{corollary}{Corollary}

\usepackage{ifthen}
\newcommand{\refif}[1]{%
  \ifthenelse { \equal {#1} {} }%
              {}%
              {, \ref{#1}}}

\def\reptitletemp{}
\newtheorem*{reptheorembase}{\reptitletemp}
\newenvironment{reptheorem}[3]{\def\reptitletemp{Theorem \ref{#1}\refif{#2}\refif{#3}}\begin{reptheorembase}}{\end{reptheorembase}}%

    \def\01{\{0,1\}}
    \newcommand{\ceil}[1]{\left\lceil{#1}\right\rceil}
    
    \newcommand{\nrm}[1]{\left\| #1 \right\|}
    \newcommand{\trm}[1]{\left({#1}\right)}
    
    \newcommand{\Tr}{\mbox{\rm Tr}}
	\newcommand{\tr}[1]{\Tr\left(#1\right)}
	\newcommand{\ipc}[2]{\left\langle #1 , #2 \right\rangle}
	\newcommand{\bigO}[1]{\mathcal{O}\left( #1 \right)}
	\newcommand{\bigOt}[1]{\widetilde{\mathcal{O}}\left( #1 \right)}	

    \newcommand{\diag}{\mbox{\rm diag}}
    \newcommand{\polylog}{\mbox{\rm polylog}}
    \newcommand{\amp}{\ensuremath{\textrm{amp}}}

    \newcommand{\Prob}[1]{\mathbb{P}\left(#1\right)}

    \renewcommand{\P}{\ensuremath{\mathbb{P}}}
    \newcommand{\R}{\ensuremath{\mathbb{R}}}
    \newcommand{\C}{\ensuremath{\mathbb{C}}}
    \newcommand{\Z}{\ensuremath{\mathbb{Z}}}
    \newcommand{\N}{\ensuremath{\mathbb{N}}}

    \newcommand{\eps}{\varepsilon}

	\DeclareMathOperator{\sinc}{sinc}    
	
	\newcommand{\vertiii}[1]{{\left\vert\kern-0.25ex\left\vert\kern-0.25ex\left\vert #1 
		\right\vert\kern-0.25ex\right\vert\kern-0.25ex\right\vert}}
    
    \newcommand{\ket}[1]{|#1\rangle}
    \newcommand{\bra}[1]{\langle#1|}
    \DeclarePairedDelimiterX\ketbra[2]{| }{|}{#1 \delimsize\rangle\!\delimsize\langle #2}

\DeclarePairedDelimiterX\braket[2]{\langle}{\rangle}{#1 \delimsize\vert #2}
\DeclarePairedDelimiterX\dotp[2]{\langle}{\rangle}{#1, #2}

\title{Quantum tomography using state-preparation unitaries}

\usepackage{authblk}

\author[1,2]{Joran van Apeldoorn}
\author[1]{Arjan Cornelissen} 
\author[3]{Andr\'as Gily\'en}
\author[4]{Giacomo Nannicini}

\affil[1]{QuSoft, UvA, Amsterdam, the Netherlands. }
\affil[2]{IViR, UvA, Amsterdam, the Netherlands. }
\affil[3]{Alfr\'ed R\'enyi Institute of Mathematics, Budapest, Hungary.}
\affil[4]{IBM Quantum, IBM T.J.~Watson research center, Yorktown Heights, NY, USA.}

\begin{document}

\maketitle

\begin{abstract}
We describe algorithms to obtain an approximate classical description of a $d$-dimensional quantum state when given access to a unitary (and its inverse) that prepares it. For pure states we characterize the query complexity for $\ell_q$-norm error up to logarithmic factors. As a special case, we show that it takes $\widetilde{\Theta}(d/\varepsilon)$ applications of the unitaries to obtain an $\varepsilon$-$\ell_2$-approximation of the state.

For mixed states we consider a similar model, where the unitary prepares a purification of the state. In this model we give an efficient algorithm for obtaining Schatten $q$-norm estimates of a rank-$r$ mixed state, giving query upper bounds that are close to optimal. In particular, we show that a trace-norm ($q=1$) estimate can be obtained with $\widetilde{\mathcal{O}}(dr/\varepsilon)$ queries. This improves (assuming our stronger input model) the $\varepsilon$-dependence over the algorithm of Haah et al. (2007) that uses a joint measurement on $\widetilde{\mathcal{O}}(dr/\varepsilon^2)$ copies of the state.

To our knowledge, the most sample efficient results for pure state tomography come from setting the rank to $1$ in generic mixed state tomography algorithms, which can be computationally demanding. We describe sample-optimal algorithms for pure states that are easy and fast to implement.

Along the way we show that an $\ell_\infty$-norm estimate of a normalized vector induces a (slightly worse) $\ell_q$-norm estimate for that vector, without losing a dimension-dependent factor in the precision. We also develop an unbiased and symmetric version of phase estimation, where the probability distribution of the estimate is centered around the true value. Finally, we give an efficient method for estimating multiple expectation values, improving over the recent result by Huggins et al.~(2021) when the measurement operators do not fully overlap. 
More specifically, we show that for $E_1,\dots,E_m$ normalized measurement operators, all expectation values $\Tr(E_j\rho)$ can be efficiently learned up to error $\varepsilon$ with $\widetilde{\mathcal{O}}(\sqrt{\|\sum_j E_j^2\|}/\varepsilon)$ applications of a state-preparation unitary for a purification of $\rho$. 	
\end{abstract}

 \clearpage

 \section{Introduction}

 Quantum state tomography is the process of obtaining a classical description of a quantum state. Tomography is a fundamental tool in quantum information science, where it finds numerous applications. In the context of quantum algorithms, pure quantum state tomography can be used to retrieve a classical description of the final state of the algorithm, e.g., the solution of a linear system \cite{harrow2009QLinSysSolver} or the evolution of a quantum system \cite{lloyd1996UnivQSim}. 
 The more general mixed quantum state tomography finds applications in quantum information theory, and in the simulation of quantum thermodynamic systems. In some settings we are not interested in the full state, but only in its expectation value under a certain set of (possibly overlapping) measurements. This was first introduced by Aaronson~\cite{aaronson2017shaddow} under the name shadow tomography, and has since received a lot of attention in the literature, e.g.,~\cite{Huang2020shaddow,acharya2021informationally,hu2022logical}.

 Most of the existing work on this topic has focused on the sample complexity of these problems: how many copies of the state are needed to perform tomography? In this paper we consider the problem under a different input model: we assume access to a unitary (and its inverse) that prepares the state. This model is very natural when the state is the output of a quantum algorithm, but it has received little attention so far. The main improvements in this model come from the ability to use techniques related to amplitude estimation to reduce the dependence on the error parameter, but attaining such quadratic improvements requires the development of several new tools, and the analysis does not follow from a simple application of amplitude estimation.

 Throughout the paper we consider either a $d$-dimensional pure state $\ket{\psi} = \sum_{j=0}^{d-1} \alpha_j \ket{j}$ or a rank-$r$ mixed state $\rho\in\C^{d\times d}$. We are interested in learning the state up to error $\eps$ in some $\ell_q$-norm or Schatten $q$-norm, often with some probability of failure $\leq \delta$. In the introduction we often use $\bigOt{\cdots}$ notation to hide polylogarithmic factors in the parameters $d$, $r$, $1/\eps$, and $1/\Delta$, even if these parameters do not appear polynomially in the $\bigOt{\cdots}$. For more precise complexity statements we refer to the relevant theorems in the main text.

\paragraph{Related work.} Classical algorithms that estimate probabilities generally depend quadratically on $1/\eps$, as that many samples are required to bring down the variance. In certain settings quantum algorithms can improve on this classical complexity. Brassard et al.~\cite{brassard2002AmpAndEst} introduced the amplitude estimation algorithm, and showed that it can estimate an amplitude (or probability) with a $1/\eps$ dependence, if a state-preparation unitary and its inverse are available.

Van~Apeldoorn~\cite{apeldoorn2021QProbOraclesMulitDimAmpEst} generalized this for finding an $\ell_\infty$-norm estimate of a discrete probability distribution. In the model of van Apeldoorn, access to the distribution is given by a state-preparation oracle (and its inverse), such that the probability distribution corresponds to computational-basis measurements of the prepared state. Van Apeldoorn~\cite{apeldoorn2021QProbOraclesMulitDimAmpEst} showed that $\tilde{O}(1/\eps)$ applications of the input unitary are sufficient to compute the desired $\ell_\infty$-norm estimate. In the same paper the question was posed whether you can also speed-up the estimation of multiple expectation values over the same distribution. A lower bound of $\Omega(\min\{\sqrt{m}/\eps,1/\eps^2\})$ was given when $m$ expectation values need to be estimated op to precision $\eps$. It was later shown by Huggins et al.~\cite{huggins2021QAlgMultipleExpectationValues} that $\bigOt{\sqrt{m}/\eps}$ queries are sufficient even when estimating expectation values of observables on a pure quantum state.

Kerenidis and Prakash~\cite{kerenidis2018QIntPoint} gave a sampling-based approach for estimating the real-valued amplitudes resulting from a quantum linear system solver, including their sign, taking $\tilde{O}(d/\eps^2)$ applications of a (controlled) state-preparation unitary to compute an $\ell_2$-norm estimate. We subsume their approach, and show that besides estimating real-valued amplitudes, one can even estimate complex amplitudes with the same sample complexity.

Besides these few results for pure quantum state tomography, the most frequently studied setting is that of mixed-state tomography. In this setting we want to determine how many copies are necessary to obtain a classical description with a given maximum error $\eps$ in trace-norm; it is often assumed that some upper bound $r$ on the rank of the state is known (if the state is pure, $r = 1$). An algorithm of Gross et al.~\cite{gross2010}, that applies measurements on one copy of the state at once, achieves $O(dr^2/\eps^2)$ sample complexity. Haah et al.~\cite{haah2017OptTomography} show that the bound is optimal when the measurements are on a single copy at a time, and Chen et al.~\cite{chen2022tight} complete our understanding of this setting by showing that the bound cannot be improved even with adaptive measurements schemes, as long as we require single-copy measurements. A better sample complexity can be achieved if we allow joint measurements on multiple copies of the state: with this more powerful access model, the best algorithm for tomography is also due to Haah et al.~\cite{haah2017OptTomography}, and it requires $\tilde{O}(dr/\eps^2)$ copies of the quantum state; see also \cite{odonnell2016EfficientQuantumTomography}. Haah et al.~also show matching lower bounds up to polylogarithmic factors (these polylogarithmic factors are eliminated by Yuen~\cite{yuen2022improved}), therefore their algorithm is essentially optimal. The main drawbacks of their approach are that it not only requires joint measurements on many states at once, but it also has time complexity exponential in $d$. 

\paragraph{Our results.}
We then start our discussion on quantum state tomography for pure states. Our analysis can be divided into two settings: the sampling-based setting, in which copies of the state are available, and the state-preparing unitary setting, in which we require controlled access to a state-preparation unitary and its inverse. 

To give optimal algorithms for other $\ell_q$-norms as well, we prove a norm-conversion lemma relating estimates in different $\ell_q$-norms. The standard approach for norm conversion is to decrease the allowed error $\eps$ by a factor $d^{1/q}$, but this introduces a dependence on the dimension that can be suboptimal. We show that a dimension-independent norm conversion is possible for normalized vectors, and therefore for quantum states. We also relate estimates of the amplitudes to estimates of the corresponding probability distribution.

To our suprise little seems to be known about pure quantum state tomgography using samples. Some results can be obtained by setting $r=1$ for the mixed-state case, but these methods are highly impractical from a computational standpoint, or require the implementation of random measurements. 
We cover three different models with our sampling based pure-state tomography results, and for each give an easy to implement tomography algorithm:
\begin{enumerate}
    \item \emph{Classical samples}. In this model we are given classical samples from computational-basis measurements. As we cannot recover information about the phases, we aim to produce an estimate of $|\alpha|$, the vector of absolute values of the amplitudes.
    \item \emph{Copies of the state}. In this model we are given copies of the quantum state, and aim to give an estimate of $\alpha$ up to a global phase. Our algorithm does not require joint measurements on different copies, but the algorithm is adaptive in the sense that it proceeds in two phases, where the outcomes of the first phase are used to transform the state before subsequent measurements. 
    \item \emph{Conditional copies of the state}. In this model we are given copies of $(\ket{0}\ket{\psi}+\ket{1}\ket{0})/\sqrt{2}$, and aim to give an estimate of $\alpha$, including the global phase. This model is inspired by controlled usage of a state-preparation unitary (but not its inverse), as this allows us to produce such samples. 
\end{enumerate}
Our algorithms for these three models all give the same sample complexity, up to polylogarithmic factors:
\begin{reptheorem}{thm:state-sampling}{thm:condsamp}{thm:uptophase}
 (Informal) In all three sampling input models $\tilde{O}(1/\eps^2)$ samples are sufficient to obtain an $\ell_\infty$-norm estimate with error at most $\eps$. For $\ell_q$-norm error ($q \ge 2$) the sample complexity\footnote{Here, and in the rest of the paper, when working with norms we use $1/0=\infty$ and $1/\infty = 0$. If one of the terms in the $\min\{\dots\}$ goes to $\infty$ due to this, then the complexity is simply the other term.} is $\tilde{O}\trm{\min\left\{\trm{\frac{3}{\eps}}^{\frac{1}{\frac{1}{2}-\frac{1}{q}}},\frac{d^{\frac{2}{q}}}{\eps^2}\right\}}$.
\end{reptheorem}

All our sampling-based approaches require a number of samples that scales quadratically with $1/\eps$ to obtain an $\ell_\infty$-norm estimate. For sampling approaches this error dependence is optimal even when estimating a single amplitude only. However, when estimating a single amplitude with access to a state-preparation unitary and its inverse, amplitude estimation can be used to improve this dependence to linear~\cite{brassard2002AmpAndEst}. Van~Apeldoorn~\cite{apeldoorn2021QProbOraclesMulitDimAmpEst} shows that this can be generalized to estimate all probabilities in the corresponding distribution with linear dependence. Unfortunately, for amplitudes it is impossible to get an $O(1/\eps)$ error dependence that is independent of the dimension. However, for the high-precision regime there is still an improvement.

\begin{reptheorem}{thm:quantum_euclidean_norm_est}{}{}
 (Informal) Given controlled access to a state-preparation unitary for $\ket{\psi}$ and its inverse, $\tilde{O}\trm{\min\left\{\frac{\sqrt{d}}{\eps}, \frac{1}{\eps^2}\right\}}$ uses of these unitaries are sufficient to estimate the vector $\alpha$ with $\ell_\infty$-norm error at most $\eps$. For $\ell_q$-norm error ($q \ge 2$) this bound becomes $\tilde{O}\trm{ \min\left\{\trm{\frac{3}{\eps}}^{\frac{1}{\frac12 - \frac1q}}, \frac{d^{\frac12 +\frac1q}}{\eps}\right\}}$.
\end{reptheorem}

In the final section of our paper we show matching lower bounds for the above sample and query complexities. We show a $\tilde{\Omega}(d/\eps^2)$ bound for $\ell_1$-norm estimation of the probability distribution induced by a state $\ket{\psi}$ given access to copies of $\frac{\ket{0}\ket{\psi}+\ket{1}\ket{0}}{\sqrt{2}}$, using a communication complexity argument.  We also show that with access to a state-preparation unitary, this requires $\tilde \Omega\trm{\frac{d}{\eps}}$ applications of the input unitary, with a reduction from the problem of determining an unknown bit string via queries to a fractional phase oracle. Using our results on the relation between different norms (and between probability estimates and amplitude estimates), we obtain the following result.
\begin{reptheorem}{thm:lowerbound-sample}{thm:lowerboundpure}{}
 (Informal) All the upper bounds on pure-state tomography given in this paper are optimal, up to polylogarithmic factors.
\end{reptheorem}

\begin{table}[h!]
\centering
\begin{tabular}{||c | c | c||} 
 \hline
  & Sampling models  & Unitary model\\ 
 \hline\hline
 $\ell_\infty$-norm & $\frac{1}{\eps^2}$ & $\min\{\frac{1}{\eps^2},\frac{\sqrt{d}}{\eps}\}$  \\[.5ex]
 \hline
  $\ell_2$-norm & $\frac{d}{\eps^2}$  &  $\frac{d}{\eps}$ \\[.5ex]
 \hline
  $\ell_q$-norm & $\min\{\trm{\frac{3}{\eps}}^{\frac{1}{\frac12-\frac1q}},\frac{d^{\frac2q}}{\eps^2}\}$ & $ \min\{\trm{\frac{3}{\eps}}^{\frac{1}{\frac12-\frac1q}}, \frac{d^{\frac12+\frac{1}{q}}}{\eps}\} $ \\[.5ex]
 \hline
\end{tabular}
\caption{Sample and query complexities of recovering a $d$-dimensional pure quantum state up to error $\eps$ in the specified norm, for the different models. All results are $\tilde{\Theta}$., i.e., they are tight up to polylogarithmic factors in $d$, $1/\eps$ and $1/\delta$, where $\delta$ is the maximum failure probability.} 
\label{table:1}
\end{table}

We then turn to mixed quantum states of rank at most $r$. We show how to find an entry-wise $\eps$-approximation using $\bigOt{\frac{\sqrt{d}}{\eps}}$ samples, and that this yields an $\eps$-operator norm estimate if we set the entry-wise error to $\eps/\sqrt{d}$. This leads to the following result.

\begin{reptheorem}{thm:mixedTomoSch}{}{}
  (Informal) Given controlled access to a state-preparation unitary (and its inverse) for a purification of a rank-$r$ quantum state $\rho \in \C^{d\times d}$, $\tilde{O}\trm{\frac{d}{\eps}}$ uses of these unitaries are sufficient to estimate $\rho$ in operator norm. For trace norm error $\tilde{O}\trm{\frac{dr}{\eps}}$ uses suffice.
\end{reptheorem}

To obtain this result we first need two new intermediate results of independent interest: unbiased and symmetric phase estimation, and shadow tomography with state-preparation unitaries. The unbiased version of phase estimation is required for the conversion between entry-wise error and operator-norm error mentioned above (if all entry-wise errors go in the same direction then the best possible conversion would give a factor $d$, not $\sqrt{d}$). We show that phase estimation can be made unbiased and symmetric by adding a random phase before applying the inverse quantum Fourier transform, then removing this phase from the estimate.

\begin{reptheorem}{thm:boostedSupressed}{}{}
 (Informal) Quantum phase estimation can be used to give an unbiased and symmetric estimator of the phase.
\end{reptheorem}

Second, we implement a version of shadow tomography when given access to a state-preparation unitary for a purification of the state. Huggins et al.~\cite{huggins2021QAlgMultipleExpectationValues}
show that we can learn the expectation value of $m$ normalized measurement operators using $\bigOt{\sqrt{m}/\eps}$ queries to the state-preparation unitary. We improve on this for the case where the measurement operators do not fully overlap, while recovering the same bound for the general case.
\begin{reptheorem}{thm:expectationValues}{}{}
  (Informal) Let $E_1,\dots E_m$ be measurement operators with operator norm at most $1$. Given controlled access to a state-preparation unitary (and its inverse) for a purification of a quantum state $\rho \in \C^{d\times d}$, $\tilde{O}\trm{\frac{\sqrt{\nrm{\sum_j E_j^2}}}{\eps}}$ uses of these unitaries are sufficient to estimate all $\tr{\rho E_j}$ up to error $\eps$. 
\end{reptheorem}

Finally, we prove lower bounds on the estimation of a density matrix given (inverse) access to a unitary that prepares a purification of it. The lower bound proof on high level consists of three steps. First, we embed a bit string of length $rd$ into a family of density matrices. Then, we quantify how much information about the embedded bit string can be obtain by an algorithm that recovers any of these density matrices up to the specified precision. We conclude by arguing that obtaining this amount of information about the bit string requires a particular number of queries to the state-preparing unitary. Our results are tight in the small error regime, in the Frobenius norm case. \autoref{table:2} gives an overview of the other results that can be derived from it.

\begin{table}[h!t]
\centering
\begin{tabular}{||c | c | c||} 
 \hline
    & \multicolumn{2}{c||}{Unitary model} \\ \cline{2-3}
    & Upper bound & Lower bound \\
 \hline\hline
 Max-norm  & $\frac{\sqrt{d}}{\eps}$ & $\frac{1}{\eps}$ \\[.5ex]
 \hline
  Operator norm &   $\frac{d}{\eps}$ & $\frac{d}{\eps}$ \\[.5ex]
  \hline
    Frobenius norm &   $\min\{\frac{d\sqrt{r}}{\eps}, \frac{d}{\eps^2}\}$ & $\frac{d\sqrt{r}}{\eps} \quad (\eps = o(\frac{1}{dr}))$ \\[.5ex]
    \hline
      Trace norm & $\frac{dr}{\eps}$ & $\frac{d\sqrt{r}}{\eps}$ when $\eps = o(\frac{1}{dr})$ \\[.5ex]
           &  & $dr/\log(dr)$ when $\eps = \Theta(1)$ \\[.5ex]

 \hline
\end{tabular}
\caption{Our upper and lower bounds on the query complexities of recovering a $d$-dimensional, rank-$r$ mixed quantum state up to error $\eps$ in the specified norm, when given access to unitary that prepares its purification. All upper bound results are $\tilde{O}$, i.e., they are given up to polylogarithmic factors in $d$, $1/\eps$ and $1/\delta$, where $\delta$ is the maximum failure probability. The lower bound results are $\Omega$, and we observe that our results are tight for constant failure probability when the desired precision is w.r.t.\ the operator norm or Frobenius norm.}
\label{table:2}
\end{table}

  \section{Preliminaries}
  Many of the algorithms presented in this paper are built on top of the block-encoding framework, and rely on a version of Jordan's gradient algorithm \cite{jordan2005QuantGrad,gilyen2017OptQOptAlgGrad}. In this section we introduce some notation, our computational model, and give a brief overview of the two components mentioned above.
  
\subsection{Notation and computational model}
  For any integer $j$, we define $[j] := \{0,\dots,j-1\}$. 
  Let $\oplus$ denote the direct sum, i.e., $A \oplus B = \begin{pmatrix} A & 0 \\ 0 & B \end{pmatrix}$. We write $\vecone$ for the all-ones vector and $J$ for the all-ones matrix, with dimensions that will be clear from context. We write $\Delta^d$ for the set of $d$-dimensional probability distributions. We write $[a,\infty]$ for the set $[a,\infty)\cup\{\infty\}$. All logarithms are base 2 unless otherwise indicated.
  
  Given a vector $v$, we write $\nrm{v}_q$ for the standard $\ell_q$-norm. We use the convention that $1/0 = \infty$ and $1/\infty = 0$ in calculations involving the value of $q$ for an $\ell_q$-norm. Although the letter $p$ is commonly used to denote norms (i.e., $\ell_p$-norms), in this paper we use $p$ to denote vectors containing the entries of a discrete probability distribution; hence, we use different letters for norms. For a matrix $M$ we use write $\nrm{M}_q$ for the Schatten $q$-norm, i.e., the $\ell_q$-norm of the vector of singular values. For operator norm (Schatten $\infty$-norm) we just write $\nrm{M}$. We write that $\tilde{\alpha}$ is an $\eps$-$\ell_q$-norm estimate of $\alpha$ if $\nrm{\alpha - \tilde{\alpha}}_q \le \eps$. For a vector $\alpha$, we denote by $|\alpha|$ the vector with entries given by the modulus of the entries of $\alpha$.
  
  We assume that the quantum computer is controlled by a classical computer (with a RAM) that can change the gates run depending on intermediate measurement results. For simplicity, we neglect the cost of any classical computation as long as it is only a polylogarithmic factor (in all input parameters) slower than the quantum gate complexity. Our gate set consists of all single-qubit gates and CNOT. To simplify the statements of our results we assume access to a QRAM-like gate, the indexed-SWAP gate. This gate acts on a state with many qubits as follows:
  \[
   \text{indexed-SWAP} \ket{i}\ket{j}\ket{x_1}\dots\ket{x_d} = \ket{i}\ket{j}\text{SWAP}_{i,j}\left(\ket{x_1}\dots\ket{x_d}\right)
  \]
  where $\text{SWAP}_{i,j}$ swaps the $i$-th and $j$-th qubit. Such a gate can be built using $\bigO{d}$ gates, and $\log(d)$ depth, see \autoref{apd:qram} for details on this implementation. We always state the number of calls to such a gate and the size of the memory it acts on.
  
  \subsection{Block-encodings}
	We begin by listing the technical results that we need to efficiently manipulate matrices given via block-encoding circuits. For more background see~\cite{gilyen2018QSingValTransfThesis}. First we define a block-encoding as follows.
	
	\begin{definition}[Block-encoding]
		A unitary $U$ is an $a$-qubit block-encoding of $A$ if the top-left block of the unitary~$U$ is $A$:
		\[
		A= \left(\bra{0}^{\otimes a}\otimes I\right) U \left(\ket{0}^{\otimes a}\otimes I\right)\Longleftrightarrow U= \left[\begin{array}{cc} A & .
		\\ .
		& .
		\end{array}\right].
		\]
	\end{definition}
	Note that we are simplifying the block-encoding framework: traditionally, block-encodings are defined with three parameters (normalization factor, number of additional qubits, error of the implementation), but in this paper the normalization factor and error of the implementation are easily tracked without additional notation. Thus, we use a simpler presentation. Readers familiar with block-encodings can easily restate our results using the more familiar notation.
	
	Although we do not use POVMs directly, we mention the following lemma to showcase that the block-encoding framework is applicable in large generality. In particular, thanks to the following lemma, some of our results in the block-encoding framework are directly applicable to POVMs.
	\begin{lemma}\cite{apeldoorn2018ImprovedQSDPSolving}
		\label{lem:povmtoblock}
		If a two-outcome POVM denoted by $E$ can be coherently implemented on a quantum computer using $a$ ancillary qubits via the unitary $U$, then an $(a+1)$-block-encoding of $E$ can be implemented using a single call to $U$, $U^{\dagger}$, and a CNOT gate.
	\end{lemma} 

	The following two lemmas show how to add and amplify block-encodings, which we use repeatedly for mixed-state tomography.

	\begin{lemma}[Linear combination of block encodings]\label{lem:linCombBlocks}\cite{gilyen2018QSingValTransf,apeldoorn2018ImprovedQSDPSolving}
	    Let $E=\sum_{j=0}^{m}y_j E_j$ be a $w$-qubit operator for $y\in \R^m$, and let $\beta\geq \nrm{y}_1$. If $U_y$ is a state-preparation oracle for $\frac{1}{\sqrt{\beta}} \sum_j \sqrt{y_j} \ket{j}\ket{0}+\ket{\psi}\ket{1}$ for some unnormalized state $\ket{\psi}$, and $U_{E}$ implements an $a$-block-encoding of $E_j$ conditioned on $j$, then a $(a+\ceil{\log(m)}+1)$-block-encoding of $E/\beta$ can be implemented with a single use of $U_y$, $U_y^\dagger$, and $U_E$, and a single two-qubit gate.
	\end{lemma}
	
	\begin{lemma}[Uniform amplification of block-encodings, {\cite{low2017HamSimUnifAmp},\cite[Theorem 33]{gilyen2018QSingValTransfArXiv}}] 
	  \label{lem:ampBlock}
	  Let $U$ be an $a$-block-encoding of $A$, and let $\nrm{A} \leq \beta \leq 1$. Then a $(a+1)$-block-encoding of $A / (2\beta)$ can be implemented, up to operator norm error $\eps$, using $\bigO{\beta \log(\beta/\eps)}$ applications of $U$ and $U^\dagger$, and $\bigO{a \beta \log(\beta/\eps)}$ additional gates.
	\end{lemma}

	One of the main motivations for defining block-encodings is the following Hamiltonian simulation result, that we use for implementing ``phase oracles'' required for gradient computation.

	\begin{lemma}[Hamiltonian simulation using block-encodings, {\cite{low2016HamSimQubitization},\cite[Corollary 63]{gilyen2018QSingValTransfArXiv}}]
	  \label{lem:blockHamSim}
	  Let $U$ be an $a$-block-encoding of $A$. Then a $(a+2)$-blockencoding of $e^{\ci t A}$ can be implemented, up to operator norm error $\eps$, using $\bigO{t+\log(1/\eps)}$ applications of $U$ and $U^{\dagger}$, and $\bigO{a(t+\log(1/\eps))}$ additional gates with depth $\bigO{\log(a)(t+\log(1/\eps))}$.
	\end{lemma}
	
	Finally, we will use the following lemma to construct block-encodings for gradient computation:
	
	\begin{lemma}[Block-encoding inner products with controlled state-preparation unitaries]
	  \label{lem:blockInnerProd}
	  Let $U:=\sum_x U_x \otimes \ketbra{x}{x}$ and $V:=\sum_x V_x \otimes \ketbra{x}{x}$ be controlled (by the second register) state-preparation unitaries, where $U_x\colon \ket{0}\ket{0}^{\otimes a}\mapsto \ket{0}\ket{\psi_x}+\ket{1}\ket{\tilde{\psi}_x}$ and $V_x\colon \ket{0}\ket{0}^{\otimes a}\mapsto \ket{0}\ket{\phi_x}+\ket{1}\ket{\tilde{\phi}_x}$ are $(a+1)$-qubit state-preparation unitaries for some (subnormalized) $a$-qubit quantum states $\ket{\psi_x}, \ket{\phi_x}$. Then $(I_1 \otimes V^\dagger)\cdot (\text{SWAP} \otimes I) \cdot (I_1 \otimes U)$ is an $(a+2)$-block-encoding of the diagonal matrix  $\diag({\{\ipc{\phi_x}{\psi_x}\}})$, where $I_1$ acts on a single qubit and the $\text{SWAP}$ gate acts on the first two qubits. 
	\end{lemma}
	\begin{proof}
	\begin{align*}
	    \bra{0}^{\otimes a+2}\bra{x}(I \otimes V^\dagger)&\cdot (\text{SWAP} \otimes I) \cdot (I \otimes U)\ket{0}^{\otimes a+2}\ket{y} \\
	    &=\bra{0}\trm{\bra{0}\bra{\phi_x}+\bra{1}\bra{\tilde{\phi}_x}}\bra{x} \trm{\text{SWAP} \otimes I} \ket{0}\trm{\ket{0}\ket{\psi_x}+\ket{1}\ket{\tilde{\psi}_x}}\ket{y} \\
	    &=(\bra{00}\bra{\phi_x}+\bra{01}\bra{\tilde{\phi}_x})\bra{x}  (\ket{00}\ket{\psi_x}+\ket{10}\ket{\tilde{\psi}_x})\ket{y} \\
	    &=\ipc{\phi_x}{\psi_y}\delta_{xy} \qedhere
	\end{align*}
	\end{proof}

\subsection{Quantum gradient computation}
\label{subsec:gradient}

	In this section we briefly review Jordan's algorithm for estimating the gradient and provide a generic analysis of its behavior. Before describing the algorithm, we introduce appropriate representation of our qubit strings suitable for fixed-point arithmetics.
	\begin{definition}[{\cite[Definition 5.1]{gilyen2017OptQOptAlgGrad}}]\label{def:labelDefiniton}
		For every $b\in\{0,1\}^n$, let $j^{(b)}\in \{0,\ldots,2^n-1\}$ be the integer corresponding to the binary string $b=(b_1,\ldots,b_n)$.
		We label the $n$-qubit basis state $\ket{b_1}\ket{b_{2}}\cdots\ket{b_n}$ by $\ket{x^{(b)}}$, where 
		\begin{equation*}
			x^{(b)}=\frac{j^{(b)}}{2^n}-\frac{1}{2}+2^{-n-1}.
		\end{equation*}
		We denote the set of corresponding labels as $G_n:=\left\{\frac{j^{(b)}}{2^n}-\frac{1}{2}+2^{-n-1} : j^{(b)}\in \{0,\ldots,2^n-1\}  \right\}$. Note that there is a bijection between $\{j^{(b)}\}_{b\in \{0,1\}^n}$ and $\{x^{(b)}\}_{b\in \{0,1\}^n}$, so we will use $\ket{x^{(b)}}$ and $\ket{j^{(b)}}$ interchangeably.
	\end{definition}
	Following \cite[Definition 5.2]{gilyen2017OptQOptAlgGrad} for $x\in G_n$ we define the Fourier transform of a state $\ket{x}$ as
	\begin{equation*}
		QFT_{G_n}: \ket{x}\mapsto \frac{1}{\sqrt{2^n}}\sum_{k\in G_n}e^{2\pi i 2^n x k}\ket{k}.
	\end{equation*}	
    In \cite[Claim 5.1]{gilyen2017OptQOptAlgGrad} it is shown that this unitary is the same as the usual quantum Fourier transform up to conjugation with a tensor product of $n$ single-qubit unitaries.		

	Let us prove a simplified version of \cite[Lemma 5.1]{gilyen2017OptQOptAlgGrad} in order to give some intuition about Jordan's gradient computation algorithm that can be viewed as a continuous extension of the Bernstein-Vazirani algorithm~\cite{bernstein1993QuantComplTheory}.
	\begin{lemma}[The core of Jordan's gradient computation algorithm]\label{lemma:genericJordan}
	Let $N=2^n$, and $\pmb{g}\in \R^d$ such that $\nrm{\pmb{g}}_\infty\leq 1/3$. If $\nrm{\left(\text{QFT}_{G_n}^{-1}\right)^{\otimes d}\ket{\psi} - \left(\text{QFT}_{G_n}^{-1}\right)^{\otimes d}\frac{1}{\sqrt{N^d}}\sum_{\pmb{x} \in G_n^d}e^{2\pi i N \ipc{\pmb{g}}{\pmb{x}}} \ket{\pmb{x}}}\leq \frac{1}{12}$, then measuring the state
	\begin{equation}\label{eq:closeFunctionApx}
		\left(\text{QFT}_{G_n}^{-1}\right)^{\otimes d} \ket{\psi}
	\end{equation}
	in the computational basis yields an estimate $\pmb{k}\in G_n^d$ such that 
    $$\Pr\left[|k_i-g_i|>\! 3/N\right]\leq 1/3 \quad \text{ for every  } i\in[d].$$
	\end{lemma}
	\begin{proof}
		The proof is analogous to that of \cite[Lemma 5.1]{gilyen2017OptQOptAlgGrad}.
		Observe that the ``ideal'' state is a product state
		\begin{equation*}
			\bigotimes_{i=1}^d\text{QFT}_{G_n}^{-1}\frac{1}{\sqrt{N}}\sum_{x_i \in G_n}e^{2\pi i N g_ix_i} \ket{x_i}=		\bigotimes_{i=1}^d\frac{1}{N}\sum_{x_i,k_i\in G_n}e^{2\pi i Nx_i(g_i-k_i)}\ket{k_i}.
		\end{equation*}
		Thus, after the measurement we obtain some coordinate-wise independent outcome $(k_1,\ldots,k_d)$. In the analysis of phase estimation \cite{nielsen2002QCQI}, it can be shown\footnote{Note that this is where we use the assumption $\nrm{\pmb{g}}_\infty\leq 1/3$ in order to convert the phases to the intervals $[-\frac13,\frac13]$. Also note that the Fourier transform we use is slightly altered, but the same argument still holds as in \cite[(5.34)]{nielsen2002QCQI}. One can also directly translate the result by considering the conjugation of the ordinary quantum Fourier transform with a tensor product of $n$ single-qubit unitaries.} that for every $i\in [d]$, the following holds:
		$$
		\Pr\left[|k_i-g_i|> \frac{3}{N}\right]\leq \frac{1}{4} \quad \text{ for every } i\in [d].
		$$
		Since we work with the state $\ket{\psi}$ instead of the ``ideal'' state, the measurement statistics might differ. On the other hand the closeness condition $\nrm{\ket{\psi} - \frac{1}{\sqrt{N^d}}\sum_{\pmb{x} \in G_n^d}e^{2\pi i N \ipc{\pmb{g}}{\pmb{x}}} \ket{\pmb{x}}}\leq \frac{1}{12}$ guarantees that the probability of the above event changes by at most $\frac{1}{12}$ (see for example \cite[Exercise 4.3]{wolf2019QCLectureNotes}).
	\end{proof}
	We will extensively use the follow corollary of for estimating various quantities.
	
	\begin{corollary}[Almost linear block-Hamiltonian to gradient]\label{cor:blockToGrad}
		Let $\eps,\delta\in\!(0,\frac{1}{6}]$, $b:=\lceil \log_2(\frac{24}{\eps}) \rceil$, $B=2^b$ and  $\beta:=\frac{1}{48}$. Suppose that we have an $a$-block-encoding $W$ of a diagonal matrix with diagonal entries $f(\pmb{x})\in \R$ for $\pmb{x} \in G_b^d$ satisfying $|f(\pmb{x}) - \ipc{\pmb{x}}{\pmb{g}}|\leq \frac{\eps\beta}{6\pi}$ for at least a $(1- \beta^2)$ fraction of the points in $G_b^d$. Then with $\bigO{\frac{1}{\eps}\log(\frac{d}{\delta})}$ (controlled) uses of $W$ (and its inverse) and 
		$\bigO{\left(d\log(\frac{1}{\eps})\log\log(\frac{1}{\eps})+\frac{a}{\eps}\right)\log(\frac{d}{\delta})}\!$ other gates with circuit depth $\bigO{\frac{\log(a)}{\eps}}\!$
		we can compute a vector $\pmb{k}\in [-4,4]^d$ such that 
		$\Pr\left[\nrm{\pmb{k}-\pmb{g}}_\infty> \eps\right]\leq \delta$.			    
	\end{corollary}
	\begin{proof}
		The main idea is to apply \autoref{lemma:genericJordan} with preparing the (approximate) initial state via block-Hamiltonian simulation \autoref{lem:blockHamSim}.
		The first step is to prepare a uniform superposition over the grid $G_b^d$ by applying a Hadamard gate to all $d\cdot b$ qubits, that are initially in the $\ket{0}$ state. 
		
		Note that due to the assumptions in the statement we have that $|f(\pmb{x})|\leq 1$ and so $|\ipc{\pmb{x}}{\pmb{g}}|\leq 1+\eps\leq \frac{7}{6}$ for at least $1-2^{-10}$ fraction of points in $G_b^d$ since $\beta\leq 2^{-5}$, in turn implying that $\nrm{\pmb{g}}_\infty\leq \frac{8}{3}$. Indeed, let us assume that $g_j> \frac{8}{3}$, we show that this would imply that for at least half of the points with $x_j\geq \frac{7}{16}$ we have that $\ipc{\pmb{x}}{\pmb{g}}> \frac{7}{6}$. First, clearly $x_j\cdot g_j > \frac{7}{6}$. Let $\pmb{\bar{g}}\in\R^{d-1}$ be the vector we get from $\pmb{g}$ by removing its $j$-th coordinate. Then for any $\pmb{\bar{x}}\in G_b^{d-1}$ we have that $\ipc{\pmb{\bar{x}}}{\pmb{\bar{g}}}=-\ipc{-\pmb{\bar{x}}}{\pmb{\bar{g}}}$ 
		so at least one of $\ipc{\pmb{\bar{x}}}{\pmb{\bar{g}}}, \ipc{-\pmb{\bar{x}}}{\pmb{\bar{g}}}$ is greater than or equal zero. Since $b\geq 4$ at least a $\frac{1}{16}$ fraction of points $x \in G_b$ satisfy $x\geq \frac{7}{16}$, and so for at least a $\frac{1}{32}$ fraction of points $x \in G_b^d$ we would get $\ipc{\pmb{x}}{\pmb{g}}> \frac{7}{6}$. Therefore, we will apply \autoref{lemma:genericJordan} to the gradient $\frac{\pmb{g}}{8}$ with precision $\frac{\eps}{8}$.
		
		First let us assume that we have access to a perfect phase oracle $P:=\sum_{\pmb{x}\in G_b^d} \ketbra{\pmb{x}}{\pmb{x}} e^{2\pi i \frac{B}{8} f(\pmb{x})}$ so that we can prepare the sate $\ket{\psi'}=\frac{1}{\sqrt{B^d}}\sum_{\pmb{x}\in G_b^d}\ket{\pmb{x}}e^{2\pi i \frac{B}{8} f(\pmb{x})}$. First let us bound the difference from the ideal state $\ket{\phi}=\frac{1}{\sqrt{N^d}}\sum_{\pmb{x} \in G_b^d}e^{2\pi i N \ipc{\pmb{g}}{\pmb{x}}} \ket{\pmb{x}}$ analogously to the proof of \cite[Lemma 5.1]{gilyen2017OptQOptAlgGrad}. Let $S\subseteq G_b^d$ be the set of points for which $|f(\pmb{x}) - \ipc{\pmb{x}}{\pmb{g}}|\leq \frac{\eps\beta}{4\pi}$ holds, then
		\begin{align*}
		\nrm{\ket{\psi'}-\ket{\phi}}^2\!
		&=\!\frac{1}{B^d}\sum_{\pmb{x}\in G_b^d}\left|e^{2\pi i \frac{B}{8} f(\pmb{x})}-e^{2\pi i \frac{B}{8} \ipc{\pmb{x}}{\pmb{g}}}\right|^2\\
		&=\!\frac{1}{B^d}\!\sum_{\pmb{x}\in S}\left|e^{2\pi i \frac{B}{8} f(\pmb{x})}-e^{2\pi i \frac{B}{8} \ipc{\pmb{x}}{\pmb{g}}}\right|^2
		\!\!+\!\frac{1}{B^d}\!\!\sum_{\pmb{x}\in G_b^d\setminus S}\!\left|e^{2\pi i \frac{B}{8} f(\pmb{x})}-e^{2\pi i \frac{B}{8} \ipc{\pmb{x}}{\pmb{g}}}\right|^2\\
		&\leq \!\frac{1}{B^d}\!\sum_{\pmb{x}\in S}\left|2\pi \frac{B}{8} f(\pmb{x})-2\pi \frac{B}{8} \ipc{\pmb{x}}{\pmb{g}}\right|^2
		\!\!+\!\frac{1}{B^d}\!\!\sum_{\pmb{x}\in G_b^d\setminus S}\!4 \tag{$|e^{iz}-e^{iy}|\leq |z-y|$}\\
		&=\!\frac{1}{B^d}\!\sum_{\pmb{x}\in S}(2\pi \frac{B}{8} )^2\left|f(\pmb{x})-\ipc{\pmb{x}}{\pmb{g}}\right|^2
		\!\!+4\frac{|G_b^d\setminus S|}{B^d} \\
		&\leq\!\frac{1}{B^d}\!\sum_{\pmb{x}\in S}4\beta^2+4\beta^2 \tag{by the assumptions of the corollary}\\	
		&\leq 8\beta^2.
		\end{align*}
		We can implement a $(3-2\sqrt{2})\beta$-approximation $\widetilde{P}$ of the perfect phase oracle $P$ by applying block-Hamiltonian simulation \autoref{lem:blockHamSim} to $W$.\footnote{An $\eps$-precise $(a+2)$-block-encoding of $e^{itH}$ is $\bigO{\sqrt{\eps}}$-close in operator norm to a perfect Hamiltonian simulation unitary $U$ of the form $\ketbra{0}{0}^{\otimes a+2}\otimes e^{\ci t A} + V$, where $V(\ket{0}^{a+2}\otimes I)=0$.} 
		This enables us to prepare an approximate state $\ket{\tilde{\psi}}$ such that $\nrm{\ket{\tilde{\psi}}-\ket{\psi'}}\leq (3-2\sqrt{2})\beta$ and so in turn $\nrm{\ket{\tilde{\psi}}\!-\!\ket{\phi}}\leq3\beta$.
		
		Let $j\in [d]$ be arbitrary an let us assume that we replace the $j$-th inverse quantum Fourier transform by an approximate circuit $\widetilde{Q}$ such that $\nrm{\widetilde{Q}-\text{QFT}_{G_b}^{-1}}\leq \beta$. Accordingly let us define $\ket{\psi^{(j)}}:=I_{G_b}^{\otimes[j-1]}\otimes\left(\text{QFT}_{G_b}\cdot\widetilde{Q}\right)\otimes I_{G_b}^{\otimes[d]\setminus[j]}$, then clearly $\nrm{\ket{\psi^{(j)}}-\ket{\phi}}\leq4\beta=\frac{1}{12}$, so that we can apply \autoref{lemma:genericJordan} to conclude that $\Pr\left[|k_j-g_j|> \frac{3}{N}\right]\leq \frac{1}{4}$. On the other hand the measurement statistics of the $j$-th register is not affected by unitaries that are applied on the other registers, so this conclusion holds even if we replace all inverse quantum Fourier transform by $\widetilde{Q}$. Thus if measure the state $\widetilde{Q}^{\otimes d}\ket{\tilde{\psi}}$ we have for every $j\in [d]$ that 	
		$$
		\Pr\left[|k_j-g_j|> \frac{3}{B}\right]\leq \frac{1}{3}.
		$$
			
		Finally, we repeat the entire procedure $2m+1$-times for $m:=\ceil{10\ln(\frac{d}{\delta})}$ and take the median of the estimates for each coordinate $j\in [d]$. If $|k_j-g_j|\leq \frac{3}{B}$ holds for at least $m+1$ estimates, then the median will give an $\frac{3}{B}\leq \frac{\eps}{8}$-precise estimate for $g_j$.
		We bound the probability of failure using the Chernoff-Hoeffding theorem~\cite [Theorem 1]{hoeffding1963ProbIneqSumsOfBoundedRVs} showing that the probability that $|k_j-g_j|> \frac{3}{B}$ holds for at least $m+1$ out of $2m+1$ estimates is at most $\exp(-D(\frac12\Vert \frac13)(2m+1))\leq\exp(-\frac{1}{20}(2m+1))\leq \exp(-\frac{m}{10})\leq\frac{\delta}{d}$, where $D(x\Vert y)=x\ln(\frac{x}{y})+(1-x)\ln(\frac{1-x}{1-y})$.
		This implies that the $\Pr\left[\nrm{\pmb{k}-\pmb{g}}_\infty> \eps\right]\leq \delta$.

		The query complexity follows from the fact that we prepare the state $\ket{\tilde{\psi}}$ a total of $\bigO{\log(\frac{d}{\delta})}$ times, each time making $\bigO{B}=\bigO{\frac{1}{\eps}}$ (controlled) queries to $W$. The additional gate complexity of preparing $\ket{\tilde{\psi}}$ is $\bigO{a}$ times the query complexity plus the number of initial Hadamard gates. The biggest contribution to the gate complexity comes from the implementation of the approximate (inverse) quantum Fourier transform $\widetilde{Q}$ \cite{barenco1996ApproxQFourierTrafo}.
		The gate complexity of $\widetilde{Q}$ can be bounded by $\bigO{b\log(b)}=\bigO{\log(\frac{1}{\eps})\log\log(\frac{1}{\eps})}$ while its depth by $\bigO{\log(b)}=\bigO{\log\log(\frac{1}{\eps})}$ as shown by~\cite{cleve2000FastParallelQFT}. The additional classical computation can also be performed in parallel with depth $\bigO{\mathrm{poly}(b,m)}$ which is $\bigO{\polylog({\frac{d}{\delta\eps}})}$, since $m=\bigO{\log(\frac{d}{\delta})}$, and $b=\bigO{\log(\frac{1}{\eps})}$.
	\end{proof}
    In \autoref{sec:unbiasedPhaseEst} we improve upon the above \autoref{lemma:genericJordan} and \autoref{cor:blockToGrad}
	by making them (essentially) unbiased, by using our new unbiased phase estimation subroutine instead of just applying $\left(\text{QFT}_{G_n}^{-1}\right)^{\otimes n}$ to each coordinate in Jordan's algorithm. Those improvements play a vital role for our mixed-state tomography results, but they are not necessary for pure-state tomography. Since the unbiased version has some additional $\log \log$ factors, we use the simpler routine for now.

 \section{Relations between vector estimates}
\label{s:normstuff}
In this section we prove two lemmas that relate different types of estimates for vectors. The first lemma shows a relation between estimates of the vector of amplitudes, and of the vector of corresponding probabilities. The second lemma relates $\ell_q$-estimates for different values of $q$ when the vector is normalized (as is the case with amplitudes and probabilities). Together, these two lemmas allow us to upper bound the complexity of giving $\ell_q$-norm estimates for both amplitudes and probabilities, starting from an $\ell_\infty$-norm estimate for amplitudes. Similarly, with these lemmas a lower bound on the complexity of finding an $\ell_1$-norm estimate for probabilities translates to a lower bound on all other cases.

\subsection{Relation between amplitude and probability estimation}\label{sec:amp2prob}

For a classical probability distribution, we learn all aspects of the distribution by estimating it in $\ell_1$-norm error, i.e., total variation distance. For pure quantum states the $\ell_2$-norm error plays a similar role. It is natural to ask how an $\ell_2$-norm estimate of a quantum state relates to an $\ell_1$-norm estimate of the probability distribution given by computational-basis measurements on that state. In a similar fashion, we want to understand this question also when using $\ell_\infty$-norm error on the quantum state. We provide answers to these questions by showing a relation between $\ell_q$-norm error on a state and $\ell_r$-norm error on the corresponding probability distribution; the special case $q=2, r=1$ is also discussed in \cite[Lemma~3.6]{bernstein1993QuantComplTheory}.

\begin{lemma}
\label{lem:linftol2}
  Let $\ket{\psi} = \sum_{j \in [d]} \alpha_j \ket{j}$ be a quantum state. Let $p\in \mathbb{R}^d$ given by $p_j = |\alpha_j|^2$ be the probability distribution arising from a computational-basis measurement. Let $q\in [2,\infty]$ and let  $t := \frac{1}{\frac{1}{q} + \frac{1}{2}} \in [1,2]$. An $\eps$-$\ell_{q}$-norm estimate $|\tilde{\alpha}|$ of $|\alpha|$ can be used to compute a $4\eps$-$t$-norm estimate $\tilde{p}$ of $p$, using $O(d)$ gates.
\end{lemma}
\begin{proof}
     We first note that for any $\ell_q$-norm estimate $|\tilde{\alpha}|$ of a unit vector $|\alpha|$ we may assume $\||\tilde{\alpha}|\|_2 = 1$, that is, the vector represents a pure state. Indeed, if this is not the case, we can instead use $|\tilde{\alpha}| / \nrm{|\tilde{\alpha}|}_q$, which we can compute using $O(d)$ gates, and which satisfies
     \begin{align*}
       \nrm{|\alpha| - \frac{|\tilde{\alpha}|}{\nrm{|\tilde{\alpha}|}_q}}_q &\leq \nrm{|\alpha| - |\tilde{\alpha}|}_q + \nrm{|\tilde{\alpha}| - \frac{|\tilde{\alpha}|}{\nrm{|\tilde{\alpha}|}_q}}_q \\
       &\leq \eps + \nrm{|\tilde{\alpha}|}_q \cdot \trm{1 - \frac{1}{\nrm{|\tilde{\alpha}|}_q}} \\
       &\leq \eps + (1+\eps) \cdot \trm{1 - \frac{1}{1+\eps}} \\
       &= 2\eps.
     \end{align*}
     
     By assumption, $q = \frac{t}{1-t/2}$. Using H\"older's inequality for $\lambda_1,\lambda_2\geq 1$ with $\frac{1}{\lambda_1}+\frac{1}{\lambda_2} = 1$, we obtain the following upper bound on the $\ell_t$-norm error in an estimate of $p$:
  \begin{align*}
  \nrm{p-\tilde{p}}_t &=  \trm{ \sum_{j \in [d]} |p_j - \tilde{p}_j|^t }^{1/t}\\
  &=  \trm{ \sum_{j \in [d]} ||\alpha_j|^2 - |\tilde{\alpha}|^2_j|^t }^{1/t}\\
  &\leq  \trm{ \sum_{j \in [d]} \left||\alpha_j| - |\tilde{\alpha}|_j\right|^t \left||\alpha_j| + |\tilde{\alpha}|_j\right|^t  }^{1/t}\\
  &\leq  \trm{ \sum_{j \in [d]} \left||\alpha_j| - |\tilde{\alpha}|\right|^{t\lambda_1}  }^{\frac{1}{t\lambda_1}}  \trm{ \sum_{j \in [d]} \left||\alpha_j| + |\tilde{\alpha}|_j\right|^{t\lambda_2}  }^{\frac{1}{t\lambda_2}} \\
  &= \nrm{|\alpha|-|\tilde{\alpha}|}_{t\lambda_1}\nrm{|\alpha|+|\tilde{\alpha}|}_{t\lambda_2}.
  \end{align*}
  (Recall that $\lambda_1, \lambda_2, t \geq 1$, so $t\lambda_1 \ge 1$ and $t\lambda_2 \ge 1$.) Pick $\lambda_2 = 2/t$. We then have
  \[
  \nrm{|\alpha|+|\tilde{\alpha}|}_{t\lambda_2} = \nrm{|\alpha|+|\tilde{\alpha}|}_2  \leq \nrm{|\alpha|}_2+\nrm{|\tilde{\alpha}|}_2 \leq 2.
  \]
  Combining this with $\lambda_1 = \frac{1}{1-\frac{1}{\lambda_2}} = \frac{1}{1-\frac{t}{2}}$, we get
  \begin{align*}
  \nrm{p-\tilde{p}}_t &\leq   \nrm{|\alpha|-|\tilde{\alpha}|}_{t\lambda_1}\nrm{|\alpha|+|\tilde{\alpha}|}_{t\lambda_2}\\
    & \leq  \nrm{|\alpha|-|\tilde{\alpha}|}_{\frac{t}{1-t/2}} \cdot 2\\
    &= 2\nrm{|\alpha|-|\tilde{\alpha}|}_{q}\\
    &\leq 4\eps. 
  \end{align*}
\end{proof}

Note that the reverse does not hold, and in particular $\eps$-$\ell_1$-norm estimates of $p$ are not equivalent to  $\Theta(\eps)$-$\ell_2$-norm estimates of $\alpha$ . This is not only due to the information about the phases being lost: even for the case $d=2$ where the $\alpha_j$ are positive reals, estimating the probabilities is not enough to learn the amplitudes to a similar error. In particular, let  $p_0 = \alpha_0^2 = \eps \leq 1$ and let $\tilde{p}_0 = p_0+\eps$ be an $\eps$-estimate for the probability. The amplitude satisfies 
\[
\tilde{\alpha}_0 = \sqrt{\tilde{p}_0} = \sqrt{p_0+\eps} = \sqrt{2\eps} \geq \trm{1+\frac14}\sqrt{\eps} = \alpha_0 + \frac14\sqrt{\eps}.
\]
Thus, the precision gets quadratically worse for small amplitudes.

\subsection{Dimension-independent norm conversion for normalized vectors}
If we have an estimate of a vector with error at most $\eps$ in, for example, the $\ell_{\infty}$-norm, we can use norm conversion to show that this is also an estimate with error at most $\eps d^{1/q}$ in the $\ell_q$-norm. However, this bound is poor for large $d$. Here, we show that we can do better if we know that the vector we are estimating is normalized in some $\ell_s$-norm, using the fact that such a vector cannot have too many large entries. In fact, we obtain a norm conversion lemma that does not depend on the dimension at all. We first prove a very general version of the following lemma; for results in subsequent sections of the paper we always use $s=2$ (for quantum states) or $s=1$ (for probability distributions), and set $\gamma = 1$.

 \begin{lemma}\label{lem:new-norm-conversion}
  Let $\alpha \in \C^d$ be such that $\nrm{\alpha}_s \leq \gamma$. Let $\tilde{\alpha}\in \C^d$ be such that $\nrm{\alpha-\tilde{\alpha}}_{\infty}\leq \eta$. Let $\bar{\alpha} \in \C^d$ be the vector defined as $\bar{\alpha}_j = \tilde{\alpha}_j$ if $|\tilde{\alpha}_j| \ge 2 \eta$, $\bar{\alpha}_j = 0$ otherwise. Then for all $q\in (s,\infty)$ we have $\nrm{\alpha-\bar{\alpha}}_q \leq \min\{ 4\eta^{\frac{q-s}{q}}\gamma^{\frac sq}, 3 d^{1/q}\eta\}$.
 \end{lemma}
 \begin{proof}
 The second term in the $\min$ follows from the standard norm conversion and the fact that $\bar{\alpha}$ is an $3\eta$-$\ell_\infty$-approximation; thus, we only need to prove the first term.

    We know that $\nrm{\alpha-\tilde{\alpha}}_\infty\leq \eta$. Let $J = \{j\in [d] : |\tilde{\alpha}_j| \geq 2\eta\}$. Then 
     \[
         \{j\in [d] : |\alpha_j| \geq 3 \eta\} \subseteq J \subseteq \{j\in [d] : |\alpha_j| \geq \eta \}.
     \]
     And, as $\nrm{\alpha}_s \leq \gamma$, we have $|J| \leq \frac{\gamma^s}{\eta^s}$.

     Now, let $\bar \alpha$ be $\tilde{\alpha}$ on all $j\in J$ and $0$ everywhere else. On the indices in $J$ we know that $\bar \alpha$ is an $\eta$ estimate of $\alpha$. Then 
     \allowdisplaybreaks
     \begin{align*}  
         \nrm{\alpha-\bar{\alpha}}_q &= \left( \sum_{j\in[d]} |\alpha_j-\bar{\alpha}_j|^q\right)^{1/q}\\
          &\leq \left( \sum_{j\not\in J} |\alpha_j-\bar{\alpha}_j|^q\right)^{1/q} +  \left( \sum_{j\in J} |\alpha_j-\bar{\alpha}_j|^q\right)^{1/q}\\
          &\leq \left( \sum_{j\not\in J} |\alpha_j|^{q-s}|\alpha_j|^s\right)^{1/q} +  \left( |J|\max_j |\alpha_j-\bar{\alpha}_j|^q\right)^{1/q}\\
          &\leq \left( \sum_{j\not\in J} (3\eta)^{q-s}|\alpha_j|^s\right)^{1/q} +  \left( \frac{\gamma^s}{\eta^s}\eta^q\right)^{1/q}\\
           &= (3\eta)^{\frac{q-s}{q}} \left( \sum_{j\not\in J} |\alpha_j|^s\right)^{1/q} +  \eta^{\frac{q-s}{q}}\gamma^{\frac sq}\\
           &\leq 4\eta^{\frac{q-s}{q}}\gamma^{\frac sq}. \qedhere
     \end{align*}  
 \end{proof}

The lemma stated above is very general, but we only use it with several very specific parameter settings. Thus, we present the following simplified statement.

\begin{corollary}\label{cor:norm-conversion}
 Let $\eps \in (0,1]$, $s \leq q$, and let $y$ be an $\ell_s$-normalized complex vector. In order to obtain an $\eps$-$\ell_q$-norm estimate of $y$, an $\eta$-$\ell_\infty$-norm estimate suffices for
 \[
    \eta = \max\left\{\frac{1}{3} \trm{\frac{\eps}{3} }^{\frac{1}{1-\frac{s}{q}}},\frac{\eps}{ d^{\frac{1}{q}}}\right\}.
 \]
\end{corollary}
\begin{proof}
The first term follows from \autoref{lem:new-norm-conversion} by letting $\gamma = 1$. The second term in the max comes from a standard norm-conversion on the vector of errors, as a $v\in [-\eps,\eps]^d$ has $q$-norm at most $\eps d^{1/q}$.
\end{proof}

This corollary has an immediate consequence. If one is only interested in finding an $\ell_q$-estimate of the vector of probabilities $p_j$ prepared by some unitary operation $U : \ket{0} \mapsto \sum_{j=1}^d \sqrt{p_j}\ket{j}$, then one can directly apply \autoref{cor:norm-conversion} in conjunction with the $\tilde{O}(1/\eps)$-query algorithm for $\ell_{\infty}$-algorithm from \cite{apeldoorn2021QProbOraclesMulitDimAmpEst}. The number of controlled and inverse calls to $U$ then becomes \[
\widetilde{O}(1/\eta) = \widetilde{O}\left(\min\left\{\left(\frac{3}{\eps}\right)^{\frac{1}{1-\frac1q}}, \frac{d^{\frac1q}}{\eps}\right\}\right),
\]
which we show to be optimal up to polylogarithmic factors in \autoref{subsec:lb-with-inverses}.

\section{Pure-state tomography using copies}
\label{s:classicaltomo}
In this section we present and analyze pure-state tomography algorithms that use very little quantum power. The first algorithm that we describe, in \autoref{subsec:classical-samples}, is part of the folklore: we just take measurements in the computational basis, and obtain the absolute values of the amplitudes from the measurement outcomes. We are not aware of a specific reference for the sample complexity of this method, hence we provide a proof for completeness. Then, we add the ability to perform some operations on the quantum state, in \autoref{subsec:cond-samples} and \autoref{subsec:attempt}: the first section simplifies the analysis of the tomography algorithm given in \cite{kerenidis2018QIntPoint}, the second one relaxes some of the assumptions with a slight increase in the sample complexity. The strongest model in this section relates to the case where we have access to a state-preparation unitary and its controlled version, but not its inverse. We discuss the setting where the inverse is available in \autoref{s:quantumtomo}.

\subsection{Absolute values using computational-basis measurements}
\label{subsec:classical-samples}

Given classical samples via computational-basis measurements, how many samples do we need for an $\ell_q$-norm estimate of $\alpha$? Clearly we cannot learn the phases, so we have to limit ourselves to the absolute values of the amplitudes. Even then, the remark at the end of \autoref{sec:amp2prob} seems discouraging: on the surface, estimation of the related distribution $p$ seems the best that we can do with computational-basis measurements, and converting the error bound from probabilities to amplitudes makes the precision quadratically worse.  However, as we discuss next, $\tilde O(\frac{1}{\eps^2})$ samples suffice and are optimal to give an $\ell_\infty$-norm estimate of $|\alpha|$.

\begin{proposition}\label{lem:sample-infty-norm}
  Let $0 < \eps, \delta < 1$. Let $\ket{\psi} = \sum_{j \in [d]} \alpha_j \ket{j}$ be a quantum state with $\alpha_j\in \C$, and let $p\in \mathbb{R}^d$, defined by $p_j = |\alpha_j|^2$, be the probability distribution of the outcomes of a computational-basis measurement. Then, $O(\log(d/\delta)/\eps^2)$ measurements of $\ket{\psi}$ in the computational basis suffice to learn an $\eps$-$\ell_{\infty}$-norm estimate $|\tilde{\alpha}|$ of $|\alpha|$, with success probability at least $1-\delta$.
\end{proposition}
\begin{proof}
 Let us consider a single coordinate $\alpha_j$ with associated probability $p_j = |\alpha_j|^2$. Our goal is to estimate $|\alpha_j|$. We take $k$ samples to find an estimate $\tilde{p}_j$ of $p_j$. The Chernoff bound tells us that for the error $\eps_j$ in this coordinate we have
  \[
    \P[\tilde{p}_j >  p_j+\eps_j] \leq e^{-D(p_j+\eps_j||p_j)k}, \qquad \text{if } p_j + \eps_j < 1,
  \]
  \[
    \P[\tilde{p_j} >  p_j-\eps_j] \leq e^{-D(p_j-\eps_j||p_j)k}, \qquad \text{if } p_j - \eps_j > 0,
  \]
 where $D(x||y)$ is Kullback–Leibler divergence. We need the conditions shown on the right-hand side because the Kullback-Leibler divergence is only defined for $x,y \in (0,1)$, but it is easily observed that if the conditions on the right are not satisfied, the probabilities on the left-hand side trivially become $0$. Since $D(x||y) \geq \frac{(x-y)^2}{2\max\{x,y\}}$, for all $x,y \in (0,1)$, we get
 \[
     \P[|\tilde{p}_j - p_j | > \eps_j] \leq 2e^{-\frac{\eps_j^2}{2(p_j+\eps_j)}k},
 \]
 and it is easily checked that this bound also holds whenever $p_j + \eps_j \geq 1$, or $p_j - \eps_j \leq 0$. Hence, picking $k \geq \frac{2(p_j + \eps_j)\ln(2/\delta')}{\eps_j^2} = \frac{2(|\alpha_j|^2 + \eps_j)\ln(2/\delta')}{\eps_j^2}$ ensures that $ P[|\tilde{p}_j - p_j | \geq \eps_j] \leq \delta'$. 
 
 We now pick $\eps_j = \eps|\alpha_j|/2 + (\eps/2)^2$. Note that we do not actually know this value, as it depends on the yet-to-be-estimated $\alpha_j$, but with this choice we find that
 \[
   \frac{2(|\alpha_j|^2 + \eps_j)\ln(\frac{2}{\delta'})}{\eps_j^2} = \frac{2(|\alpha_j|^2 + \frac{\eps|\alpha_j|}{2} + (\frac{\eps}{2})^2)\ln(\frac{2}{\delta'})}{(\frac{\eps|\alpha_j|}{2} + (\frac{\eps}{2})^2)^2} \leq \frac{2(|\alpha_j|^2 + \eps|\alpha_j| + (\frac{\eps}{2})^2)\ln(\frac{2}{\delta'})}{(\frac{\eps}{2})^2(|\alpha_j| + \frac{\eps}{2})^2} = \frac{8\ln(\frac{2}{\delta'})}{\eps^2}.
 \]
 Thus, it suffices to choose $k = 8\ln(2/\delta')/\eps^2$. Letting $\delta' = \delta / d$ and applying the union bound, we have that, with probability $1-\delta$, for all $j\in[d]$, the resulting estimates $\tilde{p}_j$ satisfy $|\tilde{p}_j - p_j| \leq \eps_j$. First, this implies
 \begin{align*}
  |\tilde{\alpha}_j| - |\alpha_j| &\leq \sqrt{p_j+\eps_j} - |\alpha_j| \\
  &= \sqrt{|\alpha_j|^2 + \frac{\eps|\alpha_j|}{2} + \trm{\frac{\eps}{2}}^2} - |\alpha_j| \\
  &\leq \sqrt{|\alpha_j|^2 + \eps|\alpha_j| + \trm{\frac{\eps}{2}}^2} - |\alpha_j| \\
  &= |\alpha_j| + \frac{\eps}{2} - |\alpha_j| \\
  &= \frac{\eps}{2} \\
  &< \eps.
 \end{align*}

 Next we show that $|\alpha_j| - |\tilde{\alpha}_j| < \eps$. First consider the case where $p_j \leq \eps_j$. In that case, we have
 \[
    |\alpha_j|^2 = p_j \leq \frac{\eps|\alpha_j|}{2} + \trm{\frac{\eps}{2}}^2 \qquad \Leftrightarrow \qquad \trm{\frac{2|\alpha_j|}{\eps}}^2 \leq \frac{2|\alpha_j|}{\eps} + 1 \qquad \Leftrightarrow \qquad |\alpha_j| \leq \frac{1+\sqrt{5}}{4}\eps.
 \]
 Hence, we find that
 \[
    |\alpha_j| - |\tilde{\alpha_j}| \leq |\alpha_j| \leq \frac{1+\sqrt{5}}{4}\eps < \eps.
 \]
 On the other hand, if $p_j > \eps_j$, we have
 \begin{align*}
  |\alpha_j| - |\tilde{\alpha}_j| &\leq |\alpha_j| - \sqrt{p_j-\eps_j} = |\alpha_j| - \sqrt{|\alpha_j|^2 - \frac{\eps|\alpha_j|}{2} - \trm{\frac{\eps}{2}}^2} \\
  &= |\alpha_j| - \trm{|\alpha_j| - \frac{\eps}{4}} \sqrt{1 - \frac{5(\frac{\eps}{2})^2}{4(|\alpha_j| - \frac{\eps}{4})^2}} \\&\leq |\alpha_j| - \trm{|\alpha_j| - \frac{\eps}{4}} \trm{1 - \frac{5(\frac{\eps}{2})^2}{4(|\alpha_j| - \frac{\eps}{4})^2}} \\
  &= \frac{\eps}{4} + \frac{5(\frac{\eps}{2})^2}{4(|\alpha_j| - \frac{\eps}{4})} \leq \frac{\eps}{4} + \frac{5\eps^2}{16\frac{\sqrt{5}}{4}\eps} \\& = \frac{\eps}{4} + \frac{\sqrt{5}\eps}{4} \\&< \eps,
 \end{align*}
 where in the last line, we used that $|\tilde{\alpha}_j| > (1+\sqrt{5})\eps/4$.
 Thus, we can compute a vector $|\tilde{\alpha}|$ that is an $\eps$-$\ell_\infty$-norm estimate of $|\alpha|$ with $O\trm{\log(d/\delta)/\eps^2}$ samples.
 \end{proof}

We can use the above theorem to approximate the vector of absolute values of the amplitudes in other norms as well.
 
\begin{theorem}\label{thm:state-sampling}
  Let $0 < \delta < 1$, $\eps > 0$, $d \in \N$ and $q \in [2,\infty]$. Let $\ket{\psi} = \sum_{j \in [d]} \alpha_j \ket{j}$ be a quantum state with $\alpha_j\in \C$. Then,
  \[
    O\trm{\min\left\{\trm{\frac{3}{\eps}}^{\frac{1}{\frac12 - \frac1q}}, \frac{d^{\frac2q}}{\eps^2}\right\} \cdot \log\trm{\frac{d}{\delta}}}
  \]
  computational-basis measurements of $\ket{\psi}$ suffice to learn an $\eps$-$\ell_{q}$-norm estimate $|\tilde{\alpha}|$ of $|\alpha|$, with success probability $1-\delta$.
\end{theorem}
\begin{proof}
Since the vector $|\alpha|$ is normalized in $\ell_2$-norm, we know from \autoref{cor:norm-conversion}, that in order to obtain an $\eps$-$\ell_q$-norm estimate of $|\alpha|$, it suffices to find an $\eta$-$\ell_{\infty}$-estimate of $|\alpha|$, where
\[\eta = \max\left\{\frac13\trm{\frac{\eps}{3}}^{\frac{1}{1-\frac{2}{q}}}, \frac{\eps}{d^{\frac1q}}\right\}.\]
From \autoref{lem:sample-infty-norm}, we now find that this can be done using
\[O\trm{\frac{\log\trm{\frac{d}{\delta}}}{\eta^2}} = O\trm{\min\left\{\trm{\frac{3}{\eps}}^{\frac{1}{\frac12 - \frac1q}}, \frac{d^{\frac2q}}{\eps^2}\right\} \cdot \log\trm{\frac{d}{\delta}}}.\]
computational-basis measurements.
\end{proof}

\subsection{Recovering the phase information using conditional samples}
\label{subsec:cond-samples}

Our discussion above shows how to estimate the vector of the absolute values of the amplitudes $\alpha_j$ with $O(\log d / \eps^2)$ copies of the quantum state. In this section, we consider having conditional samples of the state $\ket{\psi}$, by which we mean states of the form
\[
\frac{\ket{0}\ket{\psi}+\ket{1}\ket{0}}{\sqrt{2}},
\]
and we consider the problem of recovering all complex amplitudes of $\ket{\psi}$, including the phases.

We note here that if we have access to a controlled state-preparation unitary, we can prepare such a conditional sample of the state with one call to this operation. Crucially, we do not need the inverse of the state-preparation unitary -- if we have access to that as well, then the results from \autoref{s:quantumtomo} improve over those presented here. 

The algorithm is based on the Hadamard test, as described in the next result. 
\begin{lemma}\label{lem:Htrick}
Let $\ket{\psi_0} = \sum_{j \in [d]} \alpha_j \ket{j}$, $\ket{\psi_1} = \sum_{j \in [d]} \beta_j \ket{j}$, and let
\[
  \ket{\phi} = \frac{\ket{0}\ket{\psi_0} + \ket{1}\ket{\psi_1}}{\sqrt{2}}.
\]
Using $O(\log(d/\delta)/\eps^2)$ copies of $\ket{\phi}$, we can, with success probability at least $1-\delta$, compute an $\eps$-$\ell_\infty$-norm estimate of the $2d$-dimensional vector containing entries $|\alpha_j \pm \beta_j|$, for all $j=0,\dots,d-1$.
\end{lemma}
\begin{proof}
Note that
\[
\ket{\phi} = \frac{1}{\sqrt{2}}\trm{ \ket{0}\sum_{j \in [d]} \alpha_j \ket{j} + \ket{1} \sum_{j \in [d]} \beta_j \ket{j} }.
\]
Applying a Hadamard gate to the first qubit yields
\[
\frac{1}{2}\trm{ \ket{0}\sum_{j \in [d]} (\alpha_j +\beta_j) \ket{j} + \ket{1} \sum_{j \in [d]} (\alpha_j-\beta_j) \ket{j} }.
\]
We can now perform computational-basis measurements on this new state, and build up a histogram of the observed outcomes. The proposition then follows from \autoref{lem:sample-infty-norm}, by setting the precision to $\eps/2$.
\end{proof}

Inspired by the method used by Kerenidis and Prakash~\cite{kerenidis2018QIntPoint}, we apply the above proposition to two states with amplitudes $\alpha$ and $|\alpha|$. The algorithm of Kerenidis and Prakash is only concerned with real amplitudes, hence it only needs to estimate the sign of each large $\alpha_j$. Clearly this can be learned from a sufficiently precise estimate of $|\alpha_j-|\alpha_j||$. Since we consider general phases we need to be more careful, as we need to distinguish between a very small positive complex component $\eps\ci$ and a very small negative complex $-\eps\ci$. To this end we also apply the proposition to $\alpha$ and $\ci|\alpha|$, and we give a more careful geometric analysis.

\begin{figure}[ht]

\centering
\includegraphics[width=.5\textwidth]{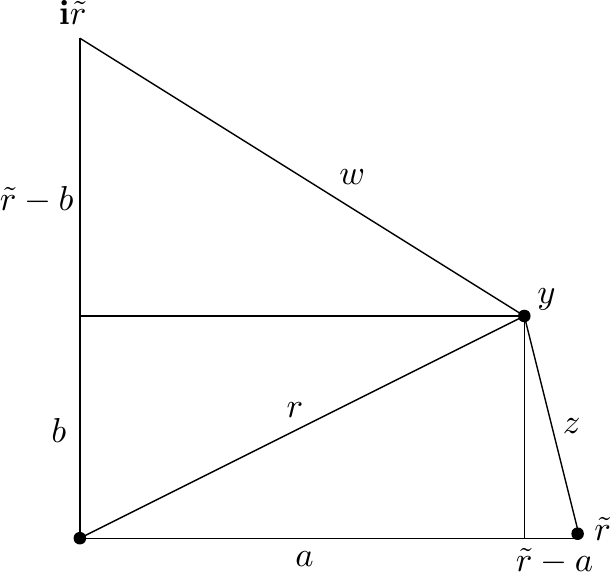}
\caption{Geometry of the different points in $\C$ involved in the proof of \autoref{prop:sammplisch-upper-bound}. Subscripts ``$j$'' have been dropped for clarity.}\label{fig:geometric}
\end{figure}

\begin{proposition}\label{prop:sammplisch-upper-bound}
  Let $0 < \eps, \delta < 1$, and let $\ket{\psi} = \sum_{j \in [d]} \alpha_j \ket{j}$ be a quantum state with $\alpha_j\in \C$. Then $O(\log(d/\delta)/\eps^2)$ copies of $(\ket{0}\ket{\psi} + \ket{1}\ket{0})/\sqrt{2}$ suffice to compute an $\eps$-$\ell_{\infty}$-norm estimate $\tilde{\alpha}$ of $\alpha$, with success probability at least $1-\delta$.
\end{proposition}

\begin{proof}
Let us consider a single amplitude $\alpha_j = (a_j,b_j)$. Let $r_j = |\alpha_j|$.
Using \autoref{lem:sample-infty-norm} we get an $\eps /32$ estimate $\tilde{r}_j$ of $r_j$ with $O(\log(d/\delta)/\eps^2)$ samples. We only consider the elements where $\tilde{r}_j \geq \eps/2$, as we can set the rest to $0$ without introducing too much error.

Let $z_j = |\alpha_j-\tilde{r}_j|$, and $w_j = |\alpha_j-\ci \tilde{r}_j|$; see \autoref{fig:geometric}. Using \autoref{lem:Htrick} we can find an $\eps / 32$-approximation of $z_j$ with the same number of samples as before, and similarly for $w_j$, because we can construct a controlled unitary that transforms $(\ket{0}\ket{\psi} + \ket{1}\ket{0})/\sqrt{2}$ into $(\ket{0}\ket{\psi} + \ket{1}\sum_{j} \tilde{r}_j \ket{j})/\sqrt{2}$, and similarly with an extra phase (we assume w.l.o.g.\ that $\nrm{\tilde{r}}_2 = 1$, as in the proof of \autoref{lem:linftol2}).

We now show how to find $a_j$ from $r_j$ and $z_j$, should we know them exactly; note that $z_j$ is still defined using measured value $\tilde{r}_j$, i.e., we are in the situation of \autoref{fig:geometric}. The squared length of the vertical line down from $\alpha_j$ is $b_j^2 = r_j^2-a_j^2$ by Pythagoras. Applying Pythagoras again, we find that 
\[
z_j^2 = b_j^2+ (\tilde{r}_j-a_j)^2 = r_j^2-a_j^2 +(\tilde{r}_j-a_j)^2 = r_j^2+\tilde{r}_j^2 - 2\tilde{r}_ja_j,
\]
hence $a_j = \frac12 \tilde{r}_j + \frac12 \frac{r_j^2}{\tilde{r}_j}- \frac{z_j^2}{2\tilde{r}_j}$. Thus, we have a formula for $a_j$, and we just need to bound the error that may affect $a_j$ when we use the estimates $\tilde{r}_j$ and $\tilde{z}_j$ instead of $r_j$ and $z_j$. Here it is important that we used our estimate $\tilde{r}_j$ for the $\beta_j$ in \autoref{lem:Htrick}, and hence $\tilde{r}_j$ is the exact length of the horizontal line, not an estimate of it.

To give an upper bound on the error induced by the error in our estimates $\tilde{r}_j$ and $\tilde{z}_j$, we consider the gradient of the function $f_{a_j}(r_j, z_j) = \frac12 \tilde{r}_j + \frac12 \frac{r_j^2}{\tilde{r}_j}- \frac{z_j^2}{2\tilde{r}_j}$ in terms of $r_j$ and $z_j$: 
\[
    \nabla f_{a_j} = \trm{\frac{r_j}{\tilde{r}_j}, -\frac{z_j}{\tilde{r}_j}}.
\]
We bound the $\ell_1$-norm of the gradient on the box defined by the constraints $|r_j-\tilde{r}_j| \le \eps/32$, $|z_j-\tilde{z}_j| \le \eps/32$. We know that $\eps/2 \leq \tilde{r}_j$, and hence that $r_j+\eps/32\leq \tilde{r}_j+\eps/16\leq 2\tilde{r_j}$, so the first coordinate is upper bounded in absolute value by $2$. Furthermore, $z_j\leq \tilde{r}_j+r_j\leq 3\tilde{r}_j$ by the triangle inequality, so $|-\frac{z_j+\eps/32}{\tilde{r}_j}|\leq 4$.  Hence, the sum of absolute values of the entries of the gradient is upper bounded by $8$ over the whole box, therefore the additive error on $a_j$ is at most 8 times the additive error on $r_j$ and $z_j$. Since $|r_j - \tilde{r}_j| \le \eps/16$ and $|z_j - \tilde{z}_j| \le \eps/16$, we have $|a_j - (\tilde{r}_j - \frac{\tilde{z}_j^2}{2\tilde{r}_j}) | \le 8\eps/16 = \eps/2$.

With a similar argument, but using $w_j$ instead of $z_j$, we also find $b_j$ up to error $\eps/2$, and hence $\alpha_j$ up to error $\eps$.
\end{proof}
As a controlled state-preparation unitary can be used to prepare the conditional samples we immediately get the following corollary.
\begin{corollary}
  \label{cor:controlled_state_linf}
  Let $0 < \eps,\delta < 1$, and let $\ket{\psi} = \sum_{j \in [d]} \alpha_j \ket{j}$ be a quantum state with $\alpha_j\in \C$. Then $O(\log(d/\delta)/\eps^2)$ applications (in parallel) of a controlled state-preparation unitary for $\ket{\psi}$ suffice to compute an $\eps$-$\ell_{\infty}$-norm estimate $\tilde{\alpha}$ of $\alpha$, with success probability at least $1-\delta$.
\end{corollary}
The sample complexity of \autoref{cor:controlled_state_linf} is asymptotically the same as in the algorithm of Kerenidis and Prakash \cite{kerenidis2018QIntPoint}, but our analysis is simpler thanks to \autoref{lem:sample-infty-norm}, we estimate both the real and the imaginary part, and we directly get $O(\log \frac{1}{\delta})$ dependence on the probability of failure (as opposed to probability of success $1-1/d^c$ for some constant $c$ in \cite{kerenidis2018QIntPoint}).

Using the above result, we can also construct algorithms that approximate $\alpha$ in other $\ell_q$-norms.
\begin{theorem}\label{thm:condsamp}
  Let $0 < \eps, \delta < 1$, and let $\ket{\psi} = \sum_{j \in [d]} \alpha_j\ket{j}$ be a quantum state with $\alpha_j \in \C$. Then,
  \[
    O\trm{\min\left\{\trm{\frac{3}{\eps}}^{\frac{1}{\frac12-\frac1q}}, \frac{d^{\frac2q}}{\eps^2}\right\} \cdot \log\frac{d}{\delta}}
  \]
  copies of the state $(\ket{0}\ket{\psi} + \ket{1}\ket{0})/\sqrt{2}$ suffice to compute an $\eps$-$\ell_q$-norm estimate $\tilde{\alpha}$ of $\alpha$, with success probability at least $1-\delta$.
\end{theorem}

\begin{proof}
  The proof follows from \autoref{cor:norm-conversion} combined with \autoref{prop:sammplisch-upper-bound}, in exactly the same way as in the proof of \autoref{thm:state-sampling}.
\end{proof}

\subsection{Amplitudes up to a global phase with only copies of the state}
\label{subsec:attempt}

Finally we consider the model in which we simply have access to copies of the pure state. Although this model is conceptually simple, the estimation algorithm is more complicated than before. The number of samples required is still $\tilde{O}(1/\eps^2)$, but slightly worse in polylogarithmic factors. In fact the method is very similar to the proof of \autoref{prop:sammplisch-upper-bound}, but instead of comparing $\alpha_j$ to $|\alpha_j|$, we have to compare the amplitudes to each other. 
\begin{proposition}\label{prop:state-reconstruction}
  Let $0 < \eps,\delta < 1$, and let $\ket{\psi} = \sum_{j \in [d]} \alpha_j \ket{j}$ be a quantum state with $\alpha_j\in \C$. Then $O(\log(d)\log(d/\delta)/\eps^2)$ copies of $\ket{\psi}$, with the ability to perform unitary operations on each copy before measurement, suffice to compute an $\eps$-$\ell_{\infty}$-norm estimate $\tilde{\alpha}$ of $\alpha$, up to global phase, with success probability at least $1-\delta$.
\end{proposition}

\begin{proof}
  As before, let $r_j = |\alpha_j|$ and let $\alpha_j = a_j+b_j\ci$.  We first use \autoref{lem:sample-infty-norm} to compute an $\eps/(16m)$-estimate $\tilde{r}_j$ of $r_j$. We only consider the coordinates where $\tilde{r}_j \geq \eps/2$, and permute the basis states in all remaining copies of $\ket{\psi}$ in such a way that these form the first $k$ coordinates of the state. Let $m$ be the smallest value such that $k< 2^m$, i.e., we are only interested in the amplitudes for basis states where all but the last $m$ qubits are $0$. For ease of notation we consider the sub-normalized state $\ket{\phi}$ corresponding to this part and relabel the indices so the $\tilde{r}_j$ are in decreasing order. Note that the remaining part of $\ket{\psi}$ has $\ell_\infty$-norm less than $\eps/2$, hence it can be ignored for our estimation. 
  
  For $h \in [m]$ we consider the state resulting from applying a Hadamard gate to the $h$-th qubit of $\ket{\phi}$ (below, for an $m$-digit binary string $j \in [2^m]$, we write $j_h$ to denote the $h$-th binary digit):
  \[
  I^{\otimes h-1}\otimes H\otimes I^{\otimes m-h} \ket{\phi} = \frac{1}{\sqrt{2}}\sum_{\substack{j\in[2^m]\\j_h=0}}(\alpha_j + \alpha_{j+2^h}) \ket{j} + (\alpha_j-\alpha_{j+2^h})\ket{j+2^h}.
  \]
  Hence we can learn $\eps/(16m)$-estimates of $s_{j,h} = |\alpha_j - \alpha_{j+2^h}|$ using computational-basis measurements with success probability at least $1-\delta/(2m)$. Note that this can be interpreted as an application of \autoref{lem:Htrick},
  where the $h$-th qubit is considered the flag. Repeating this with an additional phase gate also gives estimates $t_{j,h} = |\alpha_j - \ci \alpha_{j+2^h}|$.
 
  We now consider a single value of $h$ and aim to learn $\alpha_{j+2^h}$ relative to $\alpha_j$. For now, we assume that $\alpha_j\in \R_{\geq 0}$ and show how to give an estimate of $\alpha_j+2^h$; we discuss how to relax the assumption subsequently. Note that:
  \[
      s_{j,h}^2 = |\alpha_j|^2 + |\alpha_{j + 2^h}|^2 - 2 \Re(\alpha_j^{\dag} \alpha_{j+2^h}),
  \]
  and therefore, due to our assumption on $\alpha_j$, we have:
  \[
    a_{j+2^h} =  \frac{r_j^2 + r_{j+2^h}^2 - s_{j,h}^2}{2 r_j},
  \]
  with a similar argument as in \autoref{prop:sammplisch-upper-bound}. As in \autoref{prop:sammplisch-upper-bound}, we consider this estimate of $a_{j+2^h}$ as a function of $r_j, r_{j+2^h}$, and $s_{j,h}$, and compute its gradient, which now consists of three partial derivatives; we then bound the $\ell_1$-norm of the gradient over the possible values for $s_{j,h}, r_j, r_{j+2^h}$:
  \begin{align*}
      \left|\frac{\partial a_{j+2^h}}{\partial r_{j+2^h}}\right| &=  \left|\frac{r_{j+2^h}}{r_j}\right| \le 1 + \frac{3}{16m} \le \frac{3}{2}, \\
       \left|\frac{\partial a_{j+2^h}}{\partial s_{j,h}}\right| &=  \left|-\frac{s_{j,h}}{r_j}\right| \le 3,\\
\left|\frac{\partial a_{j+2^h}}{\partial r_j}\right| &=  \left|\frac{1}{2} + \frac{s_{j,h}^2 - r_{j+2^h}^2}{4r_j^2}\right| \leq \frac{1}{2}+\frac{1}{4} \left|\frac{s_{j,h}}{r_j}\right|^2 +\frac{1}{4} \left|\frac{r_{j+2^h}}{r_j}\right|^2\le \frac{7}{2}.
  \end{align*}
  For the first inequality above, we used the fact that $r_j,r_{j+2^h}\geq \eps/3$ and $r_{j}+\eps/(16m) \geq r_{j+2^h}$. The middle inequality follows from $s_{j,h} \leq r_j+r_{j+2^h} \leq 2r_j  +\eps/(16m) \le 3r_j$. Hence the gradient's $\ell_1$-norm is upper-bounded by $8$, implying that our estimate for $a_{j+2^h}$ is $\eps/(2m)$-close. Repeating this argument using $t_{j,h}$ shows how to obtain an $\eps/(2m)$-estimate of $b_{j+2^h}$, resulting in an $\eps/m$-estimate of $\alpha_{j+2^h}$. Recall that so far we assumed that $\alpha_j \in \R_{\geq 0}$; this implies that we learned $\alpha_{j+2^h}$ only up to the phase of $\alpha_j$\footnote{It is possible to fully restate our argument without this assumption in the first place, but the resulting derivation is considerably longer and less elegant.}, and we have to reconcile the different relative phases for pairs of amplitudes computed with this procedure.
   
  Thus, we now combine our estimates for the different values of $h$, to learn the entire state up to a global phase. We arbitrarily assume that one amplitude is real, say $\alpha_0\in \R_{\geq 0}$. Consider some index $j$ with Hamming weight $w$, i.e., there are $w$ positions in the binary representation of $j$ that are $1$. We start at the all-zero string, and, proceeding from the most significant bit, flip bits to obtain $j$. This yields a path of length $w$ over indices $j'$, with $w \leq m$. For all $j'$ we have $\tilde{r}_{j'} \geq \eps/2$, hence we know all amplitudes $\alpha_{j'}$ up to the phase of the previous amplitude in the path, with precision $\eps/m$ each. By the triangle inequality and the union bound we can therefore estimate $\alpha_j$ up to the phase of $\alpha_0$, with precision $\eps$ for all $j<k$. Since $k\leq d$ and hence $m\leq \log(d)$, we get the stated complexity.
\end{proof}

We can also use the above result to derive algorithms that estimate the vector $\alpha$ up to different norms, just like in the previous subsections. This results in the following theorem.

\begin{theorem}\label{thm:uptophase}
  Let $0 < \delta < 1$, $\eps > 0$, $d \in \N$, $q \in [2,\infty]$, and let $\ket{\psi} = \sum_{j \in [d]} \alpha_j \ket{j}$ be a quantum state with $\alpha_j \in \C$. Then,
  \[
    O\trm{\min\left\{\trm{\frac{3}{\eps}}^{\frac{1}{\frac12 - \frac1q}}, \frac{d^{\frac2q}}{\eps^2}\right\} \cdot \log d \log\frac{d}{\delta}}
  \]
  copies of $\ket{\psi}$, with the ability to perform unitary operations on each copy, suffice to find an $\eps$-$\ell_q$-norm estimate $\tilde{\alpha}$ of $\alpha$, up to global phase, with success probability at least $1-\delta$.
\end{theorem}

\begin{proof}
  The result follows immediately from combining \autoref{cor:norm-conversion} and \autoref{prop:state-reconstruction}, in exactly the same way as in the proof of \autoref{thm:state-sampling}.
\end{proof}

\section{Pure-state tomography using phase estimation}\label{s:quantumtomo}

In this section we turn to the strongest input model of this paper, where we have access to a state-preparation unitary and its inverse. This allows us to reduce the dependence on the error parameter $\eps$ from $1/\eps^2$ to $1/\eps$. We rely on the framework introduced in \autoref{subsec:gradient}.

\subsection{State preparation for amplitude encoding}
Let $x \in [-1,1]^d$, so $\nrm{x}_2\leq \sqrt{d}$. Define \[
\ket{\amp(x)} := \frac{1}{\sqrt{d}}\sum_{j \in [d]} x_j \ket{j}\ket{0}+ \frac{1}{\sqrt{d}}\sum_{j \in [d]} \sqrt{1-x_j^2} \ket{j}\ket{1}.
\]
In this section we give a simple subroutine that as input takes a binary description of $x$ and constructs the state $\ket{\amp(x)}$.
\begin{lemma}
    \label{lem:uamp}
    Let $\ket{x}$ be a binary encoding of an $x \in [-1,1]^d$ where each $x_j$ can be written exactly with $b$ bits of precision. There is a quantum algorithm $U_{amp}$ that acts as
    \[
    U_{amp}\ket{x}\ket{0} = \ket{x}\ket{\amp(\tilde{x})}
    \]
    where $\nrm{\tilde{x}-x}_{\infty}\leq \eps$. $U_{amp}$ uses $\bigO{\log(d)+\log(1/\eps)\log^2\log(1/\eps)}$ gates, and $2 \min\{b, \log(2/\eps) \}$ indexed-SWAP gates acting on $d$ bits.
\end{lemma}
\begin{proof}
The algorithm is as follows, starting from $\ket{x}\ket{0}\ket{0}\ket{0}\ket{0} = \ket{x_1}\dots\ket{x_d}\ket{0}\ket{0}\ket{0}\ket{0}$
\begin{itemize}
    \item Use $\bigO{\log(d)}$ gates to setup a uniform superposition over $[d]$:
    \[
        \ket{x}\frac{1}{\sqrt{d}}\sum_{i=1}^d \ket{i}\ket{0}\ket{0}\ket{0}
    \]
   \item Swap in the first $\min\{b,\log(1/\eps)\}$ bits of $x_i$, conditioned on the 3th to last register, using  $\min\{b,\log(2/\eps)\}$ indexed-SWAP gates on $d$ bits each:
    \[
        \frac{1}{\sqrt{d}}\sum_{i=1}^d   \ket{\underline{x}^{(i)}} \ket{i}\ket{\bar{x}_i}\ket{0}\ket{0}
    \]
    where $\bar{x}_i$ is the cut-off version of $x_i$, so $|x_i - \bar{x}_i|\leq \eps/2$ and $\underline{x}^{(i)}$ is the remaining part of $x$.
    \item Approximate $\tilde{a}_i = \arcsin(\bar{x}_i)$ up to $\eps/2$ precision using $\bigO{\log(1/\eps)\log^2\log(1/\eps)}$ gates (see \cite{wikiFunctions} for the complexity):
    \[
      \frac{1}{\sqrt{d}}\sum_{i=1}^d   \ket{\underline{x}^{(i)}}  \ket{i}\ket{\bar{x}_i}\ket{\tilde{a}_i}\ket{0}
    \]
    \item Use  $\bigO{\log(1/\eps)}$ rotations with exponentially decreasing angle, controlled on the bits of $\tilde{a}_i$ to rotate the last qubit:
    \[
         \frac{1}{\sqrt{d}}\sum_{i=1}^d   \ket{\underline{x}^{(i)}}  \ket{i}\ket{\bar{x}_i}\ket{\tilde{a}_i}\left(\tilde{a}_i\ket{0}+\sqrt{1-\tilde{a}_i^2}\ket{1}\right)
    \]
    \item Uncompute $\tilde{a}_i$ and swap back $\bar{x}_j$:
        \[
        \ket{x} \frac{1}{\sqrt{d}}\sum_{i=1}^d  \ket{i}\left(\sin(\tilde{a}_i)\ket{0}+\sqrt{1-\sin(\tilde{a}_i)^2}\ket{1}\right)
    \]
\end{itemize}
As $|\tilde{a}_i-\arcsin(\bar{x}_j)|\leq \eps/2$, and the $\sin$ function is $1$-Lipschitz, we get that $|\bar{x}_j - \sin(\tilde{a}_j)|\leq \eps/2$. Combining this with $|x_j-\bar{x}_j|\leq \eps/2$ gives the desired precision.
\end{proof}

\subsection{Pure-state tomography}

With the quantum circuit of \autoref{lem:uamp} for preparing $\ket{\amp(\tilde{x})}$ we have all the ingredients for our pure-state tomography algorithm. We only state it for the estimation of the real part of $\ket{\psi}$, but one can also extract the imaginary part with the same running time simply by applying the algorithm to the quantum state $\ci\ket{\psi}$.
\begin{proposition}
\label{prop:quantum_infty_norm_est}
Let $\ket{\psi} = \sum_{j \in [d]} \alpha_j \ket{j}$ be a quantum state, and $U\ket{0} = \ket{\psi}$. There is a quantum algorithm that, with probability at least $1-\delta$, outputs $\tilde{\alpha} \in \R^{d}$ such that $\nrm{ \Re(\alpha) - \tilde{\alpha}}_\infty \le \eps$, using 
\[
\bigO{\frac{\sqrt{d}}{\eps}\log(\frac{d}{\delta})}
\]
applications of $U$ and $U^\dagger$, $\bigO{\frac{\sqrt{d}}{\eps}\log (\frac{d}{\delta})\log(\frac{d}{\eps})}$ indexed-SWAP gates acting on $d$ bits, and $\bigOt{d + \frac{\sqrt{d}}{\eps}}$ additional gates.  If $\eps \ge \frac{1}{\sqrt{d}}$, the number of applications of $U$ can be reduced to $\tilde{O}(\frac{1}{\eps^2})$ (while potentially increasing the gate complexity to $\tilde{O}(\frac{1}{\eps^4})$).
\end{proposition}
\begin{proof}
We first describe the algorithm that calls $U$ $\tilde{O}(\frac{\sqrt{d}}{\eps})$ times. Let $f(x) = \ipc{\Re(\alpha)}{x}$. Taking $U':=(I\otimes U)$ and $V:=U_{amp}$, \autoref{lem:blockInnerProd} gives us an approximate block-encoding $W$ of the diagonal matrix $\ipc{x}{\alpha}/\sqrt{d}$, and from this we get an approximate block-encoding of $f(x)/\sqrt{d}$ averaging $W$ and $W^\dagger$ via \autoref{lem:linCombBlocks} so we can apply \autoref{cor:blockToGrad}. Then we get an $\eps$-$\ell_\infty$-norm estimate of $\Re(\alpha)/\sqrt{d}$ (where the denominator $\sqrt{d}$ comes from the normalization in $\ket{\amp(x)}$) with $\bigO{\frac{1}{\eps}\log (\frac{d}{\delta})}$ uses of the block-encoding of $(W+W^{\dagger})/2$, the construction of which requires a constant number of calls to $U$ and $U_{amp}$. To obtain the desired estimate of $\Re(\alpha)$ we elevate the precision to $\eps/\sqrt{d}$, which brings the total number of uses of $U$ and $U_{amp}$ to $\bigO{\frac{\sqrt{d}}{\eps}\log (\frac{d}{\delta})}$. The gate complexity
is $\bigO{\frac{\sqrt{d}}{\eps}\log (\frac{d}{\delta})\log (\frac{d}{\eps})}$ indexed-SWAP gates acting on $d$ bits, and $$\bigO{\left(d \log (\frac{1}{\eps}) \log \log(\frac{1}{\eps}) + \frac{\sqrt{d}}{\eps} \log \frac{d}{\eps}\log d\right)\log \frac{d}{\delta}}$$
additional gates, where the first term in the summation comes from the additional gates of \autoref{cor:blockToGrad}, whereas the second term comes from the cost of $U_{amp}$.

Next, we show how to improve the algorithm when $\eps \ge \frac{1}{\sqrt{d}}$. Any $j$ such that $|\alpha_j| \le \eps$ can be ignored because of the $\ell_\infty$-norm objective, so we can simply set $\tilde{\alpha}_j = 0$. Since we are only interested in $j$ such that $|\alpha_j| > \eps$, we note that there are at most $1/\eps^2$ such $j$ because $\sum_j |\alpha_j|^2 = 1$. After taking $O(\frac{\log(n/\delta)}{\eps^2})$ measurements of $\ket{\psi}$ in the computational basis, the probability that all such $j$ are observed is at least $1-\delta$. We can then apply the algorithm described above to obtain $\alpha_j$ only for those $j$. As this set has cardinality $O(1/\eps^2)$, the quantum algorithm requires $\tilde{O}(\frac{1}{\eps^2})$ applications of $U$. This concludes the proof.
\end{proof}
Note that the algorithm in \autoref{prop:quantum_infty_norm_est} with complexity $\bigOt{\frac{\sqrt{d}}{\eps}}$ can be made essentially unbiased by using \autoref{cor:unbiasedBlockToGrad} instead of \autoref{cor:blockToGrad}.
\begin{theorem}
\label{thm:quantum_euclidean_norm_est}
Let $\ket{\psi} = \sum_{j \in [d]} \alpha_j \ket{j}$ be a quantum state, $\alpha \in \C^d$ the vector with elements $\alpha_j$, and $U\ket{0} = \ket{\psi}$. Then, for $q \ge 2$
\begin{equation*}
    O\trm{\min\left\{ \trm{\frac{3}{\eps}}^{\frac{1}{\frac12-\frac1q}},\frac{d^{\frac12+\frac1q}}{\eps}\right\} \log \frac{d}{\delta}}
\end{equation*}
conditional applications of $U$ and its inverse suffice to compute an $\eps$-$\ell_q$-norm estimate $\tilde{\alpha}$ of $\alpha$, with success probability at least $1-\delta$. 
\end{theorem}
\begin{proof}
The first term follows from \autoref{prop:quantum_infty_norm_est} and \autoref{cor:norm-conversion}. The second term follows from \autoref{thm:condsamp}.
\end{proof}

\subsection{Tomography for sparse vectors}
To conclude this section, we show that the tomography algorithm based on phase estimation can be improved if we know that the quantum state contains at most $s < d$ large amplitudes. We proceed by finding the large elements first, then applying the tomography algorithm only to extract a description of only the corresponding part of the quantum state.

\begin{proposition}
\label{prop:sparse_ell_infty_norm}
Let $\ket{\psi} = \sum_{j \in [d]} \alpha_j \ket{j}$ be a quantum state, and $U\ket{0} = \ket{\psi}$. Let $0 < \delta < 1$, and let $R$ be such that $\{j \in [d] : |\alpha_j| \geq \eps\}\subseteq R  \subseteq [d]$ be a subset of the indices that contains all large elements.
Let $P:= \sum_{i\not\in R} |\alpha_i|^2$  and $s = |R|$. There is a quantum algorithm that, with probability at least $1-\delta$, outputs an $\bigO{(s+P/\eps^2)\log(s)\log(1/\delta)}$-sparse $\tilde{\alpha} \in \R^{d}$ such that $\nrm{ \alpha - \tilde{\alpha}}_\infty \le \eps$ using \[
\bigO{ \trm{\trm{\frac{\sqrt{s}}{\eps} +\frac{\sqrt{P}}{\eps^2}} + \log\frac{\log(s + P/\eps)}{\delta} } \log(s)\log\trm{\frac{s+P/\eps}{\delta}}}
\]
applications of $U$ and its inverse, $\bigO{ \trm{\trm{\frac{\sqrt{s}}{\eps} +\frac{\sqrt{P}}{\eps^2}} + \log \frac{\log(\frac{s + P}{\eps})}{\delta} } \log(s)\log\trm{\frac{s+P/\eps}{\delta}}\log\trm{\frac{s+P}{\eps}}}$ indexed-SWAP gates acting on $\bigO{\left(s+\frac{P}{\eps^2}\right)\log(s)\ln(1/\delta)}$ bits, and $\bigOt{s+ \frac{\sqrt{s}}{\eps} + \frac{P}{\eps^2} }$ additional gates.
\end{proposition}
\begin{proof}
  We start by finding all (at most $s$) elements in $R$ that are at least $\eps$ in size using amplitude amplification. We can then ignore all other elements and apply our state tomography algorithm on the relevant elements.

  Let $k = |\{j \in [d]: |\alpha_j|\geq \eps\}|$ be the number of large elements. So there are $s-k$ elements in $R$ that are smaller than $\eps$. We start by simply measuring the state and observing an index $j$. Note that with probability at least $k\eps^2$ this is one of the relevant entries, although this is unknown to us. After observing a single entry, we mark all other entries as ``good'' and amplify the ``good'' part of the state before measuring again, to avoid seeing an element twice. We repeat this until we have seen $T$ different elements, for some $T$ to be determined later.

  If at some point we have seen $j$ large elements (which we do not know), then the probability on the ``good'' elements is at most $(s-j)\eps^2+P$. Hence, this part can be amplified to find a new element with probability $\geq 2/3$ using $\bigO{\frac{1}{\sqrt{(s-j)\eps^2+P}}}$ queries. The probability that this new element is one of the large ones is at least $\frac{2(k-j)\eps^2}{3 (s-j)\eps^2+P}$. Thus, the expected number of samples before seeing a new large element is at most $\frac{3 (s-j)\eps^2+P}{2(k-j)\eps^2}$. For the expected number of queries needed before seeing all large elements we then get
  \begin{align*}
    \bigO{  \sum_{j=0}^{k-1} \frac{1}{\sqrt{(s-j)\eps^2+P}} \frac{ (s-j)\eps^2+P}{(k-j)\eps^2} } &=
    \bigO{  \sum_{j=1}^{k} \left( \frac{\sqrt{(s-k+j)\eps^2}}{j\eps^2}  +  \frac{\sqrt{P}}{j\eps^2}\right) }\\
    &= \bigO{ \left(\frac{\sqrt{s}}{\eps}+\frac{\sqrt{P}}{\eps^2}\right)\log(s)}.
  \end{align*}
  By Markov's equality we can stop the algorithm after $6$ times the expected number of queries and still be successful with probability $\geq 5/6$.

  We also need to ensure that we do not return too many elements. The expected number of elements found is equal to the expected number of samples, hence it is 
  \[
     \sum_{j=0}^{k-1}  \frac{ (s-j)\eps^2+P}{(k-j)\eps^2}  \leq  \left(s+\frac{P}{\eps^2}\right)\log(s).
  \]
  Again, by Markov's inequality we can stop the algorithm if we see more then $6$ times this number of samples, and still succeed with probability at least $5/6$. By the union bound both conditions are met with probability at least $2/3$. 
  Repeating $\bigO{\ln(1/\delta)}$ times and taking all elements seen in runs with not too many samples gives us as success probability at least $1-\delta/2$. The query complexity then becomes \[
\bigO{ \left(\frac{\sqrt{s}}{\eps}+\frac{\sqrt{P}}{\eps^2}\right)\log(s)\log(1/\delta)}
\]
  for this entire procedure.
  
  We now assume the last step was successful, so we have a set $I$ of indices such that $I$ contains all large elements and $|I| = \bigO{\left(s+\frac{P}{\eps^2}\right)\log(s)\ln(1/\delta)}$.
  Applying \autoref{thm:quantum_euclidean_norm_est} on just these indices gives the query and gate complexity from the lemma.
  
\end{proof}

The above proposition gives an improvement if the state is close to a sparse state in $\ell_2$-norm. We can directly get a bound on this closeness if almost all elements are small.
\begin{corollary}
\label{cor:quantum_sparse_infinity_norm}
Let $\ket{\psi} = \sum_{j \in [d]} \alpha_j \ket{j}$ be a quantum state, and $U\ket{0} = \ket{\psi}$. Let $0 < \delta < 1$, and let $s$ be such that $|\{j \in [d] :|\alpha_j| \geq \eps \sqrt{\frac{s}{d}}\}|\leq s$. There is a quantum algorithm that, with probability at least $1-\delta$, outputs an $\bigO{s\log(s)\log(1/\delta)}$-sparse $\tilde{\alpha} \in \R^{d}$ such that $\nrm{ \alpha - \tilde{\alpha}}_\infty \le \eps$ using \[
\bigO{ \trm{\frac{\sqrt{s}}{\eps} + \log \frac{\log s}{\delta}} \log(s)\log\trm{\frac{s}{\delta}}}
\]
applications of $U$ and its inverse, and $...$ additional gates. 
  \end{corollary}
\begin{proof}
  Let $R = \{j \in [d]:|\alpha_j| \geq \eps \sqrt{\frac{s}{d}}\}$. Then, $P = \sum_{i\not \in R} |\alpha_i|^2 \leq d \eps^2 \frac{s}{d} \leq \eps^2 s$. Applying \autoref{prop:sparse_ell_infty_norm} with this choice of $R$ and the above bound on $P$ gives the desired result.
 \end{proof}

As usual, we now convert the above bound for the $\ell_\infty$-norm to other norms.

\begin{theorem}
\label{thm:quantum_euclidean_norm_est_sparse}
Let $\ket{\psi} = \sum_{j \in [d]} \alpha_j \ket{j}$ be a quantum state, $\alpha \in \C^d$ the vector with elements $\alpha_j$, and $U\ket{0} = \ket{\psi}$. Let $0 < \delta < 1$, let $q \ge 2$ and let $s$ be such that $|\{j \in [d] : |\alpha_j| \ge \eps \sqrt{\frac{s}{d}}\}| \le s$. Then we can compute an $\eps$-$\ell_q$-norm estimate $\tilde{\alpha}$ of $\alpha$ using
\begin{equation*}
    O\trm{\trm{\min\left\{ \trm{\frac{3}{\eps}}^{\frac{1}{\frac12-\frac1q}},\frac{s^{\frac12+\frac1q}}{\eps}\right\} + \log \frac{\log s}{\delta}}\log s \log \frac{s}{\delta}}
\end{equation*}
conditional applications of $U$ and its inverse, with success probability at least $1-\delta$. 
\end{theorem}
\begin{proof}
Follows from \autoref{prop:sparse_ell_infty_norm} and \autoref{cor:quantum_sparse_infinity_norm}.
\end{proof}

\section{First intermezzo: unbiased phase estimation}
\label{sec:unbiasedPhaseEst}
We now describe a method for phase estimation that is unbiased, more precisely symmetric in the sense that for a phase $\phi$ it provides an estimate $\varphi$ such that the probability of getting estimate $\phi+\epsilon$ is the same as getting estimate $\phi-\epsilon$ (modulo $2\pi$) for all $\eps$. Note that this is not satisfied by ordinary phase estimation, but this property is highly desirable, as we showcase in our applications. In particular, we need unbiased phase estimation to recover unbiased estimates of the entries of a density matrix, allowing us to give tighter error bounds with high probability.

Our method is based on adding and later subtracting a random phase shift; this idea is not new, see, e.g., \cite[Section 3.2]{linden2021AvgCaseQFTtoWorstCaseQFT}. The first step in our analysis is to show that the resulting estimator is symmetric. We subsequently show how to boost the precision of this symmetric estimator in a symmetric way. Since the problem is invariant under shifting by $2\pi$ we can always interpret phases $\phi,\varphi$ modulo $2\pi$; in particular for phases we define the distance modulo $2\pi$ introducing the notation $|\phi-\varphi|_{2\pi}:=\min\{|\phi-\varphi-2\pi\ell| \colon \ell\in \Z \}$. In this section, ``digit'' always refers to ``binary digit'', i.e., all numbers are expressed in fixed-point binary encoding; for example, $0.b_1 b_2\ldots b_n$ where $b_i\in\{0,1\}$ for $i\in[n]$ is the $n$-digit encoding of $b$. Recall that the function $\sinc(x)$ is a complex entire function defined as $\sin(x)/x$ for $x\neq 0$ and $\sinc(0)=1$.

\begin{algorithm}[!ht]
	\caption{Suppressed-Bias Phase Estimation}\label{alg:unbiased}
	\begin{algorithmic}[1] 
		\STATEx	{\bf Input:} $\ket{\psi(\phi)}=\frac{1}{\sqrt{M}}\sum_{k=0}^{M-1}e^{i \phi k}\ket{k}$ (for unknown $\phi$), and a parameter $n\in \N$
		\STATE Sample a uniformly random $n$-digit binary number $u\in[0,1)$ and define $\xi:=\frac{2\pi u}{M}$ \label{line:randomXi}
		\STATE Apply multi-phase gate $\sum_{k=0}^{M-1}e^{-i\xi k}\ketbra{k}{k}$ to $\ket{\psi(\phi)}$ \label{line:applyPhase}
		\STATE Perform inverse Fourier transform over $\mathbb{Z}_M$ and measure the state, yielding outcome $j$
		\STATE	{\bf Return} $\varphi:=\frac{2\pi j}{M}+\xi=\frac{2\pi}{M}(j+u)$ \label{line:correctPhase} 
	\end{algorithmic}
\end{algorithm}	
\begin{theorem}[Unbiased Phase Estimation]\label{thm:unbiased}
	If we run \autoref{alg:unbiased} with $n=\infty$ in Line~\ref{line:randomXi}, then it returns a random phase $\varphi\in[0,2\pi)$ with probability density function
	\begin{equation} \label{eq:probDens}
	f(\varphi):=\frac{M}{2\pi}\frac{\sinc^2(\frac{M}{2}|\phi-\varphi|_{2\pi})}{\sinc^2(\frac{1}{2}|\phi-\varphi|_{2\pi})}.
	\end{equation}
\end{theorem}
This probability density function is normalized so that $\int_{0}^{2\pi} f(\varphi) d\varphi=1$,
moreover it only depends on $|\phi-\varphi|_{2\pi}$ showing that this procedure satisfies our criterion for unbiasedness, see \autoref{fig:densityPlot}.
\begin{figure}[ht]
	\begin{center}
		\begin{tikzpicture}[scale=1.]
		\begin{axis}[xlabel={$x$}, axis lines=middle, samples=200, domain=0:
		pi, ymin=0, ymax=2.65,
		thick, trig format plots=rad,no markers,
		xtick = {-pi,-0.5*pi,0,0.5*pi,pi},
		xticklabels = {$-\pi$,$-\frac{\pi}{2}$,$0$,$\frac{\pi}{2}$,$\pi$},
		ytick = {1,2},
		yticklabels = {$1$,$2$},
		]
		\addplot[blue,domain=-pi:pi] {8/pi*sin(8*x)*sin(8*x)/256/sin(x/2)/sin(x/2)};				
		\end{axis}
		\end{tikzpicture}
	\end{center}
	\caption{
		Plot of \autoref{eq:probDens} for $x=\phi-\varphi$ and $M=16$.
	}\label{fig:densityPlot}
\end{figure}
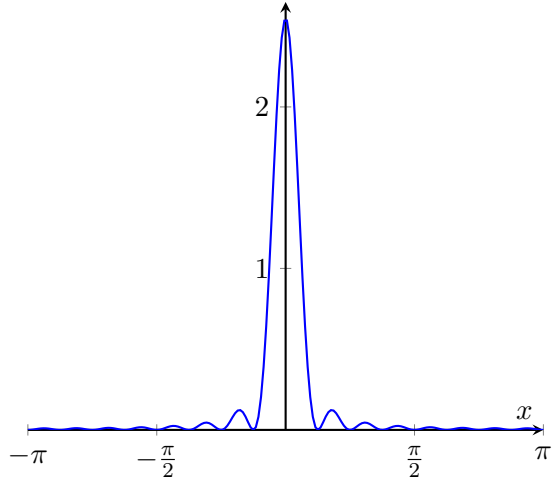
\begin{proof}
Suppose we have a quantum state $\frac{1}{\sqrt{M}}\sum_{k=0}^{M-1}e^{i \phi k}\ket{k}$ and we wish to estimate the phase $\phi$. Then by applying the inverse quantum Fourier transform over $\mathbb{Z}_M$ and measuring we get outcome $j$, giving rise to estimate $\varphi=\omega_j=\frac{2\pi j}{M}$ with probability (using \cite[Lemma 10]{brassard2002AmpAndEst}):
\begin{align} \label{eq:shiftProbDens}
 \left|\frac{1}{M}\sum_{k=0}^{M-1}e^{-i \omega_j k}e^{i\phi k}\right|^2
=\left|\frac{1}{M}\sum_{k=0}^{M-1}e^{i(\phi-\omega_j) k}\right|^2
=\left|\frac{1}{M}\sum_{k=0}^{M-1}e^{i|\phi-\omega_j| k}\right|^2
=\frac{\sinc^2(\frac{M}{2}|\phi-\omega_j|_{2\pi})}{\sinc^2(\frac{1}{2}|\phi-\omega_j|_{2\pi})}.
\end{align}
Now let us modify this procedure by first choosing a uniformly random phase $\xi\in[0,\frac{2\pi}{M})$ and applying phase estimation to the state $\frac{1}{\sqrt{M}}\sum_{k=0}^{M-1}e^{i (\phi k-\xi k)}\ket{k}$ then outputting $\varphi=\omega_j+\xi$ for the resulting $j$. Then, for a fixed $\xi$ the probability of outputting $\varphi=\omega_j+\xi$ is 
\begin{equation} \label{eq:preProbDens}
\frac{\sinc^2(\frac{M}{2}|\phi-\xi-\omega_j|_{2\pi})}{\sinc^2(\frac{1}{2}|\phi-\xi-\omega_j|_{2\pi})}
=\frac{\sinc^2(\frac{M}{2}|\phi-\varphi|_{2\pi})}{\sinc^2(\frac{1}{2}|\phi-\varphi|_{2\pi})}.
\end{equation}
Since the choice of $\xi$ is uniformly random over $[0,\frac{2\pi}{M})$ this implies that the probability density function of getting estimate $\varphi\in[0,2\pi)$ is 
given by \autoref{eq:probDens}.
\end{proof}

According to \autoref{eq:probDens} the probability of getting an outcome $\varphi$ with error at most $\frac{c}{M}$ for some $c\leq \pi M$ is
\begin{align}
\Pr[|\phi-\varphi|_{2\pi}\leq \frac{c}{M}]
&=\int_{-\frac{c}{M}}^{\frac{c}{M}}\frac{M}{2\pi}\frac{\sinc^2(M x/2)}{\sinc^2(x/2)} dx \tag*{substitute $y= \frac{M}{2} x \Rightarrow$}\\
&=\frac{1}{\pi}\int_{-\frac{c}{2}}^{\frac{c}{2}}\frac{\sinc^2(y)}{\sinc^2(y/M)} dy \tag*{use $|\sinc(z)|\leq 1 \Rightarrow$}\\
&\geq\frac{1}{\pi} \int_{-\frac{c}{2}}^{\frac{c}{2}}\sinc^2(y) dy.
\end{align}
In particular one can compute the value of this bound for $c=1,2,3$ resulting in:
\begin{align}\label{eq:concreteBounds}
\kern-1mm\Pr[|\phi-\varphi|_{2\pi}\leq \frac{1}{M}]\geq 0.30\ldots & & 
\Pr[|\phi-\varphi|_{2\pi}\leq \frac{2}{M}]\geq 0.57\ldots & & 
\Pr[|\phi-\varphi|_{2\pi}\leq \frac{3}{M}]\geq 0.75\ldots & & 
\end{align}
Note the increased accuracy compared to ordinary phase estimation: the difference between two distinct phase estimates is at least $\frac{2\pi}{M}$, so in case the true phase we try to estimate is, say, $\frac{\pi}{M}$, then ordinary phase estimation always has an error at least $\frac{\pi}{M}>\frac{3}{M}$. On the other hand, here we get $\frac{3}{M}$-accuracy with probability greater than $\frac{3}{4}$.

\subsection{Unbiased boosting}

Now we show how to boost this procedure so that it gives an unbiased estimate that is also $\frac{6}{M}$-accurate with exponentially high probability. We achieve this essentially by the usual median trick, except some care is needed because the median is ill-defined modulo $2\pi$.

\begin{algorithm}[!ht]
	\caption{Boosted Unbiased Phase Estimation}\label{alg:boostedUnbiased}
	\begin{algorithmic}[1] 
		\STATEx	{\bf Input:} $(2m+1)$ copies of $\ket{\psi(\phi)}=\frac{1}{\sqrt{M}}\sum_{k=0}^{M-1}e^{i \phi k}\ket{k}$ (for unknown $\phi$)
		\STATE {\bf For} $j=1$ to $2m+1$
		\State  ~~~ Run \autoref{alg:unbiased} setting $n=\infty$ on the $j$-th copy of $\ket{\psi(\phi)}$ and record the estimate $\varphi_j$
		\STATE 	Find the shortest interval $I=[a,b]\subseteq [-2\pi,2\pi]$ such that $I\cup (I+2\pi)$ contains at least $m+1$ of the estimates $\varphi_j$  \label{line:shortest}
		\STATE 	{\bf If} $a+b\geq 0$ {\bf then return} $\overline{\varphi}:=\frac{a+b}{2}$ {\bf else return} $\overline{\varphi}:=\frac{a+b}{2}+2\pi$					
	\end{algorithmic}
\end{algorithm}	
\begin{theorem}\label{thm:boostedUnbiased}
		\autoref{alg:boostedUnbiased} returns an unbiased $\overline{\varphi}\in[0,2\pi]$ s.t. $\Pr[|\phi-\overline{\varphi}|_{2\pi}\leq \frac{6}{M}]\geq 1-\exp(-\frac{m}{4})$.
\end{theorem}
\begin{proof}
	In Line~\ref{line:shortest} almost surely there is a unique shortest interval (modulo $2\pi$), since the endpoints $a,b$ of the shortest interval must come from the $2m+1$ estimates (modulo $2\pi$), which themselves come from the continuous distribution of \autoref{eq:probDens}. Alternatively, if there are multiple shortest intervals (modulo $2\pi$) we can just choose one uniformly at random. This algorithm is naturally unbiased as the distribution of the shortest intervals (modulo $2\pi$) is symmetric with respect to $\phi$. 

	Moreover, the probability that there are at least $m+1$ estimates $\varphi_j$ such that $|\phi-\varphi_j|_{2\pi}\leq \frac{3}{M}$ is at least $1-\exp(-\frac{m}{4})$ due to the Chernoff bound. Indeed, the probability that $|\phi-\varphi|_{2\pi}> \frac{3}{M}$ is at most $\frac{1}{4}$ by \autoref{eq:concreteBounds}. So by the Chernoff-Hoeffding theorem~\cite [Theorem 1]{hoeffding1963ProbIneqSumsOfBoundedRVs} the probability that $|\phi-\varphi|_{2\pi}> \frac{3}{M}$ holds for at least $m+1$ out of $2m+1$ estimates is at most $\exp(-D(\frac12\Vert \frac14)(2m+1))\leq\exp(-\frac{1}{8}(2m+1))\leq \exp(-\frac{m}{4})$, where $D(x\Vert y)=x\ln(\frac{x}{y})+(1-x)\ln(\frac{1-x}{1-y})$.
	This implies that the shortest interval has length at most $\frac{6}{M}$ and it also must overlap with the interval $[\phi-\frac{3}{M},\phi+\frac{3}{M}]$ (modulo $2\pi$), so in particular $|\phi-\overline{\varphi}|_{2\pi}\leq \frac{6}{M}$.
\end{proof}

\subsection{Unbiased estimators of \texorpdfstring{$e^{i\phi}$}{exp(i phi)}}\label{subsec:unbiasedComplex}

Our unbiased phase estimators can be used for constructing unbiased estimators of the complex number $e^{i\phi}$. The unbiased nature of our phase estimates $\varphi$ means that $\mathbb{E}[e^{i\varphi}]=\lambda e^{i\phi}$, for some $\lambda\in [-1,1]$. \pnote{Note that aprioir it could be computationally difficult to compute the value of $\lambda$ to sufficient precision.} Moreover, due to the shift invariance of $f(\varphi)$, i.e., the fact that $f(\varphi)$ depends only on $|\phi-\varphi|_{2\pi}$, we have that $\lambda$ only depends on $M$ (and $m$ in the boosted case). Therefore $e^{i\varphi}/\lambda$ is an unbiased estimator of $e^{i\phi}$. One can also compute the value
\begin{align*}
\lambda(M) = \int_{-\pi}^\pi \cos(x)\cdot\frac{M}{2\pi}\frac{\sinc^2(Mx/2)}{\sinc^2(x/2)}dx
&=\int_{-\pi}^\pi(1-2\sin^2(x/2))\cdot\frac{1}{2M\pi}\frac{\sin^2(Mx/2)}{\sin^2(x/2)}dx\\
&=1-\frac{1}{M\pi}\int_{-\pi}^\pi\sin^2(Mx/2)dx\\
&=1-\frac{1}{M}.
\end{align*}
To compute the variance it is useful to note that for $M\geq 2$
\begin{align*}
\int_{-\pi}^\pi \sin^2(x)\cdot\frac{M}{2\pi}\frac{\sinc^2(Mx/2)}{\sinc^2(x/2)}dx
&=\int_{-\pi}^\pi4\sin^2(x/2)\cos^2(x/2)\cdot\frac{1}{2M\pi}\frac{\sin^2(Mx/2)}{\sin^2(x/2)}dx\\
&= \frac{2}{M\pi}\int_{-\pi}^\pi\cos^2(x/2)\sin^2(Mx/2)dx\\
&=\frac{1}{M}.
\end{align*}
Thus we have that $(1+\frac{1}{M-1})e^{i\varphi}$ is an unbiased estimator of $e^{i\phi}$ with variance $\Theta(\frac{1}{M})$.
	
For the boosted version it is harder to compute the value of $\lambda(M,m)$, but due to the concentration proven in \autoref{thm:boostedUnbiased} we know that its value must be $\lambda(M,m)=1-\bigO{\frac{1}{M^2}+\exp(-\frac{m}{4})}$. So $e^{i\overline{\varphi}}/\lambda(M,m)$ is an unbiased estimator of $e^{i\phi}$ with variance $\bigO{\frac{1}{M^2}+\exp(-\frac{m}{4})}$.

\subsection{Unbiased probability estimation}

Unbiased estimators for $\lambda e^{i\phi}$ give us the possibility of modifying the standard amplitude estimation algorithm~\cite{brassard2002AmpAndEst}, so that we estimate the squared amplitude without bias. To that end, suppose that we have access to a state-preparation unitary that prepares the state $\sqrt{1-p}\ket{\psi_0}\ket{0} + \sqrt{p}\ket{\psi_1}\ket{1}$, and our goal is to estimate $p$. Recall that the amplitude estimation algorithm runs phase estimation on the Grover iterate, which is a $2$-dimensional rotation with eigenvalues $e^{\pm2i\theta}$, where $\theta = \arcsin\sqrt{p}$. Consequently, it obtains an estimate for $2\theta$ or $-2\theta$, both with probability $1/2$.

If we now substitute our unbiased phase estimation algorithm into this procedure, we obtain an unbiased estimate of either $\lambda e^{2i\theta}$ or $\lambda e^{-2i\theta}$, both with probability $1/2$. In either case, taking the real part of our estimate now estimates $\lambda \cos(2\theta) = \lambda\cos(2\arcsin\sqrt{p}) = \lambda(1 - 2p)$ without bias. Thus, if we denote the outcome of the unbiased phase estimation algorithm by $Z = e^{2\pi i\varphi}$, then
\[
    \mathbb{E}\left[\frac12 - \frac{\mathrm{Re}[Z]}{2\lambda}\right] = \frac12 - \frac{\lambda(1-2p)}{2\lambda} = p.
\]

We can crudely bound the variance of this estimator to be $\mathrm{Var}[\mathrm{Re}[Z]]/(2\lambda)^2 \leq \mathrm{Var}[Z]/(2\lambda)^2 = \mathcal{O}(1/M^2 + \exp(-m/4))$. However, if $p$ is very close to $0$ or $1$, then the probability distribution of $Z$ will be very tightly concentrated around $1$ or $-1$ on the unit circle in the complex plane, where the unit circle runs perpendicular to the real axis. Thus, in this regime taking the real part of $Z$ intuitively squashes samples much closer together, and as a result the variance of $\mathrm{Re}[Z]$ can be much smaller than that of $Z$.

Quantitatively, if $\theta \leq 3/(2M)$, then the endpoints of the interval of concentration for $\mathrm{Re}[Z]$, as derived in \autoref{thm:boostedUnbiased}, are $\cos(2\theta + 3/M)$ and $1$, which means that the length of the interval is $1 - \cos(2\theta + 3/M) = \mathcal{O}((\theta + 1/M)^2) = \mathcal{O}(1/M^2)$. Thus, the variance in this case is $\mathcal{O}(1/M^4 + \exp(-m/4))$. A similar analysis holds true in the case where $\theta \geq \pi/2 - 3/(2M)$.

On the other hand, if $3/(2M) < \theta < \pi/2 - 3/(2M)$, then the endpoints of the interval of concentration for $\mathrm{Re}[Z]$ are $\cos(2\theta + 3/M)$ and $\cos(2\theta - 3/M)$. This implies that the length of the concentration interval is $\cos(2\theta - 3/M) - \cos(2\theta + 3/M) = 2\sin(3/M)\sin(2\theta) \leq 12\sin(\theta)\cos(\theta)/M = \mathcal{O}(\sqrt{p(1-p)}/M)$. Thus, the variance becomes $\mathcal{O}(p(1-p)/M^2 + \exp(-m/4))$.

Putting both cases together, we obtain that the variance for unbiased probability estimation using boosted unbiased phase estimation is $\mathcal{O}(p(1-p)/M^2 + 1/M^4 + \exp(-m/4))$. Note that from this variance bound, one can essentially recover the precision that is obtained by Brassard et al.~\cite{brassard2002AmpAndEst}, up to constant factors. Thus, this way of estimating the probability gives one an unbiased estimator, while maintaining the precision attained by traditional techniques.

Notice that this estimation procedure for the probability might output estimates that are outside the interval $[0,1]$. Indeed, for example if $p=1$ any non-trivial unbiased estimator must eventually produce estimates that are larger than $1$.

Finally, note that computing the value of $\lambda(M,m)$ to high precision might be computationally difficult for larger $m$ values. However, we can compute the value approximately by Monte Carlo simulation. One can generate $2m+1$ samples corresponding to $\phi=0$ using the density function \autoref{eq:probDens}, and run \autoref{alg:boostedUnbiased} finally outputting $\Re(e^{i\overline{\varphi}})$. Clearly $\mathbb{E}[\Re(e^{i\overline{\varphi}})]=\mathbb{E}[e^{i\overline{\varphi}}]=\lambda(M,m)$. 
On the other hand as we have shown above the variance of $\Re(e^{i\overline{\varphi}})$ is $\mathcal{O}(1/M^4 + \exp(-m/4))$. Intuitively speaking this means that the computation of $\lambda(M,m)$ should not prohibit applications of this result, especially considering that $\lambda(M,m)$ can be pre-computed ahead of time.

\subsection{Implementation with finite precision}

\begin{theorem}[Suppressed-Biased Phase Estimation]\label{thm:supbiased}
	If we run \autoref{alg:unbiased} with some finite $n$ in Line~\ref{line:randomXi}, then it returns a random phase $\varphi\in[0,2\pi)$ of the form $\frac{2\pi}{M}\left(j+\frac{\ell}{2^n}\right)$ for some $j\in\{0,1,\ldots, M-1\}$ and $\ell\in\{0,1,2,3,\ldots,2^n-1\}$ such that the distribution of the outcome is $\frac{2\pi}{2^n}$-close in total variation distance to the distribution
	\begin{equation} \label{eq:probDensFinite}
	\Pr\left[\varphi=\frac{2\pi}{M}\left(j+\frac{\ell}{2^n}\right)\right]=\int_{\frac{2\pi}{M}\left(j+\frac{\ell}{2^n}\right)}^{\frac{2\pi}{M}\left(j+\frac{\ell+1}{2^n}\right)}f(x) dx,
	\end{equation}
	where $f(x)$ is defined in \autoref{eq:probDens}.
\end{theorem}
\begin{proof}
	First let us consider running \autoref{alg:unbiased} with $n=\infty$, except in Line~\ref{line:correctPhase} truncating $u$ to have only $n$ binary digits. Then it follows from \autoref{thm:unbiased} that the output distribution is given by \autoref{eq:probDensFinite}.
	
	Next consider further modifying this algorithm by truncating $u$ to $n$ binary digits in Line~\ref{line:applyPhase} as well, bringing us to the finite-$n$ version of \autoref{alg:unbiased}. This introduces a change in the applied unitary with magnitude (in terms of operator norm) no greater than $\frac{2\pi}{2^{n}}$. Thus a perturbation is induced on the state in the algorithm of magnitude (in the $\ell_2$-norm) no greater than $\frac{2\pi}{2^{n}}$, ultimately changing the measurement statistics by no more than $\frac{2\pi}{2^{n}}$ in total variation distance, cf.~\cite[Exercise 4.3]{wolf2019QCLectureNotes}.
\end{proof}

Overall we can conclude that the output distribution of the discretized \autoref{alg:unbiased} gets a perturbation that is at most $\bigO{2^{-n}}$ in the Wasserstein-$1$ distance compared to the infinite-precision version of \autoref{alg:unbiased}.

Similarly, we believe that using the discretized version of \autoref{alg:unbiased} within \autoref{alg:boostedUnbiased} would exponentially suppresses the bias, and the conclusion about boosting should not be affected. However, discretizing the proof appears to be difficult, because it heavily relies on a symmetry argument -- ultimately breaking due to the non-symmetric discretization errors. The main difficulty is that even small perturbations to the $\varphi_j$ values can induce some large jumps in the shortest interval in some edge cases. For this reason below we introduce a slightly more complicated version of \autoref{alg:boostedUnbiased} that avoids such large jumps, and so we can formally analyze its discretized version.

\begin{algorithm}[!ht]
	\caption{Boosted Suppressed-Bias Phase Estimation}\label{alg:boostedSuppressed}
	\begin{algorithmic}[1] 
		\STATEx	{\bf Input:} $(2m+1)$ copies of $\ket{\psi(\phi)}=\frac{1}{\sqrt{M}}\sum_{k=0}^{M-1}e^{i \phi k}\ket{k}$ (for unknown $\phi$)
		\STATE {\bf For each} $j\in [2m+1]$
		\STATE  ~~~ Run \autoref{alg:unbiased} with a fixed $n$ on the $j$-th copy of $\ket{\psi(\phi)}$ and record the estimate $\varphi_j$\label{line:estiamtes}
		\STATE {\bf For each} $j\in [2m+1]$\label{line:computeDsitances}
		\STATE  ~~~ Compute $d_j\!\!$, the $m$-th smallest distance in the (multi)set $\big\{|\varphi_j\!-\!\varphi_k|_{2\pi}\colon k\in [2m\!+\!1]\setminus\{j\}\big\}\!$
		\STATE  ~~~ Define $w_j:=\exp(-\frac{mM}{4}d_j)$
		\STATE	{\bf Return} $\overline{\varphi}:=\varphi_j$ with probability $\frac{w_j}{W}$ where $W=\sum_{j \in [2m+1]}w_j$ \label{line:sampleOutput}					
	\end{algorithmic}
\end{algorithm}	
\begin{theorem}\label{thm:boostedSupressed}
	If $n\geq \log_2(\pi m)$, then \autoref{alg:boostedSuppressed} returns a $\overline{\varphi}\in[0,2\pi]$ such that
	\begin{equation}\label{eq:boosted}
	\Pr[|\phi-\overline{\varphi}|_{2\pi}\leq \frac{10}{M}(1+2^{-n})]\geq 1-2 e^{-\frac{m}{4}}- 4\pi(m+1)2^{-n},
	\end{equation}
	and
	\begin{equation}\label{eq:suppressed}
	|\mathbb{E}[\phi-\overline{\varphi}]|\leq 32\pi (m+1)2^{-n},
	\end{equation}
	where we interpret $\phi-\overline{\varphi}$ as a number in $[-\pi,\pi)$.
\end{theorem}
\begin{proof}
	We start by analyzing \autoref{alg:boostedSuppressed} in the infinite precision ($n=\infty$) case. Due to symmetry, the estimate $\overline{\varphi}$ is unbiased; interpreting $\phi-\overline{\varphi}$ as a number in $[-\pi,\pi)$ (rather than, say $(-\pi,\pi]$) does not introduce bias either, since the probability density of the estimate $\overline{\varphi}$ is continuous in the case $n=\infty$, and so $\Pr[\overline{\varphi}=\phi\pm \pi]=0$.
	Thus, \autoref{eq:suppressed} trivially holds.

	As in the proof of \autoref{thm:unbiased}, the probability that there are at least $m+1$ estimates $\varphi_j$ such that $|\phi-\varphi_j|_{2\pi}\leq \frac{3}{M}$ is at least $1-\exp(-\frac{m}{4})$ due to the Chernoff bound.
	In that case for each $\varphi_j \in [\phi-\frac{3}{M},\phi+\frac{3}{M}]$ (modulo $2\pi$) we have that $d_j\leq 6/M$ so $w_j\geq\exp(-\frac{3m}{2})$ and thus $W\geq (m+1)\exp(-\frac{3m}{2})$. Further, since there are at most $m$ remaining $\varphi_i\notin [\phi-\frac{3}{M},\phi+\frac{3}{M}]$, for any such $i$ the $m$ shortest distances $|\varphi_i-\varphi_k|_{2\pi}$ must include some $\varphi_k\in [\phi-\frac{3}{M},\phi+\frac{3}{M}]$. Therefore, if $|\varphi_i-\phi|_{2\pi}\geq 10/M$ then $d_i^{(m)}\geq 7/M$ and $w_i\leq \exp(-\frac{7m}{4})$. We can conclude the proof of \autoref{eq:boosted} for the case $n=\infty$ by using the union bound, observing that 
	\begin{align}\label{eq:infBoosted}
	\Pr[|\phi-\overline{\varphi}|_{2\pi}\geq 10/M]
	\leq \exp(-m/4) + m\exp(-7m/4)/W
	\leq 2\exp(-m/4).
	\end{align}
	
	To analyze the version of \autoref{alg:boostedSuppressed} with finite $n$, we proceed similarly to the proof of \autoref{thm:supbiased}: we start from $n=\infty$ and in several steps consider using the truncated $u$ in various lines of \autoref{alg:boostedSuppressed} until we get to the final algorithm that uses only $n$ binary digits of $u$ throughout. For this let us introduce the notation $\overline{\varphi}^{(\ell)}$, meaning the output of the modified infinite precision \autoref{alg:boostedSuppressed} where $u$ is truncated to $n$ digits just before executing Line $\ell$.
	
	First, let us consider truncating $u$ in Line \ref{line:sampleOutput}, thus returning $\overline{\varphi}^{(6)}$, i.e., $\varphi_j$ but represented with only finite precision. This introduces a change in the output $|\overline{\varphi}^{(6)}-\overline{\varphi}|_{2\pi}$ that is at most $\frac{2\pi}{M}2^{-n}<\frac{10}{M}2^{-n}$, thereby \autoref{eq:infBoosted} implies
	\begin{equation}\label{eq:infBoosted2}
	\Pr[|\phi-\overline{\varphi}^{(6)}|_{2\pi}\geq \frac{10}{M}(1+2^{-n})]
	\leq \Pr[|\phi-\overline{\varphi}|_{2\pi}\geq \frac{10}{M}]	
	\leq 2\exp(-m/4).
	\end{equation}
	Also, the value of $\phi-\overline{\varphi}^{(6)}\in [-\pi,\pi)$ changes by at most $\frac{2\pi}{M}2^{-n}$ compared to $\phi-\overline{\varphi}\in [-\pi,\pi)$, unless $\overline{\varphi}\in[\phi+\pi,\phi+\pi+\frac{2\pi}{M}2^{-n}]$ modulo $2\pi$ (when the change might be as large as $2\pi$). The probability of the latter happening can be bounded by
	\begin{align*}
	\Pr[\overline{\varphi}\in[\phi+\pi,\phi+\pi+\frac{2\pi}{M}2^{-n}]]
	&\leq \Pr[\exists j \colon \varphi_j\in[\phi+\pi,\phi+\pi+\frac{2\pi}{M}2^{-n}]]\\
	&\leq (2m+1)\Pr[\varphi\in[\phi+\pi,\phi+\pi+\frac{2\pi}{M}2^{-n}]]\\
	&\leq (2m+1)\frac{2\pi}{M}2^{-n}\sup_\varphi f(\varphi)\\
	&\leq (2m+1)2^{-n}, \tag{by \autoref{eq:preProbDens}}
	\end{align*}
	thus we get
	\begin{equation}\label{eq:SuppBias2}
	|\mathbb{E}[\phi-\overline{\varphi}^{(6)}]|\leq \frac{2\pi}{M}2^{-n} + 2\pi  \Pr[\overline{\varphi}\in[\phi+\pi,\phi+\pi+\frac{2\pi}{M}2^{-n}]]
	\leq 2\pi (2m+2)2^{-n}.
	\end{equation}	
	
    Second, let us consider truncating $u$ in Line \ref{line:computeDsitances}. This will change the distances $d_j$ by at most $\frac{2\pi}{M}2^{-n}$ since every pair of distances $|\varphi_j\!-\!\varphi_k|_{2\pi}$ is changed by no more than $\frac{2\pi}{M}2^{-n}$. Thus every weight $w_j $ is perturbed by a multiplicative factor $\exp(\pm\frac{mM}{4}\frac{2\pi}{M}2^{-n})=\exp(\pm\frac{m\pi}{2}2^{-n})$ and consequently $W$ gets a multiplicative perturbation up to $\exp(\pm\frac{m\pi}{2}2^{-n})$. This induces a perturbation of the probabilities $p_j:=w_j/W$ by a multiplicative factor up to $\exp(\pm m\pi2^{-n})$, resulting in an up to $(\exp(m\pi2^{-n})-1)/2$-perturbation in total variation distance to the sampling distribution in Line \ref{line:sampleOutput}. If $n\geq \log_2(\pi m)$ this can be upper bounded by $3m2^{-n}$. Sine the only change between the outputs $\overline{\varphi}^{(6)}$ and $\overline{\varphi}^{(3)}$ is due to the change in the sampling distribution in Line \ref{line:sampleOutput} we get by \autoref{eq:infBoosted2} 
	\begin{equation}\label{eq:infBoosted3}
	\Pr[|\phi-\overline{\varphi}^{(3)}|_{2\pi}\geq \frac{10}{M}(1+2^{-n})] 
	\leq \Pr[|\phi-\overline{\varphi}^{(6)}|_{2\pi}\geq \frac{10}{M}(1+2^{-n})] + 3m2^{-n}
	\leq 2\exp(-m/4) + 3m2^{-n},
	\end{equation}
	and by \autoref{eq:SuppBias2}
	\begin{equation}\label{eq:SuppBias3}
	|\mathbb{E}[\phi-\overline{\varphi}^{(3)}]|\leq |\mathbb{E}[\phi-\overline{\varphi}^{(6)}]| + 2\pi\cdot 3m2^{-n}
	\leq 2\pi (2m+5)2^{-n}.
	\end{equation}	
	
	Finally, let us truncate $u$ right at the beginning. This only affects the unitary applied in Line  \ref{line:applyPhase} of \autoref{alg:unbiased}. As we showed in the proof of \autoref{thm:supbiased} this introduces a change in the applied unitary with magnitude (in terms of operator norm) no greater than $2\pi2^{-n}$, ultimately changing the measurement statistics by no more than $2\pi2^{-n}$ in total variation distance. Since \autoref{alg:unbiased} is repeated $(2m+1)$ times, the overall perturbation in total variation distance can be bounded by $(2m+1)2\pi2^{-n}$.
    Thus we get by \autoref{eq:infBoosted3} 
	\begin{equation*}\label{eq:infBoosted4}
	\Pr[|\phi-\overline{\varphi}^{(1)}|_{2\pi}\geq \frac{10}{M}(1+2^{-n})] 
	\leq \Pr[|\phi-\overline{\varphi}^{(3)}|_{2\pi}\geq \frac{10}{M}(1+2^{-n})] + (2m+1)\frac{2\pi}{2^{n}}
	\leq 2e^{-\frac{m}4} + \frac{4\pi(m+1)}{2^{n}},
	\end{equation*}
	and by \autoref{eq:SuppBias3}
	\begin{equation*}
	|\mathbb{E}[\phi-\overline{\varphi}^{(1)}]|\leq |\mathbb{E}[\phi-\overline{\varphi}^{(3)}]| + 2\pi\cdot (2m+1)2\pi2^{-n}
	\leq 32\pi (m+1)2^{-n}.\qedhere
	\end{equation*}		
\end{proof}

We remark that one can similarly show that our suppressed-bias phase estimators give rise to suppressed-bias estimators of $e^{i \phi}$ by using the constructions of~\autoref{subsec:unbiasedComplex}, and consequently also allows for the implementation of suppressed-bias probability estimators.

\subsection{Application to gradient estimation}

	Now we show that our suppressed-bias phase estimation techniques lead to suppressed-bias gradient estimation further improving over Jordan's gradient estimation algorithm, and its variants \cite[Lemma 5.1]{gilyen2017OptQOptAlgGrad}. Note that if we would use perfect input states and the exact unbiased phase-estimation then we would get a symmetric error distribution. Unfortunately, this no longer holds due to the approximation errors in the input state and the finite bit precision.
	
	To simplify the exposition, let us introduce the notation $B_\infty(\pmb{g},\eps)$ to denote the (closed) $\eps$-ball around $\pmb{g}$ containing all points $\pmb{x}$ such that $\nrm{\pmb{x}-\pmb{g}}_\infty\leq \eps$.

	\begin{theorem}[Suppressed-bias gradient estimation]\label{thm:unbiasedGradient}
		Let $\eps,\delta\in\!(0,\frac{1}{6}]$ and $\pmb{g}\!\in\!\R^d$ such that $\nrm{\pmb{g}}_\infty\!\leq\frac13$. 
		Let $b:=\lceil \log_2(\frac{2}{\eps}) \rceil$ and $B := 2^{b}$. If $\beta:=\nrm{\ket{\psi} - \frac{1}{\sqrt{B^d}}\sum_{\pmb{x} \in G_b^d}e^{2\pi i B \ipc{\pmb{g}}{\pmb{x}}} \ket{\pmb{x}}}\leq \frac{\delta}{24\lceil\ln(6d/\delta)\rceil+3}$ and we are given $8\lceil\ln(6d/\delta)\rceil+1$ copies of $\ket{\psi}$, then we can compute a vector $\pmb{k}\in [-\frac{1}{2},\frac{1}{2}]^d$ such that 
				\begin{equation}\label{eq:boostedGradient}
		\Pr\left[\nrm{\pmb{k}-\pmb{g}}_\infty> \eps\right]\leq \delta
		\end{equation}
		and 	
		\begin{equation}\label{eq:suppressedGradient}
		\nrm{\mathbb{E}[\pmb{k}]-\pmb{g}}_\infty\leq \delta.
		\end{equation}
		Furthermore, the gate complexity of the procedure is $\bigO{d\log(\frac{d}{\delta})\log(\frac{1}{\eps})\log\left(\frac{d}{\delta}\log(\frac{1}{\eps})\right)}$, and the circuit depth is $\bigO{\log(\frac{1}{\eps})\log\left(\frac{d}{\delta}\log(\frac{1}{\eps})\right)}$. Finally, there is a random variable $\pmb{k'}\in B_\infty(\pmb{g},\eps)$ with independent coordinates that is $\delta$-close in total variation distance to $\pmb{k}$ and satisfies $\mathbb{E}[\pmb{k'}]\in B_\infty(\pmb{g},\delta)$.		
	\end{theorem}
\begin{proof}
	We apply \autoref{thm:boostedSupressed} with $M=B$, $m=4\lceil\ln(6d/\delta)\rceil$ and $n=\left\lceil\log_2 \left(\frac{96 d (m+1)}{\delta}\right)\right\rceil$. This choice of parameters imply that 
	\begin{align}
		2 d e^{\frac{-m}{4}}& \leq \frac{\delta}{3},\label{eq:unbGradPar1}\\
		16 d (m+1) 2^{-n}	& \leq \frac{\delta}{6},\label{eq:unbGradPar2}\\
		(2m+1) \beta		& \leq \frac{\delta}{3}.\label{eq:unbGradPar3}
	\end{align}
	First let us assume that we have access to the ``ideal'' state 
	$$\ket{\phi}^{\otimes 2m+1}:=\left(\frac{1}{\sqrt{B^d}}\sum_{\pmb{x} \in G_b^d}e^{2\pi i B \ipc{\pmb{g}}{\pmb{x}}} \ket{\pmb{x}}\right)^{\!\!\!\otimes 2m+1}.$$ Since this is a product state when we apply \autoref{alg:boostedSuppressed} on each of the $d$ coordinates independently, then the guarantees of \autoref{thm:boostedSupressed} apply to each coordinate independently. Let us denote the output of \autoref{alg:boostedSuppressed} on the ``ideal' state $\ket{\phi}$ by $\pmb{\tilde{k}}$ after dividing by $2\pi$ and subtracting $\frac12$ from each coordinate. Then we have that 
	\begin{align*}
		\frac{\delta}{2d}
		&\geq 	2 e^{-\frac{m}{4}}+ 4\pi(m+1)2^{-n}\tag{by \autoref{eq:unbGradPar1}-\eqref{eq:unbGradPar2}}\\		
		&\geq \Pr[|2\pi \tilde{k}_i-2\pi g_i|_{2\pi}> \frac{10}{B}(1+2^{-n})]\tag{by \autoref{eq:boosted}}\\
		&\geq \Pr[|2\pi \tilde{k}_i-2\pi g_i|_{2\pi}> \frac{4\pi}{B}] \tag{$n\geq 4\Rightarrow 10(1+2^{-n})\leq 4\pi$}\\	
		&\geq \Pr[|2\pi \tilde{k}_i-2\pi g_i|_{2\pi}> 2\pi\eps] \tag{$\frac{2}{B}\leq\eps$}\\
		&= \Pr[|\tilde{k}_i-g_i|> \eps] \tag{$\eps\leq \frac16, |g_i|\leq \frac13$}.		
	\end{align*}
	By the union bound we get that 
	\begin{equation}\label{eq:idealClose}
	    \Pr\left[\|\pmb{\tilde{k}}-\pmb{g}\|_\infty> \eps\right] \leq \frac{\delta}{2}.
	\end{equation} 
	
	The closeness condition $\nrm{\ket{\psi} - \ket{\phi}}\leq \beta$ guarantees that $\nrm{\ket{\psi}^{\otimes 2m+1} - \ket{\phi}^{\otimes 2m+1}}\leq \frac{\delta}{3}$, and if we use a $\frac{\delta}{(144\lceil\ln(6d/\delta)\rceil+6)d}$-precise implementation\footnote{Probably it is enough if the implementation is about $d$-times less precise analogously to the proof of \autoref{cor:blockToGrad}.} of the quantum Fourier transform in \autoref{thm:boostedSupressed}, then the distance from the ``ideal'' state before the measurement in \autoref{thm:boostedSupressed} can be bounded by $\frac{\delta}{3}+\frac{\delta}{6}=\frac{\delta}{2}$, and so the total variation distance of the ``ideal'' $\pmb{\tilde{k}}$ and the actual $\pmb{k}$ estimators can be bounded by $\frac{\delta}{2}$ (see for example \cite[Exercise 4.3]{wolf2019QCLectureNotes}). In particular the the probability of any event changes by at most $\frac{\delta}{2}$ and therefore \autoref{eq:idealClose} implies \autoref{eq:boostedGradient}. 
	
	A similar argument shows the suppression of bias for all $i\in d$
	\begin{align*}
	\frac{\delta}{6d}
	&\overset{\eqref{eq:unbGradPar2}}{\geq} 16(m+1)2^{-n}
	\overset{\eqref{eq:suppressed}}{\geq} |\mathbb{E}[\tilde{k}_i-g_i]|.
	\end{align*}
	This then implies that $\|\mathbb{E}[\pmb{\tilde{k}}]-\pmb{g}\|_\infty\leq \frac{\delta}{6d}\leq \frac{\delta}{6}$. On the other hand the total variation distance of the distributions of $\pmb{\tilde{k}}$ and $\pmb{k}$ is at most $\frac{\delta}{2}$, therefore $\|\mathbb{E}[\pmb{\tilde{k}}]\|_\infty-\|\mathbb{E}[\pmb{k}]\|_\infty\leq \frac{\delta}{3}$ holds\footnote{One can see this by a coupling argument: if two random variables $X,Y$ satisfy that $\vertiii{X},\vertiii{Y}\leq L$ for some norm $\vertiii{\cdot}$ and their total variation distance is at most $K$, then $\vertiii{\mathbb{E}[X]}-\vertiii{\mathbb{E}[Y]}\leq\vertiii{\mathbb{E}[X-Y]}\leq \mathbb{E}[\vertiii{X-Y}]\leq 2LK$.\label{foot:tv}} implying via the triangle inequality that $\|\mathbb{E}[\pmb{k}]-\pmb{g}\|_\infty\leq \frac{2\delta}{3}$ proving \autoref{eq:suppressedGradient}. 
	
	Let us define $\pmb{k'}$ as the truncation of $\pmb{\tilde{k}}$ into $B_\infty(\pmb{g},\eps)$. Due to \autoref{eq:idealClose} the the total variation distance between $\pmb{k'}$ and $\pmb{\tilde{k}}$ is at most $\frac{\delta}{2}$, so the total variation distance between $\pmb{k'}$ and $\pmb{g}$ is at most $\delta$. Also it is easy to see\textsuperscript{\normalfont\ref{foot:tv}} that $\|\mathbb{E}[\pmb{\tilde{k}}]-\mathbb{E}[\pmb{k'}]\|_\infty\leq \frac{\delta}{2}$ and so by the triangle inequality we get $\|\mathbb{E}[\pmb{k'}]-\pmb{g}\|_\infty\leq \frac{5\delta}{6}$.
	
	Since we run \autoref{alg:boostedSuppressed} independently for each of the $d$ coordinates, the complexity is $d$ times the complexity of executing \autoref{alg:boostedSuppressed}. The gate complexity (and depth) of \autoref{alg:boostedSuppressed} is dominated by the quantum Fourier transform, which we implement approximately\cite{barenco1996ApproxQFourierTrafo} with precision about $\frac{\delta}{d\log(d/\delta)}$. Each such implementation cost $\bigO{b\log\left(\frac{bd\log(d/\delta)}{\delta}\right)}\!=\bigO{\log(\frac{1}{\eps})\log\left(\frac{d}{\delta}\log(\frac{1}{\eps})\right)}$ gates, cf. \cite[Exercise 4.4]{wolf2019QCLectureNotes}. This gives the gate complexity $\bigO{d\log(\frac{d}{\delta})\log(\frac{1}{\eps})\log\left(\frac{d}{\delta}\log(\frac{1}{\eps})\right)}$. The classical computation required by \autoref{alg:boostedSuppressed} can be performed in time $\bigO{\mathrm{poly}(n,m,b)}$ which is $\bigO{\polylog({\frac{d}{\delta\eps}})}$, since $n,m=\bigO{\log(\frac{d}{\delta})}$, and $b=\bigO{\log(\frac{1}{\eps})}$.
\end{proof}

Finally, we prove a corollary analogous to \autoref{cor:blockToGrad} which will be the main technical tool in the following  \autoref{sec:multipleExp}-\ref{s:mixed-state-tomo}.

\begin{corollary}[Almost linear block-Hamiltonian to gradient]\label{cor:unbiasedBlockToGrad}
	Let $\eps,\delta\in\!(0,\frac{1}{6}]$, $b:=\lceil \log_2(\frac{16}{\eps}) \rceil$, $B=2^b$ and  $\beta:=\frac{\delta}{96\lceil\ln(6d/\delta)\rceil+12}$. Suppose that we have an $a$-block-encoding $W$ of a diagonal matrix with diagonal entries $f(\pmb{x})\in \R$ for $\pmb{x} \in G_b^d$ satisfying $|f(\pmb{x}) - \ipc{\pmb{x}}{\pmb{g}}|\leq \frac{\eps\beta}{4\pi}$ for at least a $(1- \beta^2)$ fraction of the points in $G_b^d$. Then with $\bigO{\left(\frac{1}{\eps}+\log(\frac{\log(d)}{\delta})\right)\log(\frac{d}{\delta})}$ (controlled) uses of $W$ (and its inverse) and 
	$\bigO{\left(d\log(\frac{1}{\eps})\log\left(\frac{d}{\delta}\log(\frac{1}{\eps})\right)+a\left(\frac{1}{\eps}+\log(\frac{\log(d)}{\delta})\right)\right) \log(\frac{d}{\delta})}$ other gates with circuit depth $\bigO{\log(\frac{1}{\eps})\log\left(\frac{d}{\delta}\log(\frac{1}{\eps})\right)+\log(a)\left(\frac{1}{\eps}+\log(\frac{\log(d)}{\delta})\right)}$
	we can compute a vector $\pmb{k}\in [-4,4]^d$ such that 
	\begin{equation}\label{eq:boostedGradientCor}
	\Pr\left[\nrm{\pmb{k}-\pmb{g}}_\infty> \eps\right]\leq \delta,
	\end{equation}
	and 	
	\begin{equation}\label{eq:suppressedGradientCor}
	\nrm{\mathbb{E}[\pmb{k}]-\pmb{g}}_\infty\leq 8\delta.
	\end{equation}
	Moreover, there is a random variable $\pmb{k'}\in B_\infty(\pmb{g},\eps)$ with independent coordinates that is $\delta$-close in total variation distance to $\pmb{k}$ and satisfies $\mathbb{E}[\pmb{k'}]\in B_\infty(\pmb{g},8\delta)$.			    
\end{corollary}
\begin{proof}
	We proceed similarly to the proof of \autoref{cor:blockToGrad}.
	The main idea is to apply \autoref{thm:unbiasedGradient} with preparing the (approximate) initial state via block-Hamiltonian simulation \autoref{lem:blockHamSim}. In the proof of \autoref{cor:blockToGrad} it is shown that the assumptions in the statement imply $\nrm{\pmb{g}}_\infty\leq \frac{8}{3}$.
	Therefore, we will apply \autoref{thm:unbiasedGradient} to the gradient $\frac{\pmb{g}}{8}$ with precision $\frac{\eps}{8}$.
	The first step is to prepare a uniform superposition over the grid $G_b^d$ by applying a Hadamard gate to all $d\cdot b$ qubits, that are initially in the $\ket{0}$ state. 
		
	First let us assume that we have access to a perfect phase oracle $P:=\sum_{\pmb{x}\in G_b^d} \ketbra{\pmb{x}}{\pmb{x}} e^{2\pi i \frac{B}{8} f(\pmb{x})}$ so that we can prepare the sate $\ket{\psi}=\frac{1}{\sqrt{B^d}}\sum_{\pmb{x}\in G_b^d}\ket{\pmb{x}}e^{2\pi i \frac{B}{8} f(\pmb{x})}$. We bound the difference from the ideal state $\ket{\phi}$ analogously to the proof of \cite[Lemma 5.1]{gilyen2017OptQOptAlgGrad}. Let $S\subseteq G_b^d$ be the set of points for which $|f(\pmb{x}) - \ipc{\pmb{x}}{\pmb{g}}|\leq \frac{\eps\beta}{4\pi}$ holds, then
	\begin{align*}
	\nrm{\ket{\psi}\!-\!\ket{\phi}}^2\!
	&=\!\frac{1}{B^d}\sum_{\pmb{x}\in G_b^d}\left|e^{2\pi i \frac{B}{8} f(\pmb{x})}-e^{2\pi i \frac{B}{8} \ipc{\pmb{x}}{\pmb{g}}}\right|^2\\
	&=\!\frac{1}{B^d}\!\sum_{\pmb{x}\in S}\left|e^{2\pi i \frac{B}{8} f(\pmb{x})}-e^{2\pi i \frac{B}{8} \ipc{\pmb{x}}{\pmb{g}}}\right|^2
	\!\!+\!\frac{1}{B^d}\!\!\sum_{\pmb{x}\in G_b^d\setminus S}\!\left|e^{2\pi i \frac{B}{8} f(\pmb{x})}-e^{2\pi i \frac{B}{8} \ipc{\pmb{x}}{\pmb{g}}}\right|^2\\
	&\leq \!\frac{1}{B^d}\!\sum_{\pmb{x}\in S}\left|2\pi \frac{B}{8} f(\pmb{x})-2\pi \frac{B}{8} \ipc{\pmb{x}}{\pmb{g}}\right|^2
	\!\!+\!\frac{1}{B^d}\!\!\sum_{\pmb{x}\in G_b^d\setminus S}\!4 \tag{$|e^{iz}-e^{iy}|\leq |z-y|$}\\
	&=\!\frac{1}{B^d}\!\sum_{\pmb{x}\in S}(2\pi \frac{B}{8} )^2\left|f(\pmb{x})-\ipc{\pmb{x}}{\pmb{g}}\right|^2
	\!\!+4\frac{|G_b^d\setminus S|}{B^d} \\
	&\leq\!\frac{1}{B^d}\!\sum_{\pmb{x}\in S}4\beta^2+4\beta^2 \tag{by the assumptions of the corollary}\\	
	&\leq 8\beta^2.
	\end{align*}
	Finally, we can implement a $(4-2\sqrt{2})\beta$-approximation $\widetilde{P}$ of the perfect phase oracle $P$ by applying block-Hamiltonian simulation \autoref{lem:blockHamSim} to $W$.\footnote{An $\eps$-precise $(a+2)$-block-encoding of $e^{itH}$ is $\bigO{\sqrt{\eps}}$-close in operator norm to a perfect Hamiltonian simulation unitary $U$ of the form $\ketbra{0}{0}^{\otimes a+2}\otimes e^{\ci t A} + V$, where $V(\ket{0}^{a+2}\otimes I)=0$.} This lets us preparing an approximate state $\ket{\tilde{\psi}}$ such that $\nrm{\ket{\tilde{\psi}}-\ket{\psi}}\leq (4-2\sqrt{2})\beta$ and so $\nrm{\ket{\tilde{\psi}}\!-\!\ket{\phi}}\leq4\beta$, enabling us to apply \autoref{thm:unbiasedGradient}.
	
	The query complexity follows from the fact that we prepare the state $\ket{\tilde{\psi}}$ a total of $\bigO{\log(\frac{d}{\delta})}$ times, each time making $\bigO{\frac{1}{\eps}+\log(\frac{1}{\beta})}=\bigO{\frac{1}{\eps}+\log(\frac{\log(d)}{\delta})}$ (controlled) queries to $W$. The additional gate complexity of preparing $\ket{\tilde{\psi}}$ is $\bigO{a}$ times the query complexity plus $d\cdot b$ for the Hadamard gates. We get the overall gate complexity by adding the gate cost in \autoref{thm:unbiasedGradient}.	
\end{proof}

\subsection{Application to low depth probability estimation}

In this section, we sketch a quick application of our results. If we have access to the operation
\[U : \ket{0} \mapsto \sqrt{1-p}\ket{\psi_0}\ket{0} + \sqrt{p}\ket{\psi_1}\ket{1},\]
then we can estimate $p$ with a depth-$t$ algorithm, by running our version of unbiased probability estimation with $1/M = \bigO{\log(t)/t}$ and $m = \bigO{\log(t)}$, to obtain an estimate with variance $\bigOt{p(1-p)/t^2 + 1/t^4}$. Moreover, we can run this procedure $K$ times in parallel, and take the average of the outcomes. This gives an estimator of $p$ that is still unbiased, and whose variance is
\[\bigOt{\frac{p(1-p)}{Kt^2} + \frac{1}{Kt^4}}.\]
Thus, we obtain a way to estimate $p$, when we are constrained to using depth-$t$ quantum algorithms, and we can obtain precision $\eps$ with high probability if we set $K = \Theta(\max\{p(1-p)/(\eps t)^2, 1/(\eps^2 t^4)\})$.

Now, let $\beta \in (0,1]$, and suppose the depth that we can use is $t = \Theta(1/\eps^{1-\beta})$. Then, in order to achieve precision $\eps$, we can set $K = \bigOt{\max\{p(1-p)/\eps^{2\beta}, 1/\eps^{4\beta-2}\}}$, from which we find that the total number of calls to $U$ becomes $Kt = \bigOt{\max\{p(1-p)/\eps^{1+\beta}, 1/\eps^{3\beta-1}\}}$. If we use the crude upper bound $p(1-p) \leq 1$, this reduces to $Kt = \bigOt{1/\eps^{1+\beta}}$, and hence we recover the result obtained in \cite{giurgica2020low}. Moreover, we get a slight improvement if we know some small upper bound $q \geq p$ a priori.

\section{Second intermezzo: estimating multiple expectation values with a state-preparation oracle}\label{sec:multipleExp}
To perform efficient mixed-state tomography we rely on an algorithm to estimate $m$ expectations with few copies of the state. The algorithm is based on constructing the phase oracle for a function whose gradient is the vector of the desired expectation values, similarly to what we did for pure states. The task here is however more complicated, because to ensure that the function is properly normalized we need to bound the weighted combination of expectation values, where the weights are taken from a hypergrid in $[-\frac12, \frac12]^m$ (as these are the points used by the gradient algorithm of \cite{gilyen2017OptQOptAlgGrad}). This requires some results on random matrices, which we use by translating properties that hold for uniformly random matrices into properties that hold for all but a constant fraction of the points in the hypergrid.

Formally, we assume access to a unitary that prepares a purification of a state $\rho\in \C^{d\times d}$, and its inverse. Our goal is to estimate the expectation values $\tr{E_j\rho}$ of measurement operators $E_j$ for $j=1,\dots,m$ up to corresponding errors $\eps_j$, with as few applications of the state-preparation unitary for $\rho$ as possible. We do not apply any gates to the purifying register, other than the state-preparation oracle and its inverse; thus, we do not need to impose any restrictions on how the purification of $\rho$ is constructed. We assume that $\nrm{E_j} \le 1$ for all $j=1,\dots,m$, which is w.l.o.g.\ as we can always scale $E_j$ and $\eps_j$ down by $\nrm{E_j}$ to achieve this. Finally, we assume that we are given access to each $E_j$ via a block-encoding. Note that other models are possible as our algorithm only requires the ability to compute $\Tr(E_j \rho)$:  the block-encoding framework is general and simplifies our exposition. For example, if we have an implementation of a POVM for $E_j, I-E_j$, then we can convert this to a block-encoding for $E_j$ via \autoref{lem:povmtoblock}.

This task was recently studied in \cite{huggins2021QAlgMultipleExpectationValues}, yielding an algorithm that solves the problem using $\bigO{\sqrt{m}/\eps}$ applications of the state-preparation unitary and its inverse, in the case where all $\eps_j$ are equal to $\eps$. Their algorithm is however not optimal in our setting: we want to give an algorithm with a sample complexity that depends on $\sqrt{\nrm{\sum_{j=1}^m E_j^2/\eps_j^2}}$, because this leads to a saving of a factor $d$ when applied to mixed-state tomography compared to the algorithm of \cite{huggins2021QAlgMultipleExpectationValues}. The details are discussed subsequently in this section. For a discussion of other existing approaches to solve the problem of computing expectation values, we refer to the excellent introduction in \cite{huggins2021QAlgMultipleExpectationValues}.

\subsection{Bounds on uniform matrix series}

As mentioned above, we first need to prove some properties of uniform random matrices. We do this by adapting a result on Gaussian / Rademacher random matrices given below. Here and in the remainder, for a random matrix $Y$ we define $v(Y) := \nrm{\mathbb{E}[Y^2] - (\mathbb{E}[Y])^2}$ as its variance.

\begin{theorem}[Gaussian \& Rademacher matrix series inequality {\cite[Theorem 4.6.1]{tropp2016IntroMatrixConcInequalities}}]
  \label{thm:gaussianMatrixBound}
  Let $E_1,\dots,E_m$ be $d\times d$ Hermitian matrices. Let $\lambda_1,\dots,\lambda_m$ be drawn from iid standard normal distributions and let $Y=\sum_{j=1}^m\lambda_j E_j$. Then $\mathbb{E}[Y] = 0, v(Y) = \nrm{\sum_j E^2_j}$ and
  \[
  \mathbb{P}[\nrm{Y}\geq t]\leq 2d e^{-\frac{t^2}{2v(Y)}}.
  \]
  The same bounds hold when $\{\lambda_j\}$ is iid uniformly random over $\{-1,1\}$.
\end{theorem}

In order to adapt the above result to our setting we invoke a technical statement from~\cite{tropp2016IntroMatrixConcInequalities}:

\begin{proposition}[Master Bound for a Sum of Independent Random Matrices, {\cite[Theorem 3.6.1]{tropp2016IntroMatrixConcInequalities}}]\label{prop:masterBound}
Consider a finite sequence $\{E_j\}$ of independent, random, Hermitian matrices of the same size. Then for all $t\in \mathbb{R}$ we have
$$\mathbb{P}[\lambda_{max}(\sum_j E_j)\geq t] \leq \inf_{\theta>0} e^{-\theta t} \tr{\exp\trm{ \sum_j \log \mathbb{E}[e^{\theta E_j}]}}.$$
\end{proposition}

With the help of this result we prove the following variant of \autoref{thm:gaussianMatrixBound} for bounded random variables:

\begin{theorem}[Bounded Matrix series inequality]
  \label{lem:uniformMatrixBound}
  Let $E_1,\dots,E_m$ be $d\times d$ Hermitian matrices. Let $\lambda_1,\dots,\lambda_m$ be independent symmetrically distributed random variables supported on $[-1,1]$ and let $Y=\sum_{j=1}^m\lambda_jE_j$. Then $\mathbb{E}[Y] = 0, v(Y) \leq  \nrm{\sum_j E^2_j}$ and
  \[
  \mathbb{P}[\nrm{Y}\geq t]\leq 2d e^{-\frac{t^2}{2v(Y)}}.
  \]
\end{theorem}
\begin{proof}
    We follow the proof of \cite[Theorem 4.6.1]{tropp2016IntroMatrixConcInequalities} and modify it where necessary.
    First we note that
  \begin{align}
    \mathbb{E}[e^{\lambda_jE_j}] &= \mathbb{E}\left[\sum_{k=0}^{\infty} \frac{\lambda_j^k}{k!}E_j^k \right ] \nonumber\\
    &= \sum_{k=0}^{\infty} \frac{\mathbb{E}[\lambda_j^k]}{k!} E_j^k \tag{linearity of expectation}\\
    &= \sum_{q=0}^{\infty} \frac{\mathbb{E}[\lambda_j^{2q}]}{(2q)!} E_j^{2q} \tag{$\lambda_j$ is symmetrically distributed}\\
    &\preceq \sum_{q=0}^{\infty} \frac{1}{(2q)!} E_j^{2q} \tag{$\lambda_j$ is bounded}\\
    &\preceq \sum_{q=0}^{\infty} \frac{1}{q!} (E_j^2/2)^q \tag{$(2q)!\geq 2^q q!$}\\
    &= e^{E_j^2/2}. \label{eq:matExpExpBound}
  \end{align}
  Now we show that the above inequality implies that 
  \begin{equation}
      \tr{\exp\trm{ \sum_j \log \mathbb{E}[e^{\lambda_jE_j}]}}\leq \tr{\exp\trm{ \frac{1}{2}\sum_j E_j^2}}.
  \end{equation}\label{eq:matExpExpBound2}
  Indeed, we know that \cite[Chapter 4.1]{hiai2014IntroMatrixAnalAndApp} the logarithm is operator monotone for positive matrices. Therefore \autoref{eq:matExpExpBound} implies that $\log(\mathbb{E}[e^{\lambda_jE_j}]) \preceq \log(e^{E_j^2/2})$, and consequently 
$\sum_j \log \mathbb{E}[e^{\lambda_j2E_j}] \preceq \sum_j\log(e^{E_j^2/2})=\sum_j E_j^2/2$. We conclude by using the fact that the trace of a monotone function is operator monotone \cite[Example 3.24]{hiai2014IntroMatrixAnalAndApp}, i.e., $A\preceq B$ implies $\tr{\exp(A)}\leq \tr{\exp(A)}$. We now use this (by absorbing $\theta$ into the $E_j$-s) to get
  \begin{align*}
    \mathbb{P}[\lambda_{max}(Y)\geq t] &\leq \inf_{\theta>0} e^{-\theta t} \tr{\exp\trm{ \sum_j \log \mathbb{E}[e^{\theta \lambda_j E_j}]}} \tag{by \autoref{prop:masterBound}}\\
    &\leq \inf_{\theta>0} e^{-\theta t} \tr{\exp\trm{ \frac{\theta^2}{2}\sum_j E_j^2}} \tag{by \autoref{eq:matExpExpBound}}\\
    &\leq \inf_{\theta>0} e^{-\theta t} d \cdot \nrm{\exp\trm{ \frac{\theta^2}{2}\sum_j E_j^2}}\\
    &= \inf_{\theta>0} e^{-\theta t} d \cdot \exp\trm{ \frac{\theta^2}{2} \nrm{\sum_j E_j^2}}\\
    &= d \inf_{\theta>0} e^{-\theta t +\frac{\theta^2}{2} v(Y)}.
  \end{align*}

As the exponential function is monotone increasing, the minimum is attained at the minimum of $-\theta t+ \frac{v(Y)}{2}  \theta^2$. By differentiating and setting equal to zero we find
\[
-t +  v(Y) \theta = 0
\]
and hence $\theta = \frac{t}{v(Y)}$. Substituting this back we find
\[
\mathbb{P}[\lambda_{max}(Y)\geq t] \leq d e^{-\frac{t^2}{2v(t)}}.
\]
By symmetry we get the same bound for the smallest eigenvalue and the theorem follows.
\end{proof}

\subsection{Application to the estimation of multiple expectation values}

With the tools from the previous section we can tighten the analysis of \cite{huggins2021QAlgMultipleExpectationValues} for the estimation of multiple expectation values. Our running time generalizes the results of \cite{huggins2021QAlgMultipleExpectationValues}, and it leads to faster algorithms in some cases that are relevant for tomography.
\begin{lemma}
  \label{lem:linCombEj}
  Let $E_1,\dots, E_m$ be Hermitian matrices with $\nrm{E_j}\leq 1$, and let $U_E=\sum_{j\in [m]}\ketbra{j}{j}\otimes U_{E_j}$, where $U_{E_j}$ is an $a$-block-encoding of $E_j$. Let $\delta > 0$, $\gamma \in \R^m$\!, $\nu =  \nrm{\gamma}_1$,  $\sigma\geq\sqrt{2\nrm{\sum_j \gamma_j^2E_j^2}\ln\left(\frac{2d}{\delta}\right)}$, and $\sigma':=\min\{\nu,\sigma\}$. For any positive integer $b=\bigO{\frac{1}{\eps}}$ we can implement a unitary $V=\sum_{\pmb{x}\in G_b^m}\ketbra{\pmb{x}}{\pmb{x}}\otimes V_x$ such that $V_x$ is an $(a+\ceil{\log_2(m)}+2)$-block-encoding of a matrix $A_x$ that is $\eps$-close in operator norm to $\frac{1}{\sigma'}\sum x_j\gamma_j E_j$ for at least a $1-\delta$ fraction of points $\pmb{x}\in G_b^m$. This implementation of $V$ uses $\bigO{\frac{\nu}{\sigma'} \log(\frac{\nu}{\sigma'\eps})}$ calls to $U_E$, and   
  $\bigO{(a\!+\!m)\frac{\nu}{\sigma'}\polylog(\frac{\nu+m}{\sigma'\eps})}$ additional two-qubit gates having depth $\bigO{\frac{\nu}{\sigma'}\polylog(\frac{\nu+m}{\sigma'\eps})}$.
  \pnote{Joran's first analysis: $\bigO{\nu/\sigma \log(\nu/\eps) m \allowbreak \log \nrm{\gamma}_{\infty}}$ two-qubit gates, and has gate depth $\bigO{\nu/\sigma \log(\nu/\eps) (\log m +\log \nrm{\gamma}_{\infty})}$.}
\end{lemma}
\begin{proof}
  Our goal is to construct a block-encoding of $\frac{1}{\sigma}\sum x_j\gamma_j E_j$. First, we note that this is a valid block-encoding (more precisely, its spectral norm is upper bounded by $\frac{1}{2}$) for at least $1-\delta$ fraction of points $\pmb{x}\in G_b^m$. To see this, we apply \autoref{lem:uniformMatrixBound} to the matrices $\gamma_1E_1, \gamma_2E_2, \dots, \gamma_mE_m$ setting $t =  \sqrt{2\nrm{\sum_j \gamma_j^2E_j^2}\ln\left(\frac{2d}{\delta}\right)}\leq \sigma$ and sampling $\pmb{x}\in G_b^m$ uniformly at random to obtain
  \begin{equation*}
	  \mathbb{P}_{\pmb{x}\in G_b^m}[ \|\sum_j 2x_j \gamma_j E_j \|\geq \sigma] \leq \mathbb{P}_{\pmb{x}\in G_b^m}[ \|\sum_j 2x_j \gamma_j E_j \|\geq t] \leq 2d e^{-\frac{t^2}{2\nrm{\sum_j \gamma^2_j E^2_j}}}
	  = \delta.
  \end{equation*}
  
  Using \autoref{lem:linCombBlocks}, we first prepare a $(a+\ceil{\log_2(m)}+1)$-block-encoding of $\sum_{j=1}^m (x_j \gamma_j/\nrm{\gamma}_1)E_j$. This requires a single application of $U_{E_j}$, and one applications of a state-preparation oracle for $\frac{1}{\sqrt{\nrm{\gamma}_1}} \sum_j \sqrt{x_j\gamma_j}\ket{j} \ket{0} + \ket{\psi}\ket{1}$ (and its inverse), which is trivial to construct with controlled rotations given the binary encoding of $\ket{\pmb{x}}$. We then amplify the block-encoding by a factor $\nu/\sigma = \nrm{\gamma}_1/\sigma$ using \autoref{lem:ampBlock}: this introduces an overhead equal to the amplification factor. Overall, this requires $\bigO{\ceil{\frac{\nu}{\sigma} \log(\frac{\nu}{\sigma\eps})}}$ calls to $U_{E}$.

  The gate complexity of implementing the state-preparation operation to precision $\bigO{\frac{\eps}{\nu m}}$ can be bounded by $\bigO{m \polylog(\frac{\nu + m}{\eps})}$, while \autoref{lem:ampBlock} multiplies this by $\bigO{\ceil{\frac{\nu}{\sigma} \log(\frac{\nu}{\sigma\eps})}}$ and additionally introduces $\bigO{(a+\log(m)+1)\ceil{\frac{\nu}{\sigma} \log(\frac{\nu}{\sigma\eps})}}$ gates proving the gate complexity bound.
  \pnote{Joran's first analysis: 
  Regarding the gate complexity, using \autoref{lem:linCombBlocks} and \autoref{lem:ampBlock}, together with a construction similar to \autoref{lem:uamp} for $\frac{1}{\sqrt{\nrm{\gamma}_1}} \sum_j \sqrt{x_j}\ket{j} \ket{0} + \ket{\psi}\ket{1}$, we see that it is $\bigO{\nu/\sigma \log(\nu/\eps)m \log \nrm{\gamma}_{\infty}}$ and the circuit depth is  $\bigO{\nu/\sigma \log(\nu/\eps) (\log m +\log \nrm{\gamma}_{\infty})}$, where the term $\log \nrm{\gamma}_\infty$ comes from the number of bits necessary to represent each entry of $x$.}
\end{proof}

\begin{theorem}
  \label{thm:expectationValues}
  Let $E_1,\dots, E_m \in \C^{d\times d}$ be Hermitian matrices with $\nrm{E_j}\leq 1$, and let $U_E=\sum_{j\in [m]}\ketbra{j}{j}\otimes U_{E_j}$, where $U_{E_j}$ is an $a_E$-block-encoding of $E_j$.   
  Let $\delta\in(0,\frac{1}{6}]$, $\eps_1,\dots,\eps_m\in(0,2]^m$ be error bounds, $\nu = \sum_j \frac{1}{\eps_j}$, $\sigma \geq\max\left\{\sqrt{2\nrm{\sum_j E_j^2/\eps_j^2}\ln\left(\frac{2d}{\delta}\right)},1\right\}$, and $\sigma':=\min\{\nu,\sigma\}$. Let $U_{\rho}$ be an $a_\rho$-qubit state-preparation unitary for a purification of $\rho\in\C^{d \times d}$. There is a quantum algorithm that makes
  $
  \bigO{\left(\sigma'+\log(\frac{\log(m)}{\delta})\right)\log(\frac{m}{\delta})}
  $
  queries to $U_{\rho}$ and $U_{\rho}^{\dag}$, and produces estimates $\pmb{z}\in\left[-1,1\right]^m$ such that, with probability at least $1-\delta$,
  \begin{equation}\label{eq:boostedGradientMult}
 	\forall j\in [m]\colon |\tr{\rho E_j}-z_j|\leq \eps_j,
  \end{equation}
  moreover 	
  \begin{equation}\label{eq:suppressedGradientMult}
  	\forall j\in [m]\colon |\tr{\rho E_j}-\mathbb{E}[z_j]|\leq 16\sigma' \eps_j \delta.
  \end{equation}
  Furthermore, the quantum algorithm can be implemented by a number of calls to $U_E$ bounded by $\bigO{\nu \log(\frac{\nu\log(m)}{\delta})\log(\frac{m}{\delta})+\frac{\nu}{\sigma'} \log(\frac{\nu\log(m)}{\delta})\log(\frac{m}{\delta})\log(\frac{\log(m)}{\delta})}$,
  and additional number of two-qubit gates bounded by $\bigO{\left(\sigma' a_\rho+\nu a_E+\nu m\right) \polylog(\frac{\nu m}{\delta})}$ and having depth $\bigO{\nu\polylog(\frac{\nu m}{\delta})}$.
  
  Finally, there is a random variable $\pmb{z'}\in \bigtimes_{j\in [m]}\left[\tr{\rho E_j}-\eps_j,\tr{\rho E_j}+\eps_j\right]$ with independent coordinates that is $\delta$-close in total variation distance to $\pmb{z}$ and also satisfies \autoref{eq:suppressedGradientMult}.
\end{theorem}
\begin{proof}
	The main idea is to apply Jordan's gradient estimation algorithm to a linear function with derivative vector $\pmb{g}$ such that $g_j=\frac{1}{\sigma'}\tr{\rho \frac{E_j}{\eps_j}}$ with accuracy $\eps':=\min\{\frac{1}{\sigma'},\frac{1}{6}\}$.
	
	Let $b:=\lceil \log_2(16/\eps') \rceil$ and let $\beta:= \frac{\delta}{96\lceil\ln(6m/\delta)\rceil+12}$.
	If $\sigma'< \nu$, then we use \autoref{lem:linCombEj} in order to construct a
	unitary $V=\sum_{\pmb{x}\in G_b^m} V_x \otimes\ketbra{\pmb{x}}{\pmb{x}}$ such that $V_x$ is a $c:=(a_E+\ceil{\log_2(m)}+2)$-block-encoding of a matrix $A_x$ that is $\frac{\beta}{4\sigma'\pi}$-close in operator norm to $\frac{1}{\sigma'}\sum \frac{x_j E_j}{\eps_j}$ for at least a $1-\beta^2$ fraction of points $\pmb{x}\in G_b^m$. Otherwise, when $\sigma'= \nu$ then we simply apply the first step in the algorithm of \autoref{lem:linCombEj}, namely \autoref{lem:linCombBlocks}.
	
	We then define $V_\ell:=I_{c}\otimes U_{\rho}\otimes I_{G_b^d}$ and $V_r:=(V\otimes I_P)\cdot V_\ell$, where $I_P$ acts on the purifying register of $U_\rho$. 
	Let use the notation $\ket{\rho}_{PS}:=U_\rho\ket{0}$. Since by definition $\Tr_P\left(\ketbra{\rho}{\rho}_{PS}\right)=\rho$, we have
	\begin{align*}
		\bra{00x}V_\ell^\dagger V_r\ket{00y}
		&=\bra{0}\bra{\rho}_{PS}\bra{x}(V\otimes I_P)\ket{0}\ket{\rho}_{PS}\ket{y}\\
		&=\delta_{xy}\bra{0}\bra{\rho}_{PS}(V_x\otimes I_P)\ket{0}\ket{\rho}_{PS}\\
		&=\delta_{xy}\bra{\rho}_{PS}(A_x\otimes I_P)\ket{\rho}_{PS}\\
		&=\delta_{xy}\tr{\bra{\rho}_{PS}(A_x\otimes I_P)\ket{\rho}_{PS}}\\	
		&=\delta_{xy}\tr{(A_x\otimes I_P)\ketbra{\rho}{\rho}_{PS}}\\
		&=\delta_{xy}\tr{A_x\rho}.	
	\end{align*}
	Since $\nrm{A_x-\frac{1}{\sigma'}\sum \frac{x_j E_j}{\eps_j}}\leq \frac{\beta}{4\sigma'\pi}$ for at least a $1-\beta^2$ fraction of points $\pmb{x}\in G_b^m$ we get that $\left|\Tr_P\left(\rho A_x \right)-\Tr_P\left(\frac{\rho}{\sigma'}\sum \frac{x_j E_j}{\eps_j} \right)\right|\leq \frac{\beta}{4\sigma'\pi}$ also holds for these points. Thus $W:=V_\ell^\dagger V_r$ is an $(a_\rho+c+1)$-block-encoding of $f(\pmb{x})$ that is $\frac{\beta}{4\sigma'\pi}$-close to $\sum_{j\in[m]}\frac{1}{\sigma'}\tr{\rho \frac{x_j E_j}{\eps_j}}$ for at least a $1-\beta^2$ fraction of points $\pmb{x}\in G_b^m$. Then \autoref{eq:boostedGradientMult} follows from \autoref{cor:unbiasedBlockToGrad} after multiplying its output coordinate-wise by $\sigma' \eps_j$ and truncating to $[-1,1]$. Similarly, \autoref{eq:suppressedGradientMult} follows from \autoref{cor:unbiasedBlockToGrad} after incrementing the bias by $8\sigma' \eps_j$ taking into account the truncation error.
	
	The query complexity 
	$\bigO{\left(\sigma'+\log(\frac{\log(m)}{\delta})\right)\log(\frac{m}{\delta})}$ for $U_\rho$ directly follows from \autoref{cor:unbiasedBlockToGrad}.
	The gate complexity of \autoref{cor:unbiasedBlockToGrad} is 
	$\bigO{\left(m+\sigma' (a_\rho+a_E+1)\right) \polylog(\frac{\sigma' m}{\delta})}$, which is supplemented by the complexity of implementing $V$ times the above query complexity. The implementation of $V$ uses $\bigO{\ceil{\frac{\nu}{\sigma'} \log(\frac{\nu}{\beta})}}=\bigO{\frac{\nu}{\sigma'} \log(\frac{\nu\log(m)}{\delta})}$ calls to $U_E$, and   
	$\bigO{(a_E\!+\!m)\ceil{\frac{\nu}{\sigma'}}\polylog(\frac{\nu+m}{\sigma'\eps})}\!$ $=\bigO{(a_E\!+\!m)\frac{\nu}{\sigma'}\polylog(\frac{\nu m}{\delta})}$ additional two-qubit gates having depth $\bigO{\frac{\nu}{\sigma'}\polylog(\frac{\nu m}{\delta})}$. This amounts to a total of $\bigO{\nu \log(\frac{\nu\log(m)}{\delta})\log(\frac{m}{\delta})+\frac{\nu}{\sigma'} \log(\frac{\nu\log(m)}{\delta})\log(\frac{\log(m)}{\delta})\log(\frac{m}{\delta})}$ calls to $U_E$.
\end{proof}

Note that the assumption $\nrm{E_j} \le 1$ is not particularly restrictive, because if $\nrm{E_j} > 1$ the corresponding block-encoding is subnormalized and we simply need to increase the precision by an amount equal to the subnormalization factor. If all $\eps_j$ are equal and we use the assumption $\nrm{E_j} \le 1$, we recover the sample complexity $\bigOt{\sqrt{m}/\eps}$ of the algorithm in \cite{huggins2021QAlgMultipleExpectationValues}. The number of calls to $U_{E}$ is not directly comparable because we use a different input model: in \cite{huggins2021QAlgMultipleExpectationValues} the algorithm assumes access to to $e^{-i \theta E_j}$ and requires $\bigOt{\sqrt{m}/\eps}$ calls to each of these operators for $j=1,\dots,m$, while we give a version that uses $\bigOt{m/\eps}$ calls in total to controlled unitaries $U_{E_j}$ block-encoding $E_j$.

Furthermore, \autoref{thm:expectationValues} also recovers the query complexity results of the probability distribution estimation problem from~\cite{apeldoorn2021QProbOraclesMulitDimAmpEst}, by taking $E_j = \ket{j}\bra{j}$, for $j \in [d]$, and observing that $\sum_{j=1}^d E_j^2 = I$. Thus, even though \cite{apeldoorn2021QProbOraclesMulitDimAmpEst} and \cite{huggins2021QAlgMultipleExpectationValues} seem to be of different flavor, this result unifies both into a single construction.

For convenience, we state a version of our result only in terms of the number of observables rather than the more involved quantity $\nrm{\sum_{j=1}^d E_j^2}$.

\begin{corollary}
	Let $E_1,\dots, E_m \in \C^{d\times d}$ be Hermitian matrices with $\nrm{E_j}\leq 1$, and let $U_E=\sum_{j\in [m]}\ketbra{j}{j}\otimes U_{E_j}$, where $U_{E_j}$ is an $a_E$-block-encoding of $E_j$, and let $\delta,\eps\in(0,\frac{1}{6}]$. Let $U_{\rho}$ be an $a_\rho$-qubit state-preparation unitary for a purification of $\rho\in\C^{d \times d}$. There is a quantum algorithm that makes
	$
	\bigO{\left(\frac{\sqrt{m\log\left(\frac{d}{\delta}\right)}}{\eps}+\log(\frac{\log(m)}{\delta})\right)\log(\frac{m}{\delta})}
	$
	queries to $U_{\rho}$ and $U_{\rho}^{\dag}$, and produces estimates $\pmb{z}\in\left[-1,1\right]^m$ such that, with probability at least $1-\delta$,
	\begin{equation}\label{eq:boostedGradientMultGoo}
	\forall j\in [m]\colon |\tr{\rho E_j}-z_j|\leq \eps,
	\end{equation}
	moreover 	
	\begin{equation}\label{eq:suppressedGradientMultGoo}
	\forall j\in [m]\colon |\tr{\rho E_j}-\mathbb{E}[z_j]|\leq 16 \sqrt{2m\log\left(\frac{2d}{\delta}\right)} \delta.
	\end{equation}
	Furthermore, the quantum algorithm can be implemented by a number of calls to $U_E$ bounded by $\bigO{\frac{m}{\eps} \log(\frac{m}{\eps\delta})\log(\frac{m}{\delta})+\sqrt{m} \log(\frac{m}{\eps\delta})\log(\frac{m}{\delta})\log(\frac{1}{\delta})}$
	and
	$\bigO{\left(\frac{\sqrt{m\log\left(\frac{d}{\delta}\right)}}{\eps} a_\rho+\frac{m}{\eps}a_E+\frac{m^2}{\eps}\right) \polylog(\frac{m}{\eps\delta})}\!$
	additional two-qubit gates while having circuit depth $\bigO{\frac{m}{\eps}\polylog(\frac{m}{\eps\delta})}$.
	
	Finally, there is a random variable $\pmb{z'}\in \bigtimes_{j\in [m]}\left[\tr{\rho E_j}-\eps,\tr{\rho E_j}+\eps\right]$ with independent coordinates that is $\delta$-close in total variation distance to $\pmb{z}$ and also satisfies \autoref{eq:suppressedGradientMultGoo}.
\end{corollary}

\section{Mixed-state tomography}\label{s:mixed-state-tomo}
In this section we generalize our pure-state results to mixed states. 
Throughout this section we use $r$ to denote the rank of the mixed state. As discussed in the introduction, results from the literature on mixed-state tomography consider the case where only copies of the mixed state are available. Gross et al.~\cite{gross2010} give an algorithm that uses $\bigO{d^2 r^2/\eps^2}$ samples; a tighter analysis of their algorithm shows that $\bigO{d r^2/\eps^2}$ suffice to get an $\eps$-trace-norm estimate, when measurements are performed on single (i.e., unentangled) copies of the state \cite{haah2017OptTomography}, and \cite{chen2022tight} shows that this is optimal even for adaptive (but still unentangled) measurements. Haah et al.~\cite{haah2017OptTomography} and O'Donnell and Wright \cite{odonnell2016EfficientQuantumTomography} further improve the sample complexity to $\bigOt{dr/\eps^2}$, at the cost of requiring joint measurements on many states at once, and with an algorithm that has super-polynomial time complexity. \cite{haah2017OptTomography} also shows matching lower bounds for both settings, up to polylogarithmic factors, so these complexities are essentially optimal.  

We consider a stronger input model where we are given access to a purification of a mixed state.
Assume that we are interested in a rank-$r$ mixed state $\rho = \sum_{j=1}^r p_j \ketbra{\psi_j}{\psi_j}$ for some orthonormal $\ket{\psi_j}$. A purification of $\rho$ is a state on two registers, $A$ and $B$ that can be written as
\[
    \ket{\rho} = \sum_{j=1}^r \sqrt{p_j} \ket{\psi_j}_A\ket{\phi_j}_B,
\]
for some orthonormal $\{\ket{\phi_j}\}$ given by Schmidt's decomposition. Note that tracing out the $B$ register yields just $\rho$. Note that there are many possible purifications for the same mixed state. 

The simplest idea to use a purification for tomography is to apply our pure-state algorithms directly to the purification, and then post-process by tracing out the unwanted part. The following lemma relates the error in a pure-state estimate to that in the resulting mixed-state estimate.
\begin{lemma}
\label{lem:infty_norm_to_trace}
  Let $\ket{\psi} := \sum_{j \in [r]} \sqrt{p_j} \ket{\psi_j}_A \ket{\phi_j}_B$ and $\rho = \Tr_B\trm{\ket{\psi}\bra{\psi}}$. Let $d,s$ be the Hilbert space dimensions of subsystems $A$ and $B$ respectively. Let $\ket{\tilde{\psi}} : \nrm{\ket{\psi} - \ket{\tilde{\psi}}}_\infty \le \eps/\sqrt{ds}$, and let $\tilde{\rho} = \Tr_B\trm{ \ket{\tilde{\psi}} \bra{\tilde{\psi}}}$. Then $\frac{1}{2}\nrm{\rho - \tilde{\rho}}_1 \leq \eps$.
\end{lemma}
\begin{proof}
By a standard norm conversion, as in \autoref{cor:norm-conversion}, we have $\nrm{\ket{\psi}-\ket{\tilde{\psi}}}_2\leq \eps$.
Hence:
\begin{align*}
    \frac{1}{2}\nrm{\ket{\psi}\bra{\psi}- \ket{\tilde{\psi}}\bra{\tilde{\psi}}}_1 &=\sqrt{1-|\braket{\psi}{\tilde{\psi}}|^2}\\
    &= \sqrt{1-|\braket{\psi}{\tilde{\psi}}|}\sqrt{1+|\braket{\psi}{\tilde{\psi}}|}\\
    &\leq \sqrt{2} \sqrt{1-\Re \braket{\psi}{\tilde{\psi}}}\\
    &=  \nrm{\ket{\psi}-\ket{\tilde{\psi}}}_2\\
    &\leq \eps,
  \end{align*}
  were we used the fact that, for pure states, $\frac{1}{2}\nrm{\ketbra{\psi}{\psi}- \ketbra{\phi}{\phi}}  = \sqrt{1-|\braket{\psi}{\phi}|^2}$ and $\nrm{\ket{\psi}-\ket{\phi}}_2 = \sqrt{2-2\Re \braket{\psi}{\phi}}$. Note that for any matrix $M$, the relationship $\nrm{\Tr_B(M)}_1\leq \nrm{M}_1$ holds, see, e.g., \cite{rastegin2012relations} for a proof. Applying this to our pure state, and using the linearity of the trace, we find: 
  \begin{align*}
  \frac{1}{2}\nrm{\rho - \tilde{\rho}}_1 &=  \frac{1}{2} \nrm{\Tr_B(\ketbra{\psi}{\psi}) - \Tr_B(\ketbra{\tilde{\psi}}{\tilde{\psi}})}_1 \\
  &=  \frac{1}{2}\nrm{\Tr_B(\ketbra{\psi}{\psi} - \ketbra{\tilde{\psi}}{\tilde{\psi}})}_1 \\
  &\leq  \frac{1}{2}\nrm{\ketbra{\psi}{\psi} - \ketbra{\tilde{\psi}}{\tilde{\psi}}}_1 \\
  &\leq \eps.\qedhere
  \end{align*}
  \end{proof}

If the purifying register is of size $s$ then our sampling and phase estimation algorithms would get a sample complexity of $\bigOt{\frac{ds}{\eps^2}}$ and query complexity of $\bigOt{\frac{ds}{\eps}}$ respectively to obtain a trace-norm estimate. As $s$ can be as small as $r$ in certain settings, this might lead to interesting results in certain settings, but in general one can not upper-bound the size of $s$. 

In the rest of this section we describe a tomography algorithm with sample complexity $\bigOt{\frac{dr}{\eps}}$, for trace norm error $\eps$, when a unitary (and its inverse) preparing a purification of $\rho$ is available. 

\subsection{Coordinate-wise unbiased tomography}
Applying \autoref{thm:expectationValues} to the set of observables $E^{(0)}_{ij}:=\frac{\ketbra{i}{j}+\ketbra{j}{i}}{2}$ and $E^{(1)}_{ij}:=\frac{\ketbra{i}{j}-\ketbra{j}{i}}{2\ci}$ we get the next result. 
\begin{theorem}	\label{thm:coordianteWiseDensTomo}
	Let $\eps,\delta\in(0,\frac{1}{6}]$, and let $U_{\rho}$ be an $a$-qubit state-preparation unitary for a purification of $\rho\in\C^{d \times d}$. There is a quantum algorithm that makes
	$
	\bigO{\left(\frac{\sqrt{d\log\left(\frac{d}{\delta}\right)}}{\eps}+\log(\frac{\log(d)}{\delta})\right)\log(\frac{d}{\delta})}
	$
	queries to $U_{\rho}$ and $U_{\rho}^{\dag}$, and produces estimates $\pmb{z}\in\left([-1,1]\times[-\ci,\ci] \right)^{d^2}$ such that, with probability at least $1-\delta$,
	\begin{equation}\label{eq:boostedGradientCoor}
	\forall i,j\in [d]\colon |\rho_{ij}-z_{ij}|\leq \sqrt{2}\eps,
	\end{equation}
	moreover 	
	\begin{equation}\label{eq:suppressedGradientCoor}
		\forall i,j\in [d]\colon |\rho_{ij}-\mathbb{E}[z_{ij}]|\leq 16\sqrt{2}\sqrt{2d\ln\left(\frac{2d}{\delta}\right)} \delta
		\leq32\sqrt{d\ln\left(2d\right)\delta}.
	\end{equation}
	Furthermore, the quantum algorithm can be implemented by  $\bigO{\frac{\sqrt{d} a+d^3}{\eps} \polylog(\frac{d}{\delta})}$ additional two-qubit gates having depth $\bigO{\frac{d}{\eps}\polylog(\frac{d}{\delta})}$.
	
	Finally, there is a random variable $\pmb{z'}\in \bigtimes_{i,j\in [d]}\left([\rho_{ij}-\eps,\rho_{ij}+\eps]\times[\rho_{ij}-\ci\eps,\rho_{ij}+\ci\eps]\right)$ with independent coordinates that is $\delta$-close in total variation distance to $\pmb{z}$ and also satisfies \autoref{eq:suppressedGradientCoor}.	
\end{theorem}
\begin{proof}
	The result follows from \autoref{thm:expectationValues} by observing that
	$\sum_{p\in\{0,1\};i,j\in[d]}(E^{(p)}_{ij})^2=d I$, and a block-encoding of $\sum_{p\in\{0,1\};i,j\in[d]}\ketbra{pij}{pij}\otimes E^{(p)}_{ij}$ can be implemented by $\bigO{\polylog(d)}$ two-qubit gates. We define our estimate as $z_{i,j}:=z_{i,j}^{(0)}+\ci z_{i,j}^{(1)}$.
	
	The second inequality in \autoref{eq:suppressedGradientCoor} follows from the following little computation:
	\begin{align*}
	\sqrt{\ln\left(\frac{2d}{\delta}\right)} \delta
	&\leq \sqrt{\ln\left(2d\right)\delta}\\
	&\Updownarrow\\
	\ln\left(\frac{2d}{\delta}\right)\delta
	&\leq \ln\left(2d\right)\\
	&\Uparrow (\delta\leq \frac16)\\
	\ln\left(\frac{1}{\delta}\right)\delta
	&\leq \frac{5}{6}\ln\left(2d\right)\\
	&\Uparrow (d\geq 1)\\	
	\ln\left(\frac{1}{\delta}\right)\delta
	&\leq \frac{1}{2}\qedhere
	\end{align*}
	The gate complexities follow by replacing \autoref{lem:linCombEj} in the proof or \autoref{thm:expectationValues} by ``sparse block-encoding'' \cite[Lemma 47-48]{gilyen2018QSingValTransfArXiv}. This results in reducing the $\Theta(d^2)$ subnormalization factor coming from the generic result of \autoref{lem:linCombEj} by a factor of $d$ coming from \cite[Lemma 47-48]{gilyen2018QSingValTransfArXiv}.
	\pnote{Add exact implementation with comparison gates.}
	This improves the gate complexities by about a $d$ factor.\footnote{There is a possibility that using the block-encoding of \cite{low2018HamSimNearlyOptSpecNorm} even an about $d^{1.5}$ factor improvement is possible, but one needs to be careful with the error bounds, since they are not poly-logarithmic in \cite{low2018HamSimNearlyOptSpecNorm}.}
\end{proof}

\subsection{Matrix norm conversions}

We now consider the relation between the element-wise $\max$-norm, and the operator norm. To do so we use the following definition and lemma due to~\cite{rudelson2010NonasymptoticTheoryRandMat}:
\begin{definition}[Subgaussian random variable~{\cite[Definition 2.2]{rudelson2010NonasymptoticTheoryRandMat}}]
	A random variable $X$ is \emph{subgaussian} if there exists a $K>0$, called the subgaussian moment of $X$, such that
	\[
	\Prob{|X|>t}\leq 2e^{-t^2/K^2} \text{ for all } t>0.
	\]
\end{definition}

Note that a bounded random variable $X\in[-B,B]$ has subgaussian moment $\leq \sqrt{\ln(2)}B$.

\begin{lemma}[Operator norm of subgaussian matrices~{\cite[Proposition 2.4]{rudelson2010NonasymptoticTheoryRandMat}}]\label{lem:matrixConcentrate}
	Let $X$ be an $N \times n$ random matrix whose entries are independent mean zero subgaussian random variables whose subgaussian moments are bounded by $K$. Then
	\[
	\Prob{\frac{\nrm{X}}{K}> C(\sqrt{N} + \sqrt{n})+t} \leq 2e^{-ct^2}, t \geq 0,
	\]
	where $C$ and $c$ denote positive absolute constants.
\end{lemma}

The above lemma shows that if we can estimate all entries of a state $\rho$ independently and in an unbiased way, then in the conversion to the operator norm error we save an essentially $\sqrt{d}$ factor compared to the worst case: this was one of the main motivation for us to develop unbiased phase estimation in \autoref{sec:unbiasedPhaseEst}. We formalize this in the following lemma:

\begin{lemma}
	\label{lem:maxToOperator}
	Let $X\in \C^{d\times d}$ be a matrix, and let $\tilde{X}$ be an $\eps$-approximation of $X$ in the entry-wise $\max$-norm, i.e., $|X_{ij} - \tilde{X}_{ij}| \le \eps$ for each $i,j\in[d]$. Then for the operator norm error we have $\nrm{X-\tilde{X}} \leq \eps d$. Also, there are absolute constants $C',c'>0$ such that, if all entries of $\tilde{X}$ are drawn from independent distributions and $\mathbb{E}[\tilde{X}]=X$, then for every $\tau\geq 1$ we have that  $\nrm{X-\tilde{X}} \leq C' \sqrt{d}\tau \eps$ with probability at least $1-2e^{-c'd\tau^2}$. 
\end{lemma}
\begin{proof} 
	Let $E = \tilde{X}-X$ be the matrix of errors. For the first statement we note that
	\[ \nrm{E} \leq \nrm{E}_2 \leq d\nrm{E}_{\max} \leq \eps d, \]
	where the first inequality follows from the relation between the operator and the Frobenius norms, and the second follows from a standard norm conversion on the $d^2$-dimensional vector of entries (we use $\nrm{\cdot}_{\max}$ for the entry-wise $\max$-norm).

	For the second statement we apply \autoref{lem:matrixConcentrate} to $E$, noting that each matrix element has subgaussian moment $\leq \sqrt{\ln(2)}\eps$, thus implying that for every $t\geq 0$ we have
	\[
	\Prob{\nrm{E} > \sqrt{\ln(2)}\eps\trm{2C\sqrt{d}+t}} \leq 2e^{-ct^2}.
	\]
	We conclude by setting $t:= 2C\sqrt{d}\tau$ showing that
	\[
	\Prob{\nrm{E} > (2+2\tau)\sqrt{\ln(2)}C\sqrt{d}\eps} \leq
	\Prob{\nrm{E} > 4\sqrt{\ln(2)}C\sqrt{d}\tau\eps} \leq 2e^{-4cC^2d\tau^2},
	\]
	so that we can choose $C':=4\sqrt{\ln(2)}C$ and $c':=4cC^2$.
\end{proof}

When considering estimates of mixed quantum states we mostly consider Schatten $q$-norms for error bounds (that is, the $q$-norm of the vector of singular values of the difference between the actual state and our estimate). The most common values for $q$ are $q=\infty$ (operator norm), $q=2$ (Frobenius norm), and $q=1$ (trace norm). Using a tiny modification of our norm conversion result, \autoref{cor:norm-conversion}, we can obtain the following as a corollary.

\begin{corollary}\label{lem:operatorToSchatten}
    Let $\eps \in (0,1]$, $1 \leq q$, and let $\rho \in \C^{d \times d}$ be a rank-$r$ quantum state. In order to obtain an $\eps$-Schatten-$q$-norm estimate of $\rho$, an $\eta$-operator norm estimate suffices, with
    \[
    \eta = \max\left\{\left(\frac{\eps}{10}\right)^{\frac{1}{1-\frac1q}}, \frac{\eps}{2(2r)^{\frac1q}}\right\}.
    \]
\end{corollary}

\begin{proof}
    Let $\tilde{\rho}$ be an $\eta$-operator norm estimate of $\rho$.
    We can assume without loss of generality that $\tilde{\rho}$ is Hermitian (otherwise take $\frac{\tilde{\rho}+\tilde{\rho}^\dagger}{2}$). First, since $\nrm{\tilde{\rho} - \rho} \leq \eta$, there must exist a density matrix $\sigma$ such that $\sigma \succeq 0$, $\nrm{\sigma}_1 \leq 1$, $\sigma$ is of rank at most $r$, and $\nrm{\tilde{\rho} - \sigma} \leq \eta$, because after all $\rho$ is an example of such a density matrix $\sigma$. Let $\tilde{\rho}'$ be any such $\sigma$.\footnote{Removing all negative eigenvalues of $\tilde{\rho}-\eta I$ produces such a matrix $\tilde{\rho}'$. Clearly, $\tilde{\rho}-\rho\preceq \eta I$ and so $\tilde{\rho}-\eta I\preceq \rho$, implying that the $n$-th largest eigenvalue of $\rho$ majorates that of $\tilde{\rho}-\eta I$ and consequently also that of $\tilde{\rho}'$. This then implies that the rank of $\tilde{\rho}'$ is at most $r$ and that $\nrm{\tilde{\rho}'}_1\leq \nrm{\rho}_1=1$; $\nrm{\tilde{\rho}'-\tilde{\rho}}\leq \eta$ can be shown similarly.} Then, by the triangle inequality we obtain that $\nrm{\tilde{\rho}' - \rho} \leq \nrm{\tilde{\rho}' - \tilde{\rho}} + \nrm{\tilde{\rho} - \rho} \leq \eta + \eta = 2\eta$.
    
    Thus, we find that $\rho$ and $\tilde{\rho}'$ are both of rank at most $r$, and therefore by the subadditivity of rank, $\rho - \tilde{\rho}'$ is of rank at most $2r$. This implies by H\"older's inequality that $\nrm{\rho - \tilde{\rho}'}_q \leq (2r)^{\frac1q} \cdot 2\eta$.
    
    On the other hand, from the norm conversion lemma, \autoref{lem:new-norm-conversion}, there exist operators $\tilde{\rho}'_{\geq 2\eta}$ and $\rho_{\geq 2\eta}$ such that they are both $\min\{4\eta^{(q-s)/q}, 3r^{1/q}\eta\}$-close to their originals in Schatten-$q$-norm and both have rank at most $1/(2\eta)$. Then, we obtain by the triangle inequality that
    \[
    \nrm{\tilde{\rho}' - \rho}_q \leq \nrm{\tilde{\rho}' - \tilde{\rho}'_{\geq2\eta}}_q + \nrm{\tilde{\rho}'_{\geq2\eta} - \rho_{\geq2\eta}}_q + \nrm{\rho_{\geq2\eta} - \rho}_q \leq 2 \cdot 4\eta^{\frac{q-1}{q}} + \left(\frac{1}{\eta}\right)^{\frac{1}{q}}2\eta = 10\eta^{\frac{q-1}{q}}.
    \]
    Combining both results yields $\nrm{\tilde{\rho}' - \rho}_q \leq \min\{10\eta^{(q-1)/q},(2r)^{\frac1q}2\eta\} = \eps$.	
\end{proof}

\subsection{Generic mixed-state tomography}\label{subsec:mixed-state-tomo}

We now have all the necessary tools to construct a tomography algorithm with $\bigOt{dr/\eps}$ sample complexity.

\begin{theorem}
	\label{thm:mixedTomoSch}
	Let $\eps,\delta\in(0,\frac{1}{3}]$, $q\in[1,\infty]$, and let $U_{\rho}$ be an $a$-qubit state-preparation unitary for a purification of $\rho\in\C^{d \times d}$. There is a quantum algorithm that makes
	$
	\bigO{\frac{d}{\eps}r^{\frac{1}{q}}\log^\frac{3}{2}\left(\frac{d}{\delta\eps}\right)\sqrt{\ceil{\frac{\log(1/\delta)}{d}}}}
	$
	queries to $U_{\rho}$ and $U_{\rho}^{\dag}$, and outputs a positive semidefinite $\tilde{\rho}'$ such that with probability at least $1-\delta$ we have $\nrm{\rho-\tilde{\rho}'}_q\leq \eps$.
	The quantum algorithm can be implemented by 
	$\bigO{\!\frac{d a+d^{3.5}}{\eps}r^{\frac{1}{q}} \polylog(\frac{d}{\delta\eps})}\!$ additional two-qubit gates having depth $\bigO{\frac{d^{1.5}}{\eps}r^{\frac{1}{q}}\polylog(\frac{d}{\delta\eps})}$.
\end{theorem}
\begin{proof}
	First we prove the statement for operator norm by combining \autoref{thm:coordianteWiseDensTomo} and 	\autoref{lem:maxToOperator}.
	We set $\delta':=\min\{\frac{\delta}{2},\frac{\eps^2}{2^{11} d^3 \ln(2d)}\}$,
	$\tau:=\sqrt{1+\frac{\ln(4/\delta)}{d c'}}$, and 
	$\eps':=\min\left\{\frac{\eps}{4 C' \sqrt{d} \tau},\frac16\right\}$ and invoke \autoref{thm:coordianteWiseDensTomo} providing us an estimate 
	$\pmb{\tilde{\rho}}\in\C^{d^2}$ that is $\frac{\delta}{2}$-close in total variation distance to a random variable $\pmb{\rho'}\in \bigtimes_{i,j\in [d]}\left([\rho_{ij}-\eps,\rho_{ij}+\eps]\times[\rho_{ij}-\ci\eps,\rho_{ij}+\ci\eps]\right)$. Due to \eqref{eq:suppressedGradientCoor} we have that $\mathbb{E}[\rho'_{ij}]-\rho_{ij}\leq \frac{\eps}{2d}$ and so by \autoref{lem:maxToOperator} we have $\nrm{\mathbb{E}[\rho']-\rho}\leq \frac{\eps}{2}$. Using our choice of $\tau$ \autoref{lem:maxToOperator} also implies that $\nrm{\mathbb{E}[\rho']-\rho'} \leq\frac{\eps}{2}$ with probability at least $1-\frac{\delta}{2}$. By the triangle inequality we get that $\nrm{\rho-\rho'} \leq\eps$ with probability at least $1-\frac{\delta}{2}$. Since $\pmb{\rho'}$ and $\pmb{\tilde{\rho}}$ are $\frac{\delta}{2}$-close in total variation distance this also implies that $\nrm{\rho-\tilde{\rho}} \leq\eps$ with probability at least $1-\delta$.
	
	As per \autoref{thm:coordianteWiseDensTomo} the algorithm makes
	$
	\bigO{\left(\frac{d\sqrt{\log\left(\frac{d}{\delta\eps}\right)\left(1+\frac{\log(1/\delta)}{d}\right)}}{\eps}+\log\left(\frac{d}{\delta\eps}\right)\right)\log\left(\frac{d}{\delta\eps}\right)}
	$
	queries to $U_{\rho}$ and $U_{\rho}^{\dag}$, and can be implemented by $\bigO{\frac{d a+d^{3.5}}{\eps} \polylog(\frac{d}{\delta\eps})}$ additional two-qubit gates having depth $\bigO{\frac{d^{1.5}}{\eps}\polylog(\frac{d}{\delta\eps})}$.
	
	In order to get a positive semidefinite $\tilde{\rho}'$ and to get the results for all Schatten-norms we apply \autoref{lem:operatorToSchatten} to $\frac{\tilde{\rho}+\tilde{\rho}^\dagger}{2}$ and adjust the value of $\eps$ accordingly.
\end{proof}

\section{Lower bounds}
In this section we prove lower bounds for state tomography in several different access models. The first model we consider, in \autoref{sec:sampling-lb}, is the case in which we have access to conditional copies of the state, i.e., we receive states of the form $(\ket{0}\ket{\psi}+\ket{1}\ket{0})/\sqrt{2}$. From here, we derive matching lower bounds for all the algorithms constructed in \autoref{s:classicaltomo}. In the second model, considered in \autoref{subsec:lb-with-inverses}, we assume to have access to a state-preparation unitary and its inverse. The lower bounds derived in this subsection match the complexities of the algorithm constructed in \autoref{s:quantumtomo}, up to logarithmic factors. Finally, in \autoref{subsec:lb-mixed-state}, we consider the setting where we access to a unitary constructing a purification of a density matrix that we wish to estimate. We derive a lower bound when the desired precision is w.r.t.\ the Frobenius norm, and it matches the complexity of the algorithm constructed in \autoref{s:mixed-state-tomo}.

Note that van Apeldoorn~\cite{apeldoorn2021QProbOraclesMulitDimAmpEst} gives a very similar lower bound result in the pure-state setting where we have access to a state-preparing unitary. However, in the setting of van Apeldoorn the unitary prepares a state of the form $\sum_j \sqrt{p_j} \ket{j} \ket{\phi_j}$, whereas in this paper the state is of the form $\sum_j \sqrt{p_j} \ket{j}$, i.e., without the additional states entangled with $\ket{j}$. Hence our input model is stricter and requires its own lower bound.

The general proof strategy in \autoref{sec:sampling-lb} and \autoref{subsec:lb-with-inverses} is very similar -- we start by proving a lower bound on estimating probability distributions in the $\ell_1$-norm, then use a sequence of reductions to obtain lower bounds on quantum pure-state tomography in any $\ell_q$-norm with $q \geq 2$. Since the reductions we use are identical in both cases, we start by presenting it here, and then focus in \autoref{sec:sampling-lb} and \autoref{subsec:lb-with-inverses} on proving the lower bound for probability distribution reconstruction in $\ell_1$-norm separately for both input models afterwards.

\begin{lemma}
\label{lem:eps_ellq_prob_est}
  Let $0 < \eps \leq 1$, $s > 0$, $q \in [1,\infty]$, and suppose that in order to produce an $\eps$-$\ell_1$-estimate of any probability distribution $p \in \Delta_d$ with high probability, one needs to perform at least $\Omega(d/\eps^s)$ queries. Then
  \[
 \Omega\trm{\min\left\{\frac{1}{\eps^{\frac{1}{1-\frac1q}}},\frac{d^{1-s+\frac{s}{q}}}{\eps^s} \right\}}
  \]
  queries are necessary to find a $\eps$-$\ell_q$-estimate of $p$ with high probability.
\end{lemma}

\begin{proof}
  Suppose that $q > 1$ and $\varepsilon \leq 1/d^{1-\frac1q}$. Then, by H\"older's inequality,
  \[
  \nrm{\widetilde{p} - p}_1 \leq d^{1-\frac1q} \nrm{\widetilde{p} - p}_q \leq d^{1-\frac1q}\varepsilon \leq 1,
  \]
  and the number of queries that is required scales as
  \[
  \Omega\trm{\frac{d}{\trm{d^{1-\frac1q}\varepsilon}^s}} = \Omega\trm{\frac{d^{1-s+\frac{s}{q}}}{\varepsilon^s}}.
  \]
  
  This leaves the case where $q > 1$ and $1/d^{1-\frac1q} < \eps \leq 1$. This immediately implies that
  \[
    1 \leq \frac{1}{\eps^{\frac{1}{1-\frac1q}}} < d^{\frac{1-\frac1q}{1-\frac1q}} = d.
  \]
  Thus, we can choose
  \[
    d' = \left\lfloor \frac{1}{\eps^{\frac{1}{1-\frac1q}}} \right\rfloor,
  \]
  and observe that it is an integer between $1$ and $d$. It follows that
  \[
    (d')^{1-\frac1q}\eps \leq \trm{\frac{1}{\eps^{\frac{1}{1-\frac1q}}}}^{1-\frac1q} \cdot \eps = 1.
  \]
  Note that we can embed any $d'$-dimensional probability distribution $p' \in \Delta_{d'}$ into the first $d'$ coordinates of $p \in \Delta_d$. Moreover, any $\eps$-$\ell_q$-estimate $\tilde{p}$ of $p$ naturally leads to an approximation $\tilde{p}'$ to $p'$ by only considering the first $d'$ entries of $\tilde{p}$. Using H\"older's inequality, we find that
  \[
    \nrm{\tilde{p}' - p'}_1 \leq (d')^{1-\frac1q}\nrm{\tilde{p}' - p'}_q \leq (d')^{1-\frac1q}\nrm{\tilde{p} - p}_q \leq (d')^{1-\frac1q}\eps \leq 1,
  \]
  and hence the number of queries in order to find an $\eps$-$\ell_q$-estimate of $p$ scales at least as
  \[
  \Omega\trm{\frac{d'}{\trm{(d')^{1-\frac1q}\eps}^s}} = \Omega\trm{\frac{(d')^{1-s+\frac{s}{q}}}{\eps^s}} = \Omega\trm{\frac{1}{\eps^{\frac{1}{1-\frac1q}}}}.
  \]
  Since this expression is indeed smaller than $d^{1/q}/\eps$ precisely when $\eps > 1/d^{1-1/q}$, we find that the lower bound becomes
  \[\Omega\trm{\min\left\{\frac{1}{\eps^{\frac{1}{1-\frac1q}}}, \frac{d^{1-s+\frac{s}{q}}}{\eps^s}\right\}},\]
  as claimed.
\end{proof}

We know from the first norm-conversion lemma, i.e., \autoref{lem:linftol2}, that obtaining an estimate of the amplitudes of a quantum state also gives you an estimate of the probability distribution arising from their absolute values squared. Therefore, our reduction from the previous lemma can be extended to give lower bounds on the problem of estimating a quantum state as well.

\begin{lemma}
\label{lem:ell1-reduction}
  Let $0 < \eps \leq 1$, $q,s \in [1,\infty]$, and suppose that in order to produce an $\eps$-$\ell_1$-estimate of any probability distribution $p \in \Delta_d$, defined as $p_j = |\alpha_j|^2$ with $\alpha \in \C^d$, with high probability, one needs to perform at least $\Omega(d/\eps^s)$ queries. Then
  \[
 \Omega\trm{\min\left\{\frac{1}{\eps^{\frac{1}{\frac12-\frac1q}}},\frac{d^{1-\frac{s}{2}+\frac{s}{q}}}{\eps^s} \right\}}
  \]
  queries are necessary to find a $\eps$-$\ell_q$-estimate of $\alpha$ with high probability.
\end{lemma}

\begin{proof}
  By \autoref{lem:linftol2}, we know that an $\eps$-$\ell_q$-norm estimate of $\alpha$ immediately gives an $4\eps$-$\ell_r$-estimate of $p$, with $r = 1/(1/q+1/2)$. Therefore, the result simply follows by substituting $r$ for $q$ into the bounds displayed in \autoref{lem:eps_ellq_prob_est}.
\end{proof}

\subsection{Lower bound on conditional samples}\label{sec:sampling-lb}
To lower bound the use of conditional samples we will use a prove based on communication complexity. In particular, we give an ensemble of states corresponding to conditional samples, such that an $\ell_1$-norm estimate of any of the states would give $\Omega(d)$ bits of information about which state was given, but a copy of the state can only communicate $\tilde{O}(\eps^2)$ bits of information.
\begin{lemma}\label{lem:probset}
Let $\eps > 0, d \ge 12$ with  $d\in \mathbb{N}$. There exists a set of $2^{d/2}$ probability distributions $\{p^{(b)}\}_{b\in\{0,1\}^{d/2}}\in \Delta^d$ indexed by length $d/2$ bit strings, such that for all $b$ and $b'$ where $\nrm{p^{(b)}-p^{(b')}}_1< 2\eps$ we have that $b$ and $b'$ differ on at most $d/6$ bits. Furthermore, let $\ket{\psi_j} = \sum_i \sqrt{p_i^{(j)}}\ket{i}$ and $\ket{\phi_j} = \frac{\ket{0}\ket{\psi_j}+\ket{1}\ket{0}}{\sqrt{2}}$; then, denoting the entropy by $S$, we have:
\[
S\trm{\frac{1}{2^{d/2}}\sum_j \ket{\phi_j}\bra{\phi_j}} \leq 27\eps^2 (1+ \log(d/\eps^2)).
\]
\end{lemma}
\begin{proof}
  We index the family of probability distributions with bit strings $b\in\{0,1\}^{d/2}$, writing $p^{(b)}$ for the distribution corresponding to string $b$.
  Each distribution is over $[d/2]\times \{0,1\}$ and is defined as
  \[
  p^{(b)}_{i,c} = \frac{1+(-1)^{b_i\oplus c}6\eps}{d}.
  \]
  In other words, it consists of $d/2$ pairs of entries that correspond to the bits of $b$, where the bit determines which of the entries in the pair is increased by $6\eps/d$ and which is decreased. 
    
  For two bit strings $b$ and $b'$ with Hamming distance $|b\oplus b'|$, the corresponding distributions will be $|b\oplus b'| \frac{12\eps}{d}$ apart in $\ell_1$-norm. Hence, if two distributions are less than $2 \eps$ apart, then for their bit strings we get $|b\oplus b'| < d/6$, i.e., less than a $1/3$ fraction of the positions differ. 
    
  It remains to upper bound the entropy of a uniform mixture of conditional samples.
  For ease of notation we will write 
  \[
    c_{\pm} = \sqrt{1 \pm 6\eps}
  \]
  and note that $(c_++c_-)^2 = 2+ 2c_+c_-$.
  In our notation, we have:
  \[
        \ket{\phi_b}\bra{\phi_b} = \frac{1}{2} \begin{pmatrix}
         \ket{\psi_b}\bra{\psi_b} & \ket{\psi_b}\\
         \bra{\psi_b} & 1
        \end{pmatrix},
  \]
  where for simplicity and without impacting subsequent calculations we have dropped the the $d-1$ all-zero columns on the right and the $d-1$ all-zero rows at the bottom.
  Let 
  \[
      \sigma = \frac{1}{2^{d/2}}\sum_b \ket{\phi_b}\bra{\phi_b} = \frac{1}{2^{d/2+1}} \begin{pmatrix}
       \sum_b \ket{\psi_b}\bra{\psi_b} & \sum_b \ket{\psi_b}\\
       \sum_b \bra{\psi_b} & 1
      \end{pmatrix}
  \]
  Considering a single entry of $\sum_b \ket{\psi_b}$, exactly half of the terms will be $\frac{c_+}{\sqrt{d}}$ and half will be $\frac{c_-}{\sqrt{d}}$. So $\frac{1}{2^{d/2+1}} \sum_b \ket{\psi_b} = \frac{c_++c_-}{4\sqrt{d}} \vecone$ and similar for the row vectors. 
    
  We now analyze the term $\ket{\psi_b}\bra{\psi_b}$. Consider the $2\times 2$ block of the matrix corresponding to $b_i$ for the rows and $b_k$ for the columns. Depending on the values of those two bits, this block can take four different forms:
\begin{multicols}{2}
\begin{enumerate}
        \item If $b_i = b_k = 0$ then the block is
        \[\frac{1}{d}
        \begin{pmatrix}
         c_+^2 & c_+c_- \\ c_+c_- & c_-^2
        \end{pmatrix}
        \]
        \item If $b_i = b_k = 1$ then the block is
        \[\frac{1}{d}
        \begin{pmatrix}
        c_-^2 & c_+c_-\\ c_-c_+ & c_+^2
        \end{pmatrix}
        \]
        \item If $b_i =0, b_k = 1$ then the block is
        \[\frac{1}{d}
        \begin{pmatrix}
        c_-c_+ & c_+^2 \\ c_-^2 & c_+c_- 
        \end{pmatrix}
        \]
        \item If $b_i =1, b_k = 0$ then the block is
        \[\frac{1}{d}
        \begin{pmatrix}
        c_-c_+ & c_-^2 \\ c_+^2 & c_+c_- 
        \end{pmatrix}
        \]
  \end{enumerate}
  \end{multicols}
  If $i=k$, i.e., on the diagonal, only (1) and (2) can happen, and by averaging over all possible $b$ (and putting back in the extra factor $1/2$ that appears in the denominator $1/(2^{d/2+1})$ in $\sigma$), we get:
       \[
         \frac{1}{2d}  \begin{pmatrix}
         1& c_+c_- \\  c_+c_- & 1
        \end{pmatrix}.
       \]
  Off-diagonally we average over all 4 possibilities, and get:
  \[
    \frac{1+ c_+c_-}{4d} \begin{pmatrix}
      1 & 1 \\ 1 & 1
    \end{pmatrix}.
  \]
  Denoting by $J$ the all-ones matrix and by $D := \bigoplus_{i\in [d/2]} \begin{pmatrix}
      1 & -1 \\ -1 & 1
      \end{pmatrix}$, the top left block of $\sigma$ can be written as
   \[
           \frac{1+c_+c_-}{4d} J +  
        \frac{ 1-c_+c_-}{4d} D.
      \]
  So
  \[
        \sigma = \begin{pmatrix}
         \frac{1+c_+c_-}{4d} J & \frac{c_++c_-}{4\sqrt{d}} \vecone \\ \frac{c_++c_-}{4\sqrt{d}} \vecone^{\top} & 1/2
        \end{pmatrix}+ \begin{pmatrix}
         \frac{ 1-c_+c_-}{4d} D & 0 \\ 0 & 0
        \end{pmatrix}.
  \]
  The first term in the above equation is a rank-1 matrix, as it is equal to the outer product of the column vector $(\frac{c_++c_-}{2\sqrt{2d}}\vecone, 1/\sqrt{2})$ with itself (recall that $(c_++c_-)^2/2 = 1 + c_+c_-$). The corresponding eigenvalue is just the norm of this vector, and it is equal to $\frac{3+c_+c_-}{4}$.
  The second term has $d/2$ nonzero eigenvalues, all equal to $\frac{1-c_+c_-}{2d}$. As  $(\frac{c_++c_-}{2\sqrt{2d}}\vecone, 1/\sqrt{2})$ is in the kernel of the second term, the $d/2+1$ eigenvalues listed above are in fact the eigenvalues of $\sigma$.  So:
  \begin{align*}
            S(\sigma) &= - \frac{3+c_+c_-}{4} \log\trm{\frac{3+c_+c_-}{4}} - \frac{d}{2} \frac{1-c_+c_-}{2d} \log\trm{\frac{1-c_+c_-}{2d}}\\
            &\leq \log\trm{\frac{4}{3+c_+c_-}} + \frac{36\eps^2}{4} \log\trm{\frac{2d}{1-c_+c_-}}\\
            &\leq \log\trm{\frac{1}{1-9\eps^2}} + 9\eps^2 \trm{ 1+\log(d/18\eps^2) }\\
            &\leq 18\eps^2 + 9\eps^2 \trm{ 1+\log(d/\eps^2)}\\
            &\leq 27\eps^2 (1+ \log(d/\eps^2)).
  \end{align*}
  In the chain of inequalities above, we used the fact that 
  \[
  1 - 36\eps^2 \leq \sqrt{1-36\eps^2} =c_+c_- \leq 1-18\eps^2
  \]
  and hence
  \[
  18\eps^2 \leq 1-c_+c_-  \leq 36\eps^2,
  \]
  together with the fact that the logarithm is monotonically increasing.
\end{proof}

With this entropy bound we are now ready to prove our sample complexity lower bound.
\begin{proposition}
\label{prop:lb-cond-samples}
 Let $p\in \Delta^d$ be a probability distribution, and let $\ket{\psi} = \sum_j \sqrt{p_j}\ket{j}$. Then 
 $\tilde{\Omega}(d/\eps^2)$ copies of 
 \[
 \frac{\ket{0}\ket{\psi}+\ket{1}\ket{0}}{\sqrt{2}}
 \] 
 are required to learn $p$ up to $\ell_1$-norm error $\eps$.
\end{proposition}
\begin{proof}
 We consider a communication scenario where Alice picks a $b\in\{0,1\}^{d/2}$ and encodes this in $k$ copies of $\ket{\phi_b}$ from \autoref{lem:probset}. She sends these copies to Bob. If Bob can estimate $p_b$ up to $\eps$-$\ell_1$-norm using $k$ copies then, by rounding to the closest distribution $p_{\tilde{b}}$, he can learn a $\tilde{b}$ that agrees with $b$ on at least a $2/3$ fraction of the bits, and hence he has learned $\Omega(d)$ bits of information about Alice's string. By Holevo's Theorem we know that the maximum amount of information that can be communicated by an ensemble of pure states is upper bounded by its entropy. As the entropy of $k$ copies of a state is equal to $k$ times the entropy of a single state, we have
 \[
  \Omega(d) \leq k S\trm{\frac{1}{2^{d/2}}\sum_j \ket{\phi_j}\bra{\phi_j}} \leq k 36\eps^2 (1+ \log(d/\eps^2))
 \]
 and hence $k = \Omega(\frac{d}{\eps^2 \log(d/\eps^2)})$ copies are needed for an $\ell_1$-norm estimate.
\end{proof}

We can now apply our norm conversion lemmas to obtain the following theorem. 
\begin{theorem}\label{thm:lowerbound-sample}
 Let $\ket{\psi} = \sum_j \alpha_j\ket{j}$. Then 
 \[
 \tilde{\Omega}\trm{\min\left\{\trm{\frac{3}{\eps}}^{\frac{1}{\frac12-\frac1q}}, \frac{d^{\frac2q}}{\eps}\right\}}
 \]
 copies of 
 \[
 \frac{\ket{0}\ket{\psi}+\ket{1}\ket{0}}{\sqrt{2}}
 \] 
 are required to learn $\alpha$ up to $\ell_q$-norm error $\eps$.
\end{theorem}
\begin{proof}
 We combine the $\ell_1$-reduction \autoref{lem:ell1-reduction} with $s=2$, with \autoref{prop:lb-cond-samples}.
\end{proof}

\subsection{Lower bound with inverse state preparation}\label{subsec:lb-with-inverses}

We start with showing that if we want to obtain an $\ell_1$-estimate of a probability distribution $p \in \Delta_d$, we need at least $\Omega(d/\eps)$ queries to the operation that prepares it $U : \ket{0} \mapsto \sum_{j=1}^d \sqrt{p_j}\ket{j}$. We do this by reducing the problem to the problem of recovering a constant fraction of the bits in a bit string, which is known to have a lower bound on of $\Omega(d/\eps)$.

\begin{lemma}
\label{lem:lb-ell1-inverse}
  Let $0 < \eps < 1/16$, $d \in \N$, $p\in \Delta^d$ a probability distribution, and let $U$ be a unitary that prepares $\sum_j \sqrt{p_j}\ket{j}$. Then $\Omega\trm{\frac{d}{\eps}}$ applications of $U$ and its inverse are necessary to find a $\eps$-$\ell_1$-estimate of $p$ with high probability.
\end{lemma}

\begin{proof}
Let $x \in \{0,1\}^d$ be a bit string, and suppose that we have controlled access to $x$ by means of a fractional phase oracle, i.e., we can access a controlled oracle that acts on $\C^d$ as
\[O_x : \ket{j} \mapsto e^{4\pi i\eps x_j}\ket{j}.\]
Recovering more than three quarters of the bits of $x$ with high probability is known to require $\Omega(d/\eps)$ queries to $O_x$.\footnote{Proving this is done in two steps -- first one proves that this takes at least $\Omega(d)$ calls to a regular phase oracle to $x$, which can be easily shown using an information theoretic argument stemming from \cite{farhi1999bounds}. Next, this can be combined with Appendix B from \cite{lee2011QQueryCompStateConv} and the general adversary bound for relations from \cite{belovs2015GeneralAdv}, to get to the desired bound of $\Omega(d/\eps)$. We can also reduce the problem to recovering the bit string exactly, and then directly apply the phase adversary bound from \cite[Ch.~6]{apeldoorn2020QConvexOptThesis}.}

Now, we construct a specific probability distribution $p$, dependent on $x$, whose corresponding quantum state can be constructed using only one call to $O_x$, and that allows for recovering at least 3/4 of the bit string if it is estimated up to $\eps$ in $\ell_1$-norm. To that end, suppose that we start in the state
\[
  \frac{1}{\sqrt{2d}} \sum_{j=1}^d \ket{j} \left(e^{-\frac{\pi i}{4}}\ket{0} + e^{\frac{\pi i}{4}}\ket{1}\right).
\]
Now, we can apply $O_x^{\dagger}$ to the first register if the last qubit is in state $\ket{0}$, and $O_x$ if the last qubit is in state $\ket{1}$. This results in the state
\[
  \frac{1}{\sqrt{2d}} \sum_{j=1}^d \ket{j} \left(e^{-\frac{\pi i}{4} - 4\pi i\eps x_j}\ket{0} + e^{\frac{\pi i}{4} + 4\pi i\eps x_j}\ket{1}\right).
\]
Next, after applying a Hadamard gate to the final qubit, we obtain the state
\[\frac{1}{\sqrt{d}} \sum_{j=1}^d \ket{j} \left(\cos\left(\frac{\pi}{4} + 4\pi \eps x_j\right)\ket{0} - i\sin\left(\frac{\pi}{4} + 4\pi \eps x_j\right)\ket{1}\right),\]
which after applying an $S$-gate to the final qubit is turned into the state
\[
  \ket{\psi} = \sum_{j \in [d]} \sum_{b \in \{0,1\}} \sqrt{p_{j,b}} \ket{j} \ket{b}, \qquad \text{where} \qquad p_{j,b} = \begin{cases}
    \frac{\cos^2\trm{\frac{\pi}{4} + 4\pi \eps x_j}}{d}, & \text{if } b = 0, \\
    \frac{\sin^2\trm{\frac{\pi}{4} + 4\pi \eps x_j}}{d}, & \text{if } b = 1.
  \end{cases}
\]
Since $\eps < 1/16$, we have
\[
  \left|\cos^2\trm{\frac{\pi}{4}} - \cos^2\trm{\frac{\pi}{4} + 4\pi\eps}\right| = \left|\sin^2\trm{\frac{\pi}{4}} - \sin^2\trm{\frac{\pi}{4} + 4\pi\eps}\right| = \frac12\sin\trm{8\pi\eps} > 8\eps.
\]
Now suppose that we can find some estimate $\tilde{p}$ such that $\nrm{\tilde{p} - p}_1 \leq \eps$. Then, define the bit string $\tilde{x} \in \{0,1\}^d$ as:
\[
  \tilde{x}_j = \begin{cases}
    1, & \text{if } \tilde{p}_{j,0} < \frac1d\trm{\frac12 - 4\eps}, \\
    0, & \text{otherwise}.
  \end{cases}
\]
It follows that $4\eps|\tilde{x}_j - x_j|/d \leq |\tilde{p}_{j,0} - p_{j,0}|$, and so the number of bits of $\tilde{x}$ that differ from those in $x$ is at most $d/4$. Hence, finding an $\eps$-$\ell_1$-norm estimate of a $2d$-dimensional probability distribution must take at least $\Omega(d/\eps)$ calls to a state-preparation oracle as well.
\end{proof}

It now remains to apply our $\ell_1$-reduction to complete the lower bound for general $\ell_q$-norms.

\begin{theorem}\label{thm:lowerboundpure}
 Let $\ket{\psi} = \sum_{j \in [d]} \alpha_j \ket{j}$ be a quantum state with and let $U$ be a unitary that prepares $\ket{\psi}$. Then 
  \[
 \tilde \Omega\trm{\min\left\{\frac{1}{\eps^{\frac{1}{\frac12 - \frac1q}}},\frac{d^{\frac1q + \frac12}}{\eps} \right\}}
  \]
  applications of $U$ and its inverse are necessary to find an $\eps$-$\ell_q$-estimate of $|\alpha|$ for $q\in [2,\infty]$.
\end{theorem}
\begin{proof}
  This follows from combining the $\ell_1$-reduction \autoref{lem:ell1-reduction} with $s = 1$, and \autoref{lem:eps_ellq_prob_est}.
\end{proof}

\subsection{Lower bounds on mixed-state state tomography}\label{subsec:lb-mixed-state}

In this section, we prove optimality of our algorithm to recover a density matrix up to Frobenius norm error $\varepsilon$, given (inverse) access to a unitary that prepares its purification. At a high level, the proof consists of three steps. First, we embed a bit string $b$ of length $rd$ inside a density matrix, whose purification can be prepared using only logarithmically many calls to an $\varepsilon$-fractional phase oracle to $b$. Then, we show that recovering the density matrix up to given precision $\varepsilon$ narrows down the number of possible choices for $b$ to a fraction $2^{-crd}$, for some small $c > 0$. Finally, we argue that consequently, we need to make at least $\widetilde{\Omega}(rd/\varepsilon)$ calls to the state-preparation unitary.

The embedding makes use of mutually unbiased bases, that we define below. Since we our construction requires some auxiliary properties of mutually unbiased bases, we modify the construction presented in \cite{Bandyopadhyay2001new}.

\begin{theorem}[Mutually unbiased bases]\label{thm:mub}
	Let $d \in \N$ be an odd prime, and let $j \in [d]$. For all $k \in [d]$, we define $\ket{\varphi_k^{(j)}} \in \C^d$ as
	\[\ket{\varphi_k^{(j)}} = \frac{1}{\sqrt{d}} \sum_{\ell=0}^{d-1} \omega_d^{-k\ell + j\ell^2 + k^2}\ket{\ell}, \qquad \text{with} \qquad \omega_d = e^{\frac{2\pi i}{d}}.\]
	For all $j \in [d]$, we define $U^{(j)} \in \C^{d \times d}$, and for all $j,j' \in [d]$ we define $\alpha^{(j,j')} \in \C^{d \times d}$ as
	\[U^{(j)} = \frac{1}{\sqrt{d}} \sum_{k=0}^{d-1} \ket{\varphi_k^{(j)}}\bra{k}, \qquad \text{and} \qquad \alpha^{(j,j')}_{k,k'} = \braket{\varphi_k^{(j)}}{\varphi_{k'}^{(j')}} = \sum_{\ell=0}^{d-1} \overline{U}^{(j)}_{\ell,k}U^{(j')}_{\ell,k'}.\]
	Next, for all $j,j' \in [d]$, we define $S^{(j,j')} \in \C$, and we let $A \in \R^{rd \times rd}$ be defined as
	\[S^{(j,j')} = \sum_{k,k'=0}^{d-1} \alpha^{(j,j')}_{k,k'}, \qquad \text{and} \qquad A^{(j,j')}_{k,k'} = \mathrm{Re}\left[\alpha^{(j,j')}_{k,k'}\overline{S}^{(j,j')}\right].\]
	These newly-defined objects satisfy the following properties.
	\begin{enumerate}
		\setlength\itemsep{-.2em}
		\item For all $j \in [d]$, $U^{(j)}$ is unitary.
		\item The bases $\{\ket{\varphi^{(j)}_k} : k \in [d]\}$, where $j \in [d]$, are mutually orthogonal.
		\item If $j = j'$, then $\alpha^{(j,j')}_{k,k'} = 1_{k = k'}$. If $j \neq j'$, then $\alpha^{(j,j')}_{k,k'} = \omega_d^{(k')^2 - k^2 - (4(j' - j))^{-1} (k - k')^2}\left(\frac{j'-j}{d}\right)c_d/\sqrt{d}$, where $\left(\frac{j'-j}{d}\right)$ is the Legendre symbol of $j' - d$ in $\mathbb{F}_d$, and $c_d = 1$ if $d \equiv 1 \mod 4$ and $c_d = i$ if $d \equiv 3 \mod 4$.
		\item If $j = j'$, then $S^{(j,j')} = d$. If $j \neq j'$, then $S^{(j,j')} = c_d\sqrt{d}\left(\frac{j'-j}{d}\right)$.
		\item $\nrm{A} \leq 2d$, and $\nrm{A}_2^2 \leq 4d^3r$.
	\end{enumerate}
\end{theorem}

\begin{proof}
	For claim 1, we need to check that $\{\ket{\varphi_k^{(j)}} : k \in [d]\}$ defines an orthonormal basis of $\C^d$. It is immediate that $\ket{\varphi_k^{(j)}}$ has unit norm, so it remains to check that
	\[\braket{\varphi^{(j)}_k}{\varphi^{(j)}_{k'}} = \frac{1}{d} \sum_{\ell=0}^{d-1} \omega_d^{k\ell - j\ell^2 - k^2 - k'\ell + j\ell^2 + (k')^2} = \frac{\omega_d^{(k')^2 - k^2}}{d} \sum_{\ell=0}^{d-1} \omega_d^{\ell(k-k')} = 0,\]
	when $k \neq k'$.
	
	For claim 2, let $j,j' \in [d]$, with $j \neq j'$. Then, for any $k,k' \in [d]$, we have
	\[\braket{\varphi_k^{(j)}}{\varphi_{k'}^{(j')}} = \frac{1}{d} \sum_{\ell=0}^{d-1} \omega_d^{k\ell - j\ell^2 - k^2 - k'\ell + j'\ell^2 + (k')^2} = \frac{\omega_d^{(k')^2 - k^2}}{d} \sum_{\ell=0}^{d-1} \omega_d^{\ell(k-k') + \ell^2(j'-j)}.\]
	The summation on the right-hand side is known as a generalized quadratic Gauss sum. For any $a,b \in \mathbb{F}_d$ with $a \neq 0$, we can calculate
	\[\sum_{\ell=0}^{d-1} \omega_d^{a\ell^2 + b\ell} = \sum_{\ell=0}^{d-1} \omega_d^{a(\ell + (2a)^{-1}b)^2 - (4a)^{-1}b^2} = \omega_d^{-(4a)^{-1}b^2} \sum_{\ell=0}^{d-1} \omega_d^{a\ell^2} = \omega_d^{-(4a)^{-1}b^2} \left(\frac{a}{d}\right)c_d\sqrt{d},\]
	where $c_d = 1$ if $d \equiv 1 \mod 4$, and $c_d = i$ if $d \equiv 3 \mod 4$, and $(\frac{a}{d})$ denotes the Legendre symbol in $\mathbb{F}_d$. Thus, by plugging in the result of this calculation, we obtain
	\[\braket{\varphi_k^{(j)}}{\varphi_{k'}^{(j')}} = \frac{\omega_d^{(k')^2 - k^2 - (4(j'-j))^{-1}(k-k')^2}}{\sqrt{d}}\left(\frac{j'-j}{d}\right)c_d,\]
	which indeed implies that the absolute value of this inner product is $1/\sqrt{d}$. Thus, the bases are unbiased.
	
	The first part of claim 3, i.e., the case where $j = j'$, follows directly from the fact that $\{\ket{\varphi_k^{(j)}} : k \in [d]\}$ is a basis, for all $j \in [d]$. The second part follows from the previous calculation.
	
	The first part of claim 4, i.e., the case where $j = j'$, is also easily verified. Thus it remains to check for the second part that
	\[S^{(j,j')} = \sum_{k,k'=0}^{d-1} \alpha_{k,k'}^{(j,j')} = \frac{c_d}{\sqrt{d}}\left(\frac{j'-j}{d}\right) \sum_{k,k'=0}^{d-1} \omega_d^{(k')^2 - k^2 - (4(j'-j))^{-1} (k - k')^2}.\]
	We can factor the exponent on the right-hand side according to
	\[(k')^2 - k^2 - x(k-k')^2 = (k'-k)(k'+k) - x(k-k')^2 = (k'-k)(k'+k-x(k'-k)),\]
	and hence plugging in $x = (4(j'-j))^{-1}$, and relabeling $k' = m + k$, we obtain
	\[S^{(j,j')} = \frac{c_d}{\sqrt{d}}\left(\frac{j'-j}{d}\right) \sum_{k,m=0}^{d-1} \omega_d^{m(m + 2k - xm)} = \frac{c_d}{\sqrt{d}}\left(\frac{j'-j}{d}\right) \sum_{m=0}^{d-1} \omega_d^{(1-x)m^2} \sum_{k=0}^{d-1} \omega_d^{2mk}.\]
	The rightmost term vanishes for all choices of $m$, except for $m = 0$, in which case it becomes $d$. Thus, we can simplify the expression to
	\[S^{(j,j')} = c_d\sqrt{d}\left(\frac{j'-j}{d}\right).\]
	
	Finally, for claim 5, observe that the second part follows from the first part, since the operator norm is the maximal absolute eigenvalue and the Frobenius norm is the $\ell_2$-norm of the vector of eigenvalues. Thus, it remains to bound the operator norm of $A$.
	
	We first observe that for any $j,j' \in [d]$ with $j \neq j$, and $k,k' \in [d]$, combining results from claims 3 and 4 yields
	\[\alpha_{k,k'}^{(j,j')}\overline{S}^{(j,j')} = \frac{\omega_d^{(k')^2-k^2-(4(j'-j))^{-1}(k-k')^2}}{\sqrt{d}} \left(\frac{j'-j}{d}\right)c_d \cdot \overline{c}_d\sqrt{d}\left(\frac{j'-j}{d}\right) = \omega_d^{(k')^2-k^2-(4(j'-j))^{-1}(k-k')^2}.\]
	Next, we define the matrix $M \in \C^{dr \times dr}$ by
	\[M_{k,k'}^{(j,j')} = \begin{cases}
		\alpha_{k,k'}^{(j,j')}\overline{S}^{(j,j')}, &\text{if } j \neq j', \\
		0, &\text{otherwise},
	\end{cases}\]
	and we observe that $A = \mathrm{Re}[M] + dI$. In particular, it follows that $\nrm{A} \leq (\nrm{M} + \nrm{\overline{M}})/2 + d = \nrm{M} + d$. Thus, it remains to prove $\nrm{M} \leq d$.
	
	Note that without loss of generality, we can assume that $r = d$. Indeed, if $r < d$, we are merely considering a submatrix of the matrix $M$ we obtain when we choose $r = d$, and hence the norm of $M$ is largest whenever $r = d$.
	
	Now, we characterize the spectrum of $M$ and its corresponding eigenvectors. To that end, for all $\ell,m \in [d]$, define the vector $v^{(\ell,m)} \in \C^{rd}$ as
	\[\left(v^{(\ell,m)}\right)^{(j)}_k = \omega_d^{\ell j - k(k-m)}.\]
	First, we prove that all these vectors are orthogonal to one another. For all $\ell,\ell',m,m' \in [d]$, we have
	\[\left(v^{(\ell,m)}\right)^{\dagger}v^{(\ell,m)} = \sum_{j,k=0}^{d-1} \omega_d^{-\ell j + k(k-m) + \ell'j - k(k-m')} = \sum_{j=0}^{d-1} \omega_d^{j(\ell'-\ell)} \sum_{k=0}^{d-1} \omega_d^{k(m'-m)},\]
	from which we easily verify that the right-hand side indeed vanishes when either $\ell \neq \ell'$ or $m \neq m'$.
	
	Next, we prove that all $v^{(\ell,m)}$ are indeed eigenvectors of $M$. To that end, let $\ell,m \in [d]$, and observe that for all $j,k \in [d]$,
	\[\left(Mv^{(\ell,m)}\right)_{j,k} = \sum_{j',k'=0}^{d-1} M^{(j,j')}_{k,k'}\left(v^{(\ell,m)}\right)^{(j')}_{k'} = \sum_{j',k'=0}^{d-1} \omega_d^{(k')^2 - k^2 - (4(j'-j))^{-1} (k'-k)^2 + \ell j' - k'(k'-m)}.\]
	We focus on the exponent on the right-hand side, and simplify the expression to
	\[-k^2 - (4(j'-j))^{-1} (k'-k)^2 + \ell j' + k'm.\]
	Next, observe that we can substitute $k'$ by $k' + k$ and $j'$ by $j' + j$, which simplifies the exponent to
	\[-k^2 - (4j')^{-1} (k')^2 + \ell j' + \ell j + k'm + km = -k(k-m) + \ell j + \ell j' - (4j')^{-1} (k' - 2j'm)^2 - j'm^2.\]
	Thus, we obtain
	\[\left(Mv^{(\ell,m)}\right)^{(j)}_k = \omega_d^{\ell j - k(k-m)} \sum_{j',k'=0}^{d-1} \omega_d^{(\ell-m^2)j' - (4j')^{-1}(k' - 2j'm)^2}.\]
	The phase factor we took outside the summation is equal to $(v^{(\ell,m)})^{(j)}_k$, and we can again substitute $k'$ by $k' + 2j'm$. Then, we recognize that we have a quadratic Gauss sum on the right-hand side, which we can evaluate to obtain
	\[\left(Mv^{(\ell,m)}\right)^{(j)}_k = \left(v^{(\ell,m)}\right)^{(j)}_k \sum_{j',k'=0}^{d-1} \omega_d^{(\ell-m^2)j' - (4j')^{-1}(k')^2} = \left(v^{(\ell,m)}\right)^{(j)}_k \sum_{j'=0}^{d-1} \omega_d^{(\ell-m^2)j'}\left(\frac{j'}{d}\right)c_d\sqrt{d},\]
	where we used standard computation rules to drop the $4$ and the inverse from the Legendre symbol. Now, recall that the Legendre symbol is only $1$ whenever $j'$ is a quadratic residue in $\mathbb{F}_d$. Thus, in general for $x \in [d]$,
	\[\sum_{j'=0}^{d-1} \omega_d^{xj'}\left(\frac{j'}{d}\right) = -\sum_{j'=0}^{d-1} \omega_d^{xj'} + \sum_{j'=0}^{d-1} \omega^{x(j')^2} = \left(\frac{x}{d}\right)c_d\sqrt{d}.\]
	Putting everything together yields
	\[\left(Mv^{(\ell,m)}\right)^{(j)}_k = dc_d^2\left(\frac{\ell-m^2}{d}\right)\left(v^{(\ell,m)}\right)^{(j)}_k.\]
	Thus, we conclude that $v^{(\ell,m)}$ is indeed an eigenvector of $M$, with eigenvalue $dc_d^2\left(\frac{\ell-m^2}{d}\right)$. Since $M$ is symmetric, the operator norm of $M$ is equal to its largest eigenvalue, and hence $\nrm{M} = d$.
\end{proof}

Next, we define the embedding of the a string $b \in \{0,1\}^{rd}$ into a density matrix $\rho_b$.

\begin{definition}
	Let $0 < \varepsilon < 1$, $d \in \N$ an odd prime, $r \in [d]$, and $b \in \{0,1\}^{rd}$. We write $b = (b^{(0)}, \dots, b^{(r-1)})$, where every block $b^{(j)}$ is a length $d$ bit string. For all $j \in [r]$, we define
	\[\ket{\psi_b^{(j)}} = \frac{1}{\sqrt{d}} \sum_{k=0}^{d-1} \left(\sqrt{\frac12 + \frac12\varepsilon(-1)^{b^{(j)}_k}}\ket{0} + \sqrt{\frac12 - \frac12\varepsilon(-1)^{b^{(j)}_k}}\ket{1}\right)\left(U^{(j)}\right)^{\dagger}\ket{k}.\]
	Then, we define the purification $\ket{\psi_b}$, embedding the bit string $b$, as
	\[\ket{\psi_b} = \frac{1}{\sqrt{r}} \sum_{j=0}^{r-1} \ket{\psi_b^{(j)}}\ket{j},\]
	and we find the density matrix $\rho_b$ by tracing out the final register in the above definition, i.e.,
	\[\rho_b = \frac1r \sum_{j=0}^{r-1} \ket{\psi_b^{(j)}}\bra{\psi_b^{(j)}}.\]
\end{definition}

Intuitively, if one learns $\rho_b$ up to high precision, then one also obtains much information about the bit string $b$. Thus, if we are given access to a density matrix $\rho_b$ for some unknown bit string $b \in \{0,1\}^{rd}$, and we find a good approximation of $\rho_b$, we can learn a small set of bit strings, one of which much be $b$ itself. In order to quantify how small this set of bit strings becomes, we analyze the distance between two given density matrices $\rho_b$ and $\rho_{\overline{b}}$, which is the objective of the following lemma.

\begin{lemma}\label{lem:frodist}
	Let $0 < \varepsilon = o(1/(d\sqrt{r}))$, $d \in \N$ an odd prime, $r \in [d]$, and $b,\overline{b} \in \{0,1\}^{rd}$. Let $\delta = (\delta^{(0)}, \dots, \delta^{(r-1)}) \in \{-2,0,2\}^{rd}$ be defined as
	\[\delta^{(j)}_k = (-1)^{b^{(j)}_k} - (-1)^{\overline{b}^{(j)}_k},\]
	and let $A \in \R^{rd \times rd}$ be as in \autoref{thm:mub}. Then,
	\[\nrm{\rho_b - \rho_{\overline{b}}}_2^2 = \frac{\varepsilon^2}{2d^2r^2}\delta^TA\delta + o\left(\frac{\varepsilon^2}{r}\right).\]
\end{lemma}

\begin{proof}
	From the definition, we observe that
	\begin{align*}
		\ket{\psi_b^{(j)}} &= \frac{1}{\sqrt{d}} \sum_{\ell=0}^{d-1} \sum_{c \in \{0,1\}} \sqrt{\frac12 + \frac12\varepsilon(-1)^{c+b_{\ell}^{(j)}}}\ket{c} \left(U^{(j)}\right)^{\dagger}\ket{\ell} \\
		&= \frac{1}{\sqrt{d}} \sum_{k,\ell=0}^{d-1} \overline{U}^{(j)}_{k,\ell} \sum_{c \in \{0,1\}} \sqrt{\frac12 + \frac12\varepsilon(-1)^{c + b_{\ell}^{(j)}}}\ket{c}\ket{k},
	\end{align*}
	which implies that
	\begin{align*}
		\rho_b &= \frac1r \sum_{j=0}^{r-1} \ket{\psi_b^{(j)}}\bra{\psi_b^{(j)}} \\
		&= \frac{1}{rd} \sum_{j=0}^{r-1} \sum_{k,k',\ell,m=0}^{d-1} \overline{U}^{(j)}_{k,\ell}U^{(j)}_{k',m} \sum_{c,c' \in \{0,1\}} \sqrt{\frac12 + \frac12\varepsilon(-1)^{c+b_{\ell}^{(j)}}} \sqrt{\frac12 + \frac12\varepsilon(-1)^{c'+b_m^{(j)}}} \ket{c}\ket{k} \bra{c'}\bra{k'}.
	\end{align*}
	From this, it follows directly that
	\[\rho_b - \rho_{\overline{b}} = \frac{1}{rd} \sum_{j=0}^{r-1} \sum_{k,k',\ell,m=0}^{d-1} \overline{U}^{(j)}_{k,\ell}U^{(j)}_{k',m} \sum_{c,c' \in \{0,1\}} \left(C^{(c,c',j)}_{\ell,m} - \overline{C}^{(c,c',j)}_{\ell,m}\right) \ket{c}\ket{k} \bra{c'}\bra{k'},\]
	where we used the abbreviation
	\[C_{\ell,m}^{(c,c',j)} = \frac12 \sqrt{1+\varepsilon(-1)^{c+b_{\ell}^{(j)}}}\sqrt{1+\varepsilon(-1)^{c'+b_m^{(j)}}},\]
	and similarly for $\overline{C}$, where we replace every occurrence of $b$ by $\overline{b}$. Since the square of the Frobenius norm of a matrix is the sum of all entries squared, we obtain
	\begin{align*}
		\nrm{\rho_b - \rho_{\overline{b}}}_2^2 &= \frac{1}{r^2d^2} \sum_{k,k'=0}^{d-1} \sum_{j,j'=0}^{r-1} \sum_{\ell,\ell',m,m'=0}^{d-1} \overline{U}_{k,\ell}^{(j)}U^{(j)}_{k',m} U^{(j')}_{k,\ell'}\overline{U}^{(j')}_{k',m'} \\
		&\qquad \cdot \sum_{c,c' \in \{0,1\}} \left(C_{\ell,m}^{(c,c',j)} - \overline{C}_{\ell,m}^{(c,c',j)}\right)\left(C_{\ell',m'}^{(c,c',j')} - \overline{C}_{\ell',m'}^{(c,c',j')}\right).
	\end{align*}
	First, we focus on the last summation. To that end, observe that
	\[\sum_{c \in \{0,1\}} \sqrt{1 + \varepsilon(-1)^{c+b_{\ell}^{(j)}}}\sqrt{1 + \varepsilon(-1)^{c+b_{\ell'}^{(j')}}} = \begin{cases}
		2, & \text{if } b_{\ell}^{(j)} = b_{\ell'}^{(j')}, \\
		2\sqrt{1-\varepsilon^2}, & \text{if } b_{\ell}^{(j)} \neq b_{\ell'}^{(j')},
	\end{cases}\]
	where we can abbreviate the right-hand side to $2 - 2f_{\varepsilon}1_{b_{\ell}^{(j)} \neq b_{\ell'}^{(j')}}$, where $f_{\varepsilon} = 1 - \sqrt{1-\varepsilon^2}$. Thus, by simply expanding all terms, we obtain
	\begin{align*}
		&4\sum_{c,c' \in \{0,1\}} \left(C_{\ell,m}^{(c,c',j)} - \overline{C}_{\ell,m}^{(c,c',j)}\right)\left(C_{\ell',m'}^{(c,c',j')} - \overline{C}_{\ell',m'}^{(c,c',j')}\right) \\
		&= 4\sum_{c,c' \in \{0,1\}} C_{\ell,m}^{(c,c',j)}C_{\ell',m'}^{(c,c',j')} + 4\sum_{c,c' \in \{0,1\}} \overline{C}_{\ell,m}^{(c,c',j)}\overline{C}_{\ell',m'}^{(c,c',j')} \\
		&\qquad - 4\sum_{c,c' \in \{0,1\}} C_{\ell,m}^{(c,c',j)}\overline{C}_{\ell',m'}^{(c,c',j')} - 4\sum_{c,c' \in \{0,1\}} \overline{C}_{\ell,m}^{(c,c',j)}C_{\ell',m'}^{(c,c',j')} \\
		&= \left(2 - 2f_{\varepsilon}1_{b_{\ell}^{(j)} \neq b_{\ell'}^{(j')}}\right)\left(2 - 2f_{\varepsilon}1_{b_m^{(j)} \neq b_{m'}^{(j')}}\right) + \left(2 - 2f_{\varepsilon}1_{\overline{b}_{\ell}^{(j)} \neq \overline{b}_{\ell'}^{(j')}}\right)\left(2 - 2f_{\varepsilon}1_{\overline{b}_m^{(j)} \neq \overline{b}_{m'}^{(j')}}\right) \\
		&\qquad - \left(2 - 2f_{\varepsilon}1_{b_{\ell}^{(j)} \neq \overline{b}_{\ell'}^{(j')}}\right)\left(2 - 2f_{\varepsilon}1_{b_m^{(j)} \neq \overline{b}_{m'}^{(j')}}\right) - \left(2 - 2f_{\varepsilon}1_{\overline{b}_{\ell}^{(j)} \neq b_{\ell'}^{(j')}}\right)\left(2 - 2f_{\varepsilon}1_{\overline{b}_m^{(j)} \neq b_{m'}^{(j')}}\right) \\
		&= 2f_{\varepsilon}\left[\delta_{\ell}^{(j)}\delta_{\ell'}^{(j')} + \delta_m^{(j)}\delta_{m'}^{(j')}\right] + \mathcal{O}(f_{\varepsilon}^2) = \varepsilon^2\left[\delta_{\ell}^{(j)}\delta_{\ell'}^{(j')} + \delta_m^{(j)}\delta_{m'}^{(j')}\right] + \mathcal{O}(\varepsilon^4),
	\end{align*}
	where the rewriting into $\delta$'s is best checked by brute forcing all assignments of the bits involved. Putting everything back together, we obtain
	\[\nrm{\rho_b - \rho_{\overline{b}}}_2^2 = \frac{\varepsilon^2}{4r^2d^2} \sum_{k,k'=0}^{d-1} \sum_{j,j'=0}^{r-1} \sum_{\ell,\ell',m,m'}^{d-1} \overline{U}^{(j)}_{k,\ell}U^{(j)}_{k',m} U^{(j')}_{k,\ell'} \overline{U}^{(j')}_{k',m'} \left[\delta_{\ell}^{(j)}\delta_{\ell'}^{(j')} + \delta_m^{(j)}\delta_{m'}^{(j')}\right] + \mathcal{O}(d^2\varepsilon^4).\]
	Since we chose $\varepsilon = o(1/(d\sqrt{r}))$, we obtain that the final term becomes $o(\varepsilon^2/r)$. Thus, we can rewrite the summation to obtain
	\begin{align*}
		\nrm{\rho_b - \rho_{\overline{b}}}_2^2 &= \frac{\varepsilon^2}{4r^2d^2}  \sum_{j,j'=0}^{r-1} \sum_{\ell,\ell',m,m'=0}^{d-1} \sum_{k=0}^{d-1} \overline{U}^{(j)}_{k,\ell}U^{(j')}_{k,\ell'} \sum_{k'=0}^{d-1} U^{(j)}_{k',m} \overline{U}^{(j')}_{k',m'} \left[\delta_{\ell}^{(j)}\delta_{\ell'}^{(j')} + \delta_m^{(j)}\delta_{m'}^{(j')}\right] + o\left(\frac{\varepsilon^2}{r}\right) \\
		&= \frac{\varepsilon^2}{4r^2d^2} \sum_{j,j'=0}^{r-1} \sum_{\ell,\ell',m,m'=0}^{d-1} \alpha^{(j,j')}_{\ell,\ell'} \overline{\alpha}^{(j,j')}_{m,m'} \left[\delta_{\ell}^{(j)}\delta_{\ell'}^{(j')} + \delta_m^{(j)}\delta_{m'}^{(j')}\right] + o\left(\frac{\varepsilon^2}{r}\right) \\
		&= \frac{\varepsilon^2}{4r^2d^2} \sum_{j,j'=0}^{r-1} \left[\sum_{\ell,\ell'=0}^{d-1} \delta_{\ell}^{(j)} \alpha^{(j,j')}_{\ell,\ell'} \sum_{m,m'}^{d-1} \overline{\alpha}^{(j,j')}_{m,m'} \delta_{\ell'}^{(j')} + \overline{\sum_{m,m'}^{d-1} \delta_m^{(j)}\alpha^{(j,j')}_{m,m'} \sum_{\ell,\ell'}^{d-1} \overline{\alpha}^{(j,j')}_{\ell,\ell'} \delta_{m'}^{(j')}}\right] + o\left(\frac{\varepsilon^2}{r}\right) \\
		&= \frac{\varepsilon^2}{2r^2d^2} \sum_{j,j'=0}^{r-1} \sum_{\ell,\ell'=0}^{d-1} \delta_{\ell}^{(j)} \mathrm{Re}\left[\alpha^{(j,j')}_{\ell,\ell'} \sum_{m,m'}^{d-1} \overline{\alpha}^{(j,j')}_{m,m'}\right] \delta_{\ell'}^{(j')} + o\left(\frac{\varepsilon^2}{r}\right) \\
		&= \frac{\varepsilon^2}{2r^2d^2} \delta^T A \delta + o\left(\frac{\varepsilon^2}{r}\right).
	\end{align*}
	This completes the proof.
\end{proof}

In the previous lemma, we related the Frobenius norm distance between two density matrices $\rho_b$ and $\rho_{\overline{b}}$ to the inner product matrix $A$ of the mutually unbiased basis, as defined in \autoref{thm:mub}. Next, we can use this characterization to investigate the Frobenius distance we can expect between two density matrices $\rho_b$ and $\rho_{\overline{b}}$, when both bit strings $b$ and $\overline{b}$ are chosen independently and uniformly at random. This is the objective of the following lemma.

\begin{lemma}\label{lem:tailbound}
	Let $d \in \N$ an odd prime, and $r \in [d]$. Let $b,\overline{b} \in \{0,1\}^{rd}$ be be bit strings chosen independently and uniformly at random. Let $\delta$ and $A$ be as in the previous lemma. Then,
	\[\mathbb{E}\left[\delta^TA\delta\right] = 2d^2r,\]
	and there exist absolute constants $c,C > 0$ such that
	\[\P\left[\delta^TA\delta \leq \frac{d^2r}{2}\right] \leq Ce^{-crd}.\]
\end{lemma}

\begin{proof}
	Observe that $\delta \in \{-2,0,2\}^{rd}$, and all entries of $\delta$ are independent and distributed according to
	\[\P[\delta_j = -2] = \frac14, \qquad \P[\delta_j = 0] = \frac12, \qquad \text{and} \qquad \P[\delta_j = 2] = \frac14.\]
	We immediately observe that $\mathbb{E}[\delta_j] = 0$ and $\mathbb{E}[\delta_j^2] = 2$. Thus,
	\[\mathbb{E}[\delta^TA\delta] = \sum_{j,j'=0}^{r-1} \sum_{k,k'=0}^{d-1} \mathbb{E}\left[\delta_k^{(j)}\delta_{k'}^{(j')}\right] A^{(j,j')}_{k,k'} = 2\sum_{j=0}^{r-1} \sum_{k=0}^{d-1} A_{k,k}^{(j,j)} = 2\sum_{j=0}^{r-1} \sum_{k=0}^{d-1} \mathrm{Re}\left[\alpha^{(j,j)}_{k,k} \sum_{m,m'}^{d-1} \overline{\alpha}^{(j,j)}_{m,m'}\right] = 2d^2r.\]
	Furthermore, observe that all elements $\delta_j$ are subgaussian with some constant parameter, i.e., a parameter that is $\Theta(1)$. This allows us to invoke the Hanson-Wright inequality of subgaussian concentration~\cite{hanson1971bound,rudelson2013hanson}. According to such inequality, there exist positive absolute constants $c,C > 0$ such that for all $t > 0$,
	\[\P\left[\left|\delta^TA\delta - \mathbb{E}[\delta^TA\delta]\right| \geq t\right] \leq C\exp\left(-c\min\left\{\frac{t^2}{\nrm{A}_2^2}, \frac{t}{\nrm{A}}\right\}\right).\]
	Thus, by plugging in $t = 3d^2r/2$, $\nrm{A}_2^2 \leq 4d^3r$ and $\nrm{A} \leq 2d$, as proved in \autoref{thm:mub}, we obtain that
	\[\P\left[\delta^TA\delta \leq \frac{d^2r}{2}\right] \leq \P\left[\left|\delta^TA\delta - \mathbb{E}[\delta^TA\delta]\right| \geq \frac{3d^2r}{2}\right] \leq C\exp\left(-c\min\left\{\frac{9d^4r^2}{16d^3r}, \frac{3d^2r}{4d}\right\}\right) = Ce^{-c'rd},\]
	where $c' = 9c/16$.
\end{proof}

The above observation allows us to conclude that the distribution of $\nrm{\rho_b - \rho_{\overline{b}}}_2^2$ is tightly concentrated around its mean. The next lemma formalizes this statement, and uses the concentration to obtain a tail bound in the low Frobenius norm regime.

\begin{lemma}
	Let $d \in \N$ an odd prime, $r \in [d]$, and $0 < \varepsilon = o(1/(d\sqrt{r}))$. Then,
	\[\mathbb{E}\left[\nrm{\rho_b - \rho_{\overline{b}}}_2^2\right] = \frac{\varepsilon^2}{r}\left(1 + o(1)\right),\]
	and there exist absolute constants $c,C > 0$ such that
	\[\P\left[\nrm{\rho_b - \rho_{\overline{b}}}_2^2 \leq \frac{\varepsilon^2}{8r}\right] \leq Ce^{-crd}.\]
\end{lemma}

\begin{proof}
	The statement follows directly from the previous two lemmas. We know from \autoref{lem:frodist} that
	\[\mathbb{E}\left[\nrm{\rho_b - \rho_{\overline{b}}}_2^2\right] = \frac{\varepsilon^2}{2r^2d^2} \mathbb{E}\left[\delta^TA\delta\right] + o\left(\frac{\varepsilon^2}{r}\right) = \frac{\varepsilon^2}{r}\left(1 + o(1)\right),\]
	where the last equality follows from \autoref{lem:tailbound}. Furthermore, note that by choosing $\varepsilon$ small enough (i.e., choosing a sufficiently small constant in the small-$o$-notation), we can ensure that $o(\varepsilon^2/r)$ is smaller than $\varepsilon^2/(4r)$. Then, we obtain that there are indeed absolute constants $c,C > 0$ such that
	\[\P\left[\nrm{\rho_b - \rho_{\overline{b}}}_2^2 \leq \frac{\varepsilon^2}{8r}\right] \leq \P\left[\frac{\varepsilon^2}{2r^2d^2}\delta^TA\delta \leq \frac{\varepsilon^2}{4r}\right] = \P\left[\delta^TA\delta \leq \frac{d^2r^2}{2}\right] \leq Ce^{-crd},\]
	where the last inequality follows from \autoref{lem:tailbound}. 
\end{proof}

Now, we are able to finish the proof. The proof strategy followed from here onward is very similar to those presented in \cite{cornelissen2021quantum},~Section~5.

\begin{theorem}\label{thm:lb-frob-norm-tomo}
	Let $d \in \N$, and $r \in [d]$. Let $0 < \varepsilon = o(1/(d\sqrt{r}))$. Suppose that we have a $Q$-query quantum algorithm that given access to an (inverse) state-preparation unitary for a purification of an $rd \times rd$ density matrix $\rho$, outputs an approximation $\widetilde{\rho}$ such that $\nrm{\widetilde{\rho} - \rho}_2 \leq \varepsilon/(2\sqrt{8r})$, with probability at least $2/3$. Then, $Q = \Omega(dr/\varepsilon)$.
\end{theorem}

\begin{proof}
	First, without loss of generality we can assume that $d$ is an odd prime. Indeed, if it is not, we can find the next odd prime that is higher than $d$, which by Bertrand's postulate~\cite{bertrand1845memoire} does not increase $d$ by more than a factor of $2$.
	
	Next, let $G$ be a bipartite graph with $2 \cdot 2^{rd}$ nodes, labeled  by the bit strings $b \in \{0,1\}^{rd}$ and $\overline{b} \in \{0,1\}^{rd}$. Let there be an edge between $b$ and $\overline{b}$, if $\nrm{\rho_b - \rho_{\overline{b}}}_2 \leq \varepsilon/\sqrt{8r}$. From the previous lemma, we obtain that there exist absolute constants $c,C > 0$ such that the number of edges of $m$ in $G$ satisfies
	\[m \leq 2^{2rd} \cdot Ce^{-crd}.\]
	We abbreviate $f = Ce^{-crd}$, and observe that
	\[\sum_{b \in \{0,1\}^{rd}} \deg(b) = m \leq f \cdot 2^{2rd},\]
	where $\deg(b)$ denotes the degree of $b$ in $G$. Next, let $B = \{b \in \{0,1\}^{rd} : \deg(b) \geq 2^{rd}\sqrt{f}\}$, i.e., the set of nodes that have high degree. Then, by an argument that is usually referred to as the pigeonhole principle, we obtain that $|B| \leq 2^{rd}\sqrt{f}$.
	
	Let $b \in \{0,1\}^{rd} \setminus B$, and suppose that we can access to $b$ through the phase oracle
	\[O_{b,\varepsilon'} : \ket{j} \mapsto e^{i\varepsilon b_j}\ket{j}.\]
	We now use our $Q$-query mixed-state tomography algorithm to construct a new algorithm that recovers $b$ with some very low probability.
	
	The first step is to implement the unitary $U_b$ that maps 
	\[U_b : \ket{0} \mapsto \frac{1}{\sqrt{r}} \sum_{j=0}^{r-1} \ket{\psi_b^{(j)}}\ket{j}.\]
	Using the same construction as in \autoref{lem:lb-ell1-inverse}, we can construct a circuit implementing this unitary $U_b$ with $K$ calls to $O_{b,\varepsilon}$, where $K = \Theta(1)$. Next, since this unitary $U_b$ prepares a purification of $\rho_b$, we can use $Q$ queries to it to obtain an estimate $\widetilde{\rho}$, such that $\nrm{\widetilde{\rho} - \rho_b}_2 \leq \varepsilon/(2\sqrt{8r})$, with probability at least $2/3$.
	
	Next, suppose that $\overline{b} \in \{0,1\}^{rd}$ satisfies $\nrm{\widetilde{\rho} - \rho_{\overline{b}}}_2 \leq \varepsilon/(2\sqrt{8r})$. Then, by the triangle inequality, we have that $\nrm{\rho_b - \rho_{\overline{b}}}_2 \leq \nrm{\rho_b - \widetilde{\rho}}_2 + \nrm{\widetilde{\rho} - \rho_{\overline{b}}}_2 \leq \varepsilon/\sqrt{8r}$, and hence we find that $b$ and $\overline{b}$ are neighbors in $G$. Since we chose $b$ to be in $\{0,1\}^{rd} \setminus B$, we know that $\deg(b) \leq 2^{rd}\sqrt{f}$, and hence there are at most $2^{rd}\sqrt{f}$ choices for $\overline{b}$, among which is $b$ itself. Thus, if we uniformly choose one such $\overline{b}$, it will be equal to $b$ with probability at least $2/3 \cdot 2^{-rd}f^{-1/2}$.
	
	The procedure above uses $KQ$ queries to $O_{b,\varepsilon'}$, and recovers $b$ with probability at least $2/3 \cdot 2^{-rd}f^{-1/2}$. It is known that if we can solve this problem with $KQ$ queries to the fractional phase oracle $O_{b,\varepsilon'}$, we can also solve it with at most $K'KQ$ queries to the regular phase oracle $O_b$, with $K' = \Theta(\varepsilon)$.\footnote{See the footnote in \autoref{lem:lb-ell1-inverse} for more details.} According to \cite{farhi1999bounds},~Equation~4, this implies that
	\[2^{rd}-|B| \leq \frac32 \cdot 2^{rd}\sqrt{f} \cdot \sum_{k=0}^{K'KQ} \binom{rd}{k} \leq \frac32 \cdot 2^{rd}\sqrt{f} \cdot 2^{rdH\left(\frac{K'KQ}{rd}\right)},\]
	where $H(x) = -x\log(x) - (1-x)\log(1-x)$ is the binary entropy function, and the rightmost inequality can be found in several text books, e.g., \cite{flum2006parametrized},~Lemma~16.19.
	
	We can now plug everything into the above equation. Since $|B| \leq 2^{rd}\sqrt{f}$, in particular it is smaller than $2^{rd}/2$ for big enough $d$, and hence we write
	\[2^{rd-1} \leq 2^{\log(3)-1 + rd + \log(C) - crd\log(e) + rdH\left(\frac{K'KQ}{rd}\right)}.\]
	Dropping the powers of $2$ leaves us with
	\[\log(3) + \log(C) + rd\left(-c\log(e) + H\left(\frac{K'KQ}{rd}\right)\right) \geq 0,\]
	and thus $H(\frac{K'KQ}{rd}) = \Omega(1)$. Since the binary entropy function is monotonously increasing from $0$ to $1$ in the interval $[0,1/2]$, we find that $K'KQ = \Omega(rd)$, and hence $Q = \Omega(rd/\varepsilon)$.
\end{proof}

We now summarize the known lower bound results on mixed-state tomography with access to a state-preparation unitary.

\begin{theorem}
	Let $d \in \N$, $1 \leq r \leq d$ and $0 < \varepsilon = o(1/(dr))$. Let $\ket{\psi}$ be a purification of a density matrix $\rho \in \C^{d \times d}$ or rank at most $r$. Suppose that we have access to a unitary that prepares $\ket{\psi}$, and its inverse. Then, we have the following lower bounds on mixed-state tomography.
	\begin{enumerate}
		\item In order to obtain an estimate $\widetilde{\rho} \in \C^{d \times d}$ such that $\nrm{\widetilde{\rho} - \rho}_1 \leq \varepsilon$, we must call the state-preparation unitary at least $\Omega(\max\{d\sqrt{r}/\varepsilon,dr/\log(dr)\})$ times.
		\item If $\varepsilon = o(1/(dr))$, then in order to obtain an estimate $\widetilde{\rho} \in \C^{d \times d}$ such that $\nrm{\widetilde{\rho} - \rho}_2 \leq \varepsilon$, we must call the state-preparation unitary at least $\Omega(d\sqrt{r}/\varepsilon)$ times.
	\end{enumerate}
\end{theorem}

\begin{proof}
It is shown in \cite{haah2017OptTomography} that $\Omega\trm{dr}$ rank-$r$ states exist that are at least a constant trace-distance away from each other. It then follows from an information theoretical argument (as we can learn at most $\log(dr)$ bits from a state-preparation unitary) that at least $\Omega\trm{dr/\log(dr)}$ queries are needed when $\eps = \Theta(1)$.

	The other lower bounds are \autoref{thm:lb-frob-norm-tomo}.
\end{proof}

The proofs of the two lower bounds in the first claim of the above theorem are somewhat different in nature, but nevertheless we expect them to hold simultaneously, i.e., we expect that the right bound is $\Omega(dr/\varepsilon)$, which matches the complexity of the algorithm that we give. The lower bound we present in this section is, to the best of our knowledge, the first to combine both the dependence on $r$ and $1/\varepsilon$. It proves tightness of our mixed-state tomography algorithm for estimating the density matrix in Frobenius norm, albeit only in the low-error regime where $\varepsilon = o(1/(dr))$. We expect that the construction outlined in this section also suffices to prove a lower bound of $\Omega(dr/\varepsilon)$ for the trace norm case, but it seems to require a more involved analysis than the one presented here.

\section{Open problems}
We end the paper with a discussion on some open questions.

\paragraph{State preparation without an inverse.} In \autoref{subsec:cond-samples} we consider tomography using conditional samples, and in \autoref{sec:sampling-lb} we show that our upper bounds are optimal up to log factors. Conditional samples are directly inspired by controlled usage of a state-preparation unitary, without access to the inverse of this unitary. Such a state-preparation unitary is at least as powerful as conditional samples, and at most as powerful as state preparation with the inverse as well. 

Even in the two dimensional case of standard amplitude estimation, the best upper bound of $\tilde{O}(1/\eps^2)$ comes from conditional samples, while the best lower bound of $\Omega(1/\eps)$ also holds when the inverse is allowed. Hence the question of finding a quantum algorithm that requires $o(1/\eps^2)$ applications of a controlled state-preparation unitary to perform amplitude estimation, and that does not require access to the inverse of this unitary. We conjecture that the answer is negative, but we are not aware of any lower bound techniques that differentiate between normal and inverse usage of an input oracle. 

\paragraph{Vector estimate conversions.} The two lemmas in \autoref{s:normstuff} still leave some open questions. While \autoref{lem:linftol2} gives the relation between amplitude and probability estimates in general, it is unclear whether a similar relation holds for amplitudes of a purification and the associated density matrix. \autoref{lem:infty_norm_to_trace} gives a relation between the $\ell_2$-norm for amplitudes and the Schatten-$1$-norm (tracer norm) for the density matrix, does a similar relation hold for the $\ell_q$-norm and Schatten-$\frac{1}{1/2+1/q}$-norm?

Similarly, \autoref{lem:new-norm-conversion} shows how to obtain a $\ell_q$-norm estimate of a $\ell_s$-normalized vector using an $\ell_\infty$-norm estimate. It is still unclear whether an $\ell_p$ norm estimate can be used in a similar manner to obtain an $\ell_q$-norm estimate, when $p>q>s$.

\paragraph{Simple sample-based estimates for mixed states in other norms.} All single-copy sampling methods for pure-state tomography that we are aware of estimate the state directly in a Schatten $q$-norm, and then convert to the trace norm. In order to find the initial estimate, a set of random measurements is performed, and an optimization problem is solved to find a $\tilde{\rho}$ that matches best with the measurement statistics. Could a very efficient estimate in the max-norm possibly lead to a simpler algorithm?  In \autoref{apdx:samplemixed} we show how a probability distribution can be constructed that is proportional to the $d^2$ elements in the density operator, so an $\ell_2$-norm approximation of this distribution gives a Frobenius norm estimate of $\rho$.

An alternative approach could be to use a procedure inspired by shadow tomography to estimate all the $E_{ij}$ and obtain a max-norm estimate with $\bigOt{1/\eps^2}$ samples. If these estimates can be made symmetric and unbiased, then this would imply an operator norm estimate with $\bigOt{d/\eps^2}$ samples, a Frobenius norm estimate with $\bigOt{dr/\eps^2}$ samples, and a trace norm estimate with $\bigOt{dr^2/\eps^2}$ samples. This would matching the optimal bound by~\cite{gross2010,haah2017OptTomography} for single copy measurements. There is some hope for this, as recent shadow tomography results~\cite{Huang2020shaddow} require only $\bigOt{1/\eps^2}$ copies when the Frobenius norm of the measurements is constant. Furthermore, these methods are rather simple, and there is no post processing needed, unlike the result by~\cite{gross2010,haah2017OptTomography} that requires the solution of a convex optimization problem. The main problem to overcome is that the outputs from shadow tomography might not be independent.  

\paragraph{Time complexity of expectation value estimation.}
When we apply expectation value estimation to mixed-state tomography, we give a tailored implementation of the block-encoding of $\sum_i \lambda_i E_i$ in order to avoid a large subnormalization. In general however the block-encoding is sub-normalized by $\sum_i  |\lambda_i|\nrm{E_i}$ due to the usage of the LCU-lemma. The pre-amplification of this block encoding then requires a number of iterations which scales with $N = \sum_i \nrm{E_i}$. 

On the other hand, the set of operators of the form $E_i/\nrm{\sum_j E_j}$ could be turned into a POVM measurement, as their sum has operator norm at most $1$. Hence, by estimating all expectation values of this POVM with precision $\eps/\nrm{\sum_j E_j}$ by simply measuring, we would be able to learn all original expectation values with precision $\eps$, and a sample complexity dependent on $\nrm{\sum_j E_j}$ (as opposed to $\sum_j \nrm{E_j}$. Can our techniques be improved to also depend on $\nrm{\sum_j E_j}$? Or, more generally, is there a version of the LCU-lemma and pre-amplification that achieves this time complexity? Low~\cite{low2018HamSimNearlyOptSpecNorm} uses a technique that might be related to this in order to improve sparse block-encodings for matrices with bounded norm, and a general answer might improve the method by Low slightly.

\paragraph{Closing the gap for mixed-state tomography in trace-norm.} We conjecture that the correct query complexity of mixed-state tomography with trace-norm error is $\Theta\trm{\frac{dr}{\eps}}$, i.e., our upper bounds are tight up to logarithmic factors. Our lower bounds, however, only show that $\Omega\trm{\frac{d\sqrt{r}}{\eps}+dr /\log(dr)}$ queries are needed.

\section*{Acknowledgements}
We are grateful to Srinivasan Arunachalam and Ronald de Wolf for useful discussions. Joran van Apeldoorn is supported by the Dutch Research Council (NWO/OCW), as part of QSC (024.003.037) and by QuantumDelta NL. András Gilyén acknowledges funding provided by the EU's Horizon 2020 Marie Skłodowska-Curie program 891889-QuantOrder.
Giacomo Nannicini is partially supported by the Army Research Office under grant number W911NF-20-1-0014.

\bibliographystyle{alphaUrlePrint}
\bibliography{biblio,qc_gily}

\appendix

\section{Direct mixed-state tomography using copies} \label{apdx:samplemixed}
We show how to perform mixed-state tomography with $O(rd^2/\eps)$
copies of the state and a small amount of quantum power. Note that the
well-known algorithm consisting of performing measurements in random
bases already achieves this sample complexity, see the discussion in
\cite{gross2010,flammia2012quantum,kueng2017low}; this is
optimal for unentangled, non-adaptive algorithms
\cite{haah2017OptTomography}. Thus, the algorithm presented here does not
improve over the known upper bounds. We discuss it anyway for the
following reasons: first, the algorithm is much easier to analyze than
existing algorithms; second, the algorithm uses very similar
techniques to Section~\ref{s:classicaltomo} for pure states; third, it is likely easier to implement.
\begin{proposition}
Let $\rho = \sum_{k=1}^r p_k \ket{\psi_k}\bra{\psi_k}$ for some
orthonormal $\ket{\psi_k} = \sum_{j \in [d]} \alpha^{(k)}_{j}
\ket{j}$. There is a quantum algorithm that, given $O(r d^2 / \eps^2)$
copies of $\rho$ and the ability to perform unitary operations on
them, outputs $\tilde{\rho}$ such that $\nrm{\rho - \tilde{\rho}}_1
\le \eps$ with probability at least $2/3$. The algorithm is
non-adaptive and does not require entangled measurements between
copies of $\rho$.
\end{proposition}
\begin{proof}
Recall that $\rho$ is a $d \times d$ matrix with entries:
\begin{align*}
  \rho_{u, v} = \sum_{k=1}^r p_k \alpha^{(k)}_{u}
  (\alpha^{(k)}_{v})^{\dag}.
\end{align*}
To avoid cumbersome equations, it is easier to analyze the algorithm
by working with a purification $\ket{\rho} = \sum_{k=1}^r \sqrt{p_k}
\ket{\psi_k}_A \ket{\phi_k}_B$ of $\rho$, where $\ket{\phi_k}$ are
orthonormal; note that we never act on the purifying register, and the
purification is solely for convenience. Add one fresh qubit in state
$\ket{0}$ to the system; suppose it is the first. For $h \in [d]$,
apply a Hadamard on the first qubit, followed by the unitary
$\ket{0}\bra{0} \otimes I_A \otimes I_B + \ket{1}\bra{1} \otimes
\sum_{j \in [d]} \ket{(j-h) \mod d}\bra{j} \otimes I_B$, and finally
another Hadamard on the first qubit. In the following, for brevity we
write $j+h$ instead of $(j+h) \mod d$: we use this notation only to
index basis elements, so the context should avoid any ambiguity. The
larger system is now described by the following pure state:
\begin{align*}
  \frac{1}{2} \ket{0} \sum_{k=1}^r \sqrt{p_k} \left(\sum_{j \in [d]}
  (\alpha^{(k)}_{j} + \alpha^{(k)}_{j+h}) \ket{j} \ket{\phi_k}\right)
  + \frac{1}{2} \ket{1} \sum_{k=1}^r \sqrt{p_k} \left(\sum_{j \in [d]}
  (\alpha^{(k)}_{j}-\alpha^{(k)}_{j+h})\ket{j} \ket{\phi_k}\right).
\end{align*}
Next, we trace out the purifying register $B$, and compute the
probability of finding the first qubit in state $\ket{0}$ and system
$A$ in state $\ket{j}$:
\begin{align*}
  \frac{1}{4} \sum_{k=1}^r p_k (\alpha^{(k)}_{j} + \alpha^{(k)}_{j+h})
  (\alpha^{(k)}_{j} + \alpha^{(k)}_{j+h})^{\dag} &= \frac{1}{4}
  \sum_{k=1}^r p_k \left(|\alpha^{(k)}_{j}|^2 + 2\Re(\alpha^{(k)}_{j}
  (\alpha^{(k)}_{j+h})^{\dag}) + |\alpha^{(k)}_{j+h}|^2\right) \\ &=
  \frac{1}{4} \left(\rho_{j,j} + 2\Re(\rho_{j,j+h}) +
  \rho_{j+h,j+h}\right) = q^{(h)}_{0j}.
\end{align*}
Similarly, the probability of finding the first qubit in state
$\ket{1}$ and system $A$ in state $\ket{j}$ is:
\begin{align*}
  \frac{1}{4} \left(\rho_{j,j} - 2\Re(\rho_{j,j+h}) +
  \rho_{j+h,j+h}\right) = q^{(h)}_{1j}.
\end{align*}
By definition the vector $q^{(h)}$ represents a discrete probability
distribution. We can obtain an $\ell_2$-norm estimate
$\tilde{q}^{(h)}$ of $q^{(h)}$ with error $\bar{\eps}$ taking
$O(1/\bar{\eps}^2)$ samples, see
\cite{apeldoorn2021QProbOraclesMulitDimAmpEst}. Note that for $h=0$,
this immediately yields an estimate
$(\tilde{\rho}_{0,0},\dots,\tilde{\rho}_{d-1,d-1})$ of the diagonal of
$\rho$ with $\ell_2$-norm error at most $\bar{\eps}$. For $h \in [d]
\setminus \{0\}$, we can then compute an estimate
$\tilde{\rho}_{j,j+h}$ for the real part of $\rho_{j,j+h}$ as
$2(\tilde{q}^{(h)}_{0j} - \frac{1}{2}\tilde{\rho}_{j,j} -
\frac{1}{2}\tilde{\rho}_{j+h,j+h})$. For convenience, let us call $v$
the vector with entries $\rho_{j,j}$ for $j \in [d]$, $v^{(h)}$ the
vector with entries $\rho_{j+h,j+h}$, and similarly for $\tilde{v}$
and $\tilde{v}^{(h)}$. The total $\ell_2$-norm squared error for a set
of $d$ of these off-diagonal elements can be bounded as follows:
\begin{align*}
  \sum_{j \in [d]} (\tilde{\rho}_{j,j+h} - \rho_{j,j+h})^2 = 2\sum_{j
    \in [d]} \left((\tilde{q}^{(h)}_{0j} - \frac{1}{2}\tilde{\rho}_{j,j} -
  \frac{1}{2}\tilde{\rho}_{j+h,j+h}) - (q^{(h)}_{0j} - \frac{1}{2}\rho_{j,j} -
  \frac{1}{2}\rho_{j+h,j+h})\right)^2 = \\ 2 \nrm{(\tilde{q}^{(h)}_{0} -
    \frac{1}{2}\tilde{v} - \frac{1}{2}\tilde{v}^{(h)}) - (q^{(h)}_{0} - \frac{1}{2}v - \frac{1}{2}v^{(h)})}^2 \le
  2 \Big( \nrm{\tilde{q}^{(h)}_{0} - q^{(h)}_{0}}^2 + \frac{1}{4}\nrm{\tilde{v} -
    v}^2 + \frac{1}{4}\nrm{\tilde{v}^{(h)} - v^{(h)}}^2 +
  \\ \frac{1}{2}\nrm{\tilde{q}^{(h)}_{0} - q^{(h)}_{0}}\nrm{\tilde{v} - v} +
  \frac{1}{2}\nrm{\tilde{q}^{(h)}_{0} - q^{(h)}_{0}}\nrm{\tilde{v}^{(h)} -
    v^{(h)}} + \frac{1}{4}\nrm{\tilde{v} - v}\nrm{\tilde{v}^{(h)} - v^{(h)}}\Big)
  \le 6 \bar{\eps}^2,
\end{align*}
where we use Cauchy-Schwarz plus the fact that
$\nrm{\tilde{q}^{(h)}_{0} - q^{(h)}_{0}}$, $\nrm{\tilde{v} - v}$ and
$\nrm{\tilde{v}^{(h)} - v^{(h)}}$ are all $\le \bar{\eps}$. This
implies that we can get an $O(\bar{\eps})$-$\ell_2$-estimate of the
real part of $d$ elements of $\rho$ with $O(1/\bar{\eps}^2)$
samples. A similar approach, with the addition of a phase gate to
multiply all coefficients by $i$, allows us to retrieve the imaginary
part with the same complexity.

The above algorithm is repeated $d$ times, for $h \in [d]$. Combining
these $d$ estimates of $d$ coefficients each, setting $\bar{\eps} =
\eps/\sqrt{d}$, we obtain $\tilde{\rho}$ such that $\nrm{\tilde{\rho}
  - \rho}_F \le \eps$ taking $O(d^2/\eps^2)$ samples. To convert from
Frobenius norm to trace norm, using the fact that there are at most
$r$ nonzero eigenvalues by assumption, we need to decrease the error
$\bar{\eps}$ by a further factor $\sqrt{r}$. Then, this yields a
trace-norm estimate of $\rho$ with $O(rd^2/\bar{\varepsilon}^2)$
samples.
\end{proof}

\section{Implementing a QRAM}\label{apd:qram}
In this appendix we prove our claim that a $d$-qubit QRAM can be implemented with $\bigO{d}$ gates in $\bigO{\log(d)}$ depth. Although QRAM implementations have been discussed at length in the literature, e.g.~\cite{glm2008qram} and follow-up works, these discussions focus on the number of ``activated'' gates. While physically relevant in order to argue about error rates, from a complexity point of view there is no difference between an activated or non-activated gate. 

We expect that the results below appear in the literature, but we were unable to locate them and hence proof them for completeness. If the reader is aware of earlier works with the same results, we would be grateful if the could inform us so that we can update this section to give proper attribution. 

\begin{lemma}
  Let $d$ be a power of $2$. 
  There is a unitary, called indexed-CNOT-out (stylized iCNOTo), acting on $\log(d)+1+d$ qubits plus $2d-3$ ancillary qubits that can be implemented using $2d-2-2\log(d)$ CNOT gates and $4d-4$ Toffoli and X gates in $10\log(d)$ depth, and acts as follows on computational basis states
  \[
  \text{iCNOTo}\ket{i}\ket{b}\ket{q_1}\dots\ket{q_d} = \ket{i}\ket{b\oplus q_i}\ket{q_1}\dots\ket{q_d}.
  \]

  There is also a unitary, called indexed-CNOT-in (stylized iCNOTi), acting on the same amount of qubits, that can be implemented in the same depth and number of gates, acting as
  \[
  \text{iCNOTi}\ket{i}\ket{b}\ket{q_1}\dots\ket{q_d} = \ket{i}\ket{b}\ket{q_1}\dots\ket{q_{i-1}}\ket{q_{i}\oplus b}\ket{q_{i+1}}\dots\ket{q_d}.
  \]
\end{lemma}
\begin{proof}
  We first note that a FANOUT gate acting (for $a\in \01$) as
  \[
  \text{FANOUT} \ket{a}\ket{0^k} = \ket{a^{k}}
  \]
  can be build using $k-1$ CNOT gates in depth $\log(k)$. 

  We will implement the ICNOTo gate as a tournament bracket. In the first step, if $i$ is even then we first copy over all $q_j$ for even $j$ to a fresh layer of $d/2$ qubits. If $i$ is odd then we do this for the odd $j$. The information whether $i$ is even or odd is contained in its least significant bit, which, using a FANOUT to $d/2$ can be distributed to $d/2$ fresh qubits in depth $\log(d)-1$. Now, conditioned on the $k$th of these parity qubits either $q_{2k}$ or $q_{2k+1}$ is put in a fresh qubits, using $2$ Toffoli gates and $2$ X gates in depth $4$.

  We then do exactly the same circuit for the next layer, as if we were implementing a iCNOTo on $d/2$ qubits. After $\log(d)$ levels we end up with a (fixed) register in the state $\ket{q_i}$, and we can CNOT this value with $\ket{b}$. In fact, we can use $\ket{b}$ as the target for the final level, instead of a fresh qubit. After this we can uncompute all intermediate values using the same depth and gate count.

  For the depth, note that all FANOUT gates can be performed in parallel. The deepest has depth $\log(d)$. The tournament bracket has depth $4$ per layer, and $\log(d)$ layers. Including the uncompute the total depth is $10\log(d)$.

  As for the ancillary qubits, there are $d-1$ parity bits used, one for each decision in the tournament bracket. There are $d-2$ intermediate bits used in the tournament, as we use $b$ for the final result. Hence the circuit uses $2d-3$ ancillary qubits.

 The CNOT count of all the fan outs is $\sum_{i=1}^{\log(d)} \left(\frac{d}{2^i}-1\right) = d-1-\log(d)$. The tournament requires $2$ Toffoli gates per decision, of which there are $d-1$, so the Toffoli count of this part is $2d-2$ (and the X count is the same). The total, including uncomputation becomes $2d-2-2\log(d)$ CNOT gates, and $4d-4$ Toffoli and X gates.
  
The iCNOTi gate is implemented in almost the same way, but now $b$ is distributed from the top of the tournament to the leave corresponding to $q_i$. 
\end{proof}

There are two types of indexed SWAP that we may build. The first type has a fixed qubit that can be swaped with the $i$th qubit controlled on $i$. The second, most general indexed SWAP is controlled by both an $i$ and $j$ register and swaps the two. In the body of the paper we do not make this disintion, as there complexities are of the same order, but as the constant differ we will do so here.

\begin{lemma}
  Let $d$ be a power of $2$. 
  There is a unitary, called single-indexed-SWAP (stylized iSWAP\footnote{Note that this is not related to the iSWAP gate that applies the phase $i$ if qubits are swapped, sometimes discussed in the literature.}), acting on $\log(d)+1+d$ qubits plus $2d-3$ ancillary qubits that can be implemented using $2d-2-2\log(d)$ CNOT gates and $12d-12$ Toffoli and X gates in $26\log(d)$ depth, and acts as follows on computational basis states
  \[
  \text{iSWAP}\ket{i}\ket{b}\ket{q_1}\dots\ket{q_d} = \ket{i}\ket{q_i}\ket{q_1}\dots\ket{q_{i-1}}\ket{b}\ket{q_{i+1}}\dots\ket{q_d}.
  \]

  There is also a unitary, called double-indexed-SWAP (stylized iiSWAP), acting on $2\log(d)+d$ qubits plus $4d-5$ ancillary qubits, that can be implemented using $4d-4-4\log(d)$ CNOT gates and $36d-36$ Toffoli and X gates in $74\log(d)$ depth, and acts as follows on computational basis states
  \[
  \text{iiSWAP}\ket{i}\ket{j}\ket{q_1}\dots\ket{q_d} = \ket{i}\ket{j}\ket{q_1}\dots\ket{q_{i-1}}\ket{q_{j}}\ket{q_{i+1}}\dots\ket{q_{j-1}}\ket{q_{i}}\ket{q_{j+1}}\dots\ket{q_d}.
  \]  
\end{lemma}
\begin{proof}
For the iSWAP implementation we note that the SWAP gate can be implemented using $3$ CNOT gates. In particular we can use two calls to iCNOTo and a single call to $iCNOTi$. Note that we can reuse the parity information bits and do not need to repeat the FANOUT. 

For the iiSWAP, we note that we can perform a doubly indexed CNOT, i.e., a CNOT from qubit $i$ to qubit $j$, by first retrieving the bit in the $i$th position with a iCNOTo, then running iCNOTi with index $j$, and then erasing the recovered bit with another call to iCNOTo. We can reuse the $b$ bit for this.
Again, $3$ of these doubly indexed CNOTs are sufficient to implement a iiSWAP. We can again reuse the parity bits without redoing the FANOUT, but we have to implement these bits for both $i$ and $j$. The stated counts follow.
\end{proof}

\end{document}